\documentclass{CSML}
\pdfoutput=1

\usepackage{lastpage}
\newcommand{\dOi}{13(3:33)2017}
\lmcsheading{}{1--\pageref{LastPage}}{}{}%
{Feb.~08,~2017}{Sep.~26, 2017}{}


\ACMCCS{[{\bf Theory of computation}]: Formal languages and automata theory, Logic}
\usepackage{harengs}

\def\ie{{\em i.e.}}

\begin{document}

\title{Petri automata}%
\titlecomment{{\lsuper*} This work has been
  supported by the European Research Council (ERC) under the European
  Union’s Horizon 2020 programme (CoVeCe, grant agreement No 678157),
  as well as the the LABEX MILYON (ANR-10-LABX-0070) of Universit\'e
  de Lyon, within the program "Investissements d'Avenir"
  (ANR-11-IDEX-0007) operated by the French National Research Agency
  (ANR).}

\author[Paul Brunet]{Paul Brunet}	
\address{Univ Lyon, CNRS, ENS de Lyon, UCB Lyon 1, LIP, France}	
\email{paul@brunet-zamansky.fr} %optional  
\email{Damien.Pous@ens-lyon.fr}  %optional
%\thanks{thanks 1, optional.}

\author[Damien Pous]{Damien Pous}
\address{}
%\thanks{thanks 2, optional.}	%optional

\begin{abstract}
  \noindent
  Kleene algebra axioms are complete with respect to both language
  models and binary relation models. In particular, two regular
  expressions recognise the same language if and only if they are
  universally equivalent in the model of binary relations. 

  We consider Kleene allegories, i.e., Kleene algebras with two
  additional operations and a constant which are natural in binary
  relation models: intersection, converse, and the full
  relation. While regular languages are closed under those operations,
  the above characterisation breaks. Putting together a few results
  from the literature, we give a characterisation in terms of
  languages of directed and labelled graphs.

  By taking inspiration from Petri nets, we design a finite automata
  model, Petri automata, allowing to recognise such graphs.  We prove
  a Kleene theorem for this automata model: the sets of graphs
  recognisable by Petri automata are precisely the sets of graphs
  definable through the extended regular expressions we consider.

  Petri automata allow us to obtain decidability of identity-free
  relational Kleene lattices, i.e., the equational theory generated by
  binary relations on the signature of regular expressions with
  intersection, but where one forbids unit. This restriction is used
  to ensure that the corresponding graphs are acyclic. We actually
  show that this decision problem is \expspace-complete.
\end{abstract}

\maketitle

\section{Introduction}
\label{sec:intro}

Consider binary relations and the operations of union~($\sumtm$),
intersection~($\cap$), composition~($\cdot$),
converse~($\conv\argument$), transitive closure~($\argument^+$),
reflexive-transitive closure~($\argument^\star$), and the constants
identity~($1$), empty relation~($0$) and universal
relation~($\top$). These objects give rise to an (in)equational theory
over the following signature.
\begin{equation*}
  \Sigma\eqdef
  \tuple{\sumtm_2,\cap_2,\cdot_2,{\conv\argument}_1,{\argument^+}_1,{\argument^\star}_1,0_0,1_0,\top_0}
\end{equation*}
A pair of terms $e,f$ built from those operations and some variables
$a,b,\dots$ is a \emph{valid equation}, denoted $\Rel\models e=f$, if
the corresponding equality holds universally. Similarly, an inequation
$\Rel\models e\leq f$ is valid when the corresponding containment
holds universally. Here are valid equations and inequations: they hold
whatever the relations we assign to variables $a$, $b$, and $c$.
\begin{align}
  \label{eq:ka}
  &\Rel~\models~ (a\sumtm b)^\star\cdot b \cdot (a\sumtm b)^\star ~=~
  (a^\star\cdot b \cdot a^\star)^+\\
  \label{eq:kc}
  &\Rel~\models~ a^\star ~\leq~ 1\,\sumtm\, a\cdot\conv a\cdot a^+\\
  \label{eq:al}
  &\Rel~\models~ a\cdot b~\cap~ c ~\leq~ a\cdot(b~\cap~ \conv a\cdot c)\\
  \label{eq:kl}
  &\Rel~\models~ a^+~\cap~1 ~\leq~ (a\cdot a)^+\\
  \label{eq:top}
  &\Rel~\models~ \top\cdot a\cdot\top\cdot b\cdot\top ~=~ \top\cdot b\cdot\top\cdot a\cdot\top
\end{align}

\noindent
Various fragments of this theory have been studied in the literature:
\begin{itemize}
\item \emph{Kleene algebra}~\cite{Conway71}, where one removes
  intersection, converse, and $\top$, so that terms are plain regular
  expressions. The equational theory is decidable~\cite{Kleene}, and
  actually \pspace-complete~\cite{MS72}.  The equational theory is not
  finitely based~\cite{redko64}, but finite quasi-equational
  axiomatisations exist~\cite{Krob90,Kozen91,Boffa95}.
  Equation~\eqref{eq:ka} lies in this fragment, and one can notice
  that the two expressions recognise the same language.
\item \emph{Kleene algebra with converse}, where one only removes
  intersection and $\top$, is also a decidable
  fragment~\cite{BES95}. It remains
  \pspace~\cite{bp:ramics14:kac,bp:jlamp15:kac}, and can be
  axiomatised relatively to Kleene
  algebra~\cite{EB95}. Inequation~\eqref{eq:kc} belongs to this
  fragment.
\item (representable, distributive) \emph{allegories}~\cite{FS90},
  sometimes called positive relation algebras, where transitive and
  reflexive-transitive closures and the full relation are not
  allowed. They are decidable~\cite[page 208]{FS90}, and not finitely
  based.  Inequation~\eqref{eq:al} is known as the \emph{modularity
    law} in this setting.
\item (representable) \emph{identity-free Kleene
    lattice}~\cite{AMN11}, where converse is removed, as well as the
  full relation, the identity relations, and reflexive-transitive
  closure (transitive closure is kept: reflexive-transitive closure is
  removed just because it hides the identity constant, we have
  $0^\star=1$). 
\end{itemize}
\noindent
Accordingly we call the whole (in)equational theory that of
\emph{representable Kleene allegories}\footnote{An algebra is
  representable when it is isomorphic to an algebra of binary
  relation. The class of Kleene algebra is usually defined as a
  finitely based quasi-variety, and its equational theory coincides
  with that generated by binary relations (i.e., by representable
  algebras); in contrast, allegories are defined by a set of equations
  which is not complete with respect to representatble algebra: there
  are some allegories which are not representable~\cite{FS90};
  similarly for Kleene lattices.}.  
We obtain several important steps
towards decidability of this (in)equational theory.
\begin{enumerate}
\item we characterise it in terms of \emph{graph} languages;
\item we design a new automata model, called \emph{Petri automata}, to
  recognise such graph languages, and we prove a Kleene like theorem:
  the languages definable through Kleene allegory expressions are
  precisely those recognisable by Petri automata;
\item we prove that the fragment where converse, full relation, and
  identity are forbidden is \expspace-complete.
\end{enumerate}

\noindent
The latter fragment was studied by Andréka, Mikulás and
Németi~\cite{AMN11}. They proved that on the corresponding signature,
the equational theory generated by binary relations coincides with
that generated by formal languages:
\begin{equation*}
  \Rel\models e=f \quad \Leftrightarrow \quad \Lang\models e=f
\end{equation*}
(where $\Lang\models e=f$ stands for universal validity in all
$\Sigma$-algebras of formal languages with the standard interpretation
of the operations---in particular, $\conv\argument$ is language
reversal and $\top$ is the language of all words).  It is important to
understand that such a property highly depends on the considered
fragment:
\begin{itemize}
\item It does not holds when $1$, $\cap$ and $\cdot$ are all
  available: we have $\Lang\models (a\cap 1)\cdot b=b\cdot (a\cap 1)$
  but this equation is not valid in $\Rel$.
\item It does hold in Kleene algebra (\ie, on the usual signature of
  regular expressions).
\item It does not holds when both $\top$ and $\cdot$ are available: we
  have
  $\Rel\models\top\cdot a\cdot\top\cdot b\cdot\top = \top\cdot
  b\cdot\top\cdot a\cdot\top$ while this equation fails in $\Lang$.
\item It does not hold when both $\conv\argument$ and $\cdot$ are
  available: we have $\Rel\models a \leq a \cdot \conv a\cdot a$ while
  this inequation fails in $\Lang$. Ésik et al. actually proved that
  this is the only missing law for the signature of Kleene algebra
  with converse, and that the two equational theories are
  decidable~\cite{BES95,EB95}.
\end{itemize}

\subsection*{Note}
This paper is a follow-up to the paper we presented at
LiCS'15~\cite{bp:lics15:paka}. It contains more details and proofs,
but also a full Kleene theorem for Petri automata (we only have the
easiest half of it in~\cite{bp:lics15:paka}), a comparison of our
automata model with the branching automata of Lodaya and
Weil~\cite{LODAYA2000347}, \expspace-hardness of identity-free Kleene
lattices (we only have it for the underlying automata
problem in~\cite{bp:lics15:paka}), and the handling of the
constant~$\top$.

\bigskip\noindent
We continue this introductory section by an informal description of
the characterisation based on graph languages, and of our automata
model.

\subsection{Languages}

In the simple case of Kleene algebra, \ie, without converse and
intersection, the (in)equational theory generated by relations (or
languages) can be characterised by using regular languages. Write
$\Ln e$ for the language denoted by a regular expression $e$; for
all regular expressions $e,f$, we have
\begin{equation}
  \label{eq:charac:ka}
  \Rel\models e\leq f \text{ if and only if } \Ln e \subseteq \Ln f.
\end{equation}
(This result is easy and folklore; proving that this is also
equivalent to provability using Kleene algebra
axioms~\cite{Krob90,Kozen91,Boffa95} is much harder.)

While regular languages are closed under intersection, the above
characterisation does not scale. Indeed, consider two distinct
variables $a$ and $b$. The extended regular expressions $a\cap b$ and
$0$ both recognise the empty language, while
$\Rel\not\models a\cap b=0$: one can interpret $a$ and $b$ by
relations with a non-empty intersection.

\subsection{Graphs}

Freyd and Scedrov' decision procedure for representable
allegories~\cite[page 208]{FS90} relies on a notion of directed,
labelled, 2-pointed graph. The same notion was proposed independently
by Andréka and Bredikhin~\cite{AB95}, in a more comprehensive way. 

Call \emph{terms} the terms built out of composition, intersection,
converse, constants $1$ and $\top$, and variables $a,b\dots$. A term
$u$ can be represented as a labelled directed graph $\Gr u$ with two
distinguished vertices called the \emph{input} and the
\emph{output}. We give some examples in
Figure~\ref{fig:ground:graphs}, see Section~\ref{sec:graphs} for a
precise definition.

\begin{figure}
  \noindent
  \begin{minipage}[b]{.55\linewidth}
  \begin{tabular}{m{.41\columnwidth}m{.5\columnwidth}}
   \hfill$\Gr a$ &
   \begin{tikzpicture}
     \position (0) (0,0);
     \position (1) (1.5,0);
     \initst (0);
     \fnst (1);
     \edge[above] (0) (1)[a];           
   \end{tikzpicture}
   \\
   \hfill$\Gr{\paren{a\cdot\paren{b\cap c}}\cap d}$ &  
   \begin{tikzpicture}
     \position (0) (0,-0.20);
     \position (1) (1, 0.20);
     \position (2) (2,-0.20);
     \initst (0); \fnst (2);
     \edge[above,out=55,in=180] (0) (1)[a]; 
     \edge[below,out=-25,in=-155] (0) (2)[d]; 
     \edge[above,out=0,in=135] (1) (2)[b]; 
     \draw[arc] (1) 
     to[below,out=-65,in=180] node[near start]{$c$} (2);
   \end{tikzpicture}
   \\
   \hfill$\Gr{\paren{a\cdot b}\cap\paren{a\cdot c}}$ &
   \begin{tikzpicture}
     \position (0) (0,0);
     \position (1) (1,0.25);
     \position (2) (1,-0.25);
     \position (3) (2,0);
     \initst (0); \fnst (3);
     \edge[above,out=25,in=-175] (0) (1)[a]; 
     \edge[below,out=-25,in=175] (0) (2)[a]; 
     \edge[above,out=-5,in=160] (1) (3)[b]; 
     \edge[below,out=5,in=-160] (2) (3)[c]; 
   \end{tikzpicture}
   \\
   \hfill$\Gr{\paren{a\;\cap\; b\cdot c}\cdot d}$ &
   \begin{tikzpicture}
     \position (0) (0,0);
     \position (2) (.8,-0.25);
     \position (3) (1.6,0);
     \position (4) (2.4,0);
     \initst (0); \fnst (4);
     \edge[above,out=20,in=160] (0) (3)[a]; 
     \edge[below,out=-25,in=175] (0) (2)[b]; 
     \edge[below,out=5,in=-160] (2) (3)[c]; 
     \edge[above] (3) (4)[d]; 
   \end{tikzpicture}
   \\
   \hfill$\Gr{a\cdot\conv a\cdot a}$ &
   \begin{tikzpicture}
     \position (0) (0,0);
     \position (2) (.8,0);
     \position (3) (1.6,0);
     \position (4) (2.4,0);
     \initst (0); \fnst (4);
     \edge (0) (2)[a]; 
     \edge (3) (2)[a]; 
     \edge (3) (4)[a]; 
   \end{tikzpicture}
   \\
  \end{tabular}
  \caption{Graphs associated with some terms.}
  \label{fig:ground:graphs}
  \end{minipage}
  \hfill
  \begin{minipage}[b]{.4\linewidth}
  \begin{tikzpicture}
   \node at (-1.5,2.05) {\normalsize$G$~:};
   \node at (-1.6,1.05) {$\blacktriangle$};
   \node (x) at (0,2) {4} ;
   \node (y) at (1.5,2.5) {5} ;
   \node (z) at (3,2) {6} ;
   \initst (x); \fnst (z);
   \edge[above,out=45,in=180] (x) (y)[a]; 
   \edge[below,out=-25,in=-155] (x) (z)[d]; 
   \edge[above,out=0,in=145] (y) (z)[b]; 
   \draw[arc] (y) 
   to[below,out=-45,in=180] node[near start]{$c$} (z);

   \node at (-1.5,0.05) {\normalsize$F$~:};
   \node (0) at (0,0) {0} ;
   \node (1) at (1.5,0.4) {1} ;
   \node (2) at (1.5,-0.4) {2} ;
   \node (3) at (3,0) {3} ;
   \initst (0); \fnst (3);
   \edge[above,out=35,in=180] (0) (1)[a]; 
   \edge[below,out=-35,in=180] (0) (2)[a]; 
   \edge[above,out=0,in=150] (1) (3)[b]; 
   \edge[below,out=0,in=-150] (2) (3)[c]; 
   
   \draw[arc,color=red,style=dotted] (0) to[bend left] (x);
   \draw[arc,color=red,style=dotted] (1) to[bend left] (y);
   \draw[arc,color=red,style=dotted] (2) to[bend right] (y);
   \draw[arc,color=red,style=dotted] (3) to[bend right] (z);
  \end{tikzpicture}
  \caption{A graph homomorphism.}
  \label{fig:GinfF}
  \end{minipage}
\end{figure}

This algebra of two-pointed graphs can be used---and was proposed by
Freyd, Scedrov, Andréka and Bredikhin, for the full syntax of
allegories. All results from this section extend to such a setting, in
particular characterisation~(\ref{eq:charac:kal}) below. In the
present case, where we forbid converse and the constants identity and
full relation, an important property is that the considered graphs are
acyclic, and in fact, series-parallel.

Graphs can be endowed with a preorder relation: we write $G\lessgr F$
when there exists a graph homomorphism from $F$ to $G$ preserving
labels, inputs, and outputs. For instance the graph corresponding to
$\paren{a\cdot\paren{b\cap c}}\cap d$ is less than the graph of
$\paren{a\cdot b}\cap\paren{a\cdot c}$, thanks to the homomorphism
depicted in Figure~\ref{fig:GinfF} using dotted arrows. Notice that
the homomorphism needs not be injective or surjective, so that this
preorder has nothing to do with the respective sizes of the graphs: a
graph may very well be smaller than another in the sense of $\lessgr$,
while having more vertices or edges (and vice versa).

The key result from Freyd and Scedrov~\cite[page 208]{FS90}, or
Andréka and Bredikhin~\cite[Theorem~1]{AB95}, is that for all
terms $u,v$, we have
\begin{equation}
  \label{eq:charac:al}
  \Rel\models u\leq v \text{ if and only if } \Gr u\lessgr\Gr v.
\end{equation}
The graphs are finite so that one can search exhaustively for a
homomorphism, whence the decidability result. 

\subsection{Graph languages}

To extend the above graph-theoretical characterisation to
identity-free Kleene lattices, we need to handle union, zero, and
transitive closure. It suffices for that to consider sets of graphs:
to each expression $e$, we associate a set of graphs $\G(e)$. This set
is infinite whenever the expression $e$ contains transitive closures.

Writing $\clgr X$ for the downward closure of a set of graphs $X$ by
the preorder $\lessgr$ on graphs, we obtain the following
generalisation of both~\eqref{eq:charac:ka} and~\eqref{eq:charac:al}:
for all expressions $e,f$,
\begin{equation}
  \label{eq:charac:kal}
  \Rel\models e\leq f \text{ if and only if } \clgr{\G(e)}\subseteq \clgr{\G(f)}.
\end{equation}
This is Theorem~\ref{thm:interlang} in the sequel, and this result is
almost there in the work by Andréka et al.~\cite{AB95,AMN11}. To the
best of our knowledge this explicit formulation is new, as well as its
use towards decidability results.

When $e$ and $f$ are terms, we recover the
characterisation~\eqref{eq:charac:al} for representable allegories:
$\G(e)$ and $\G(f)$ are singleton sets in this case. For plain regular
expressions, the graphs are just words and the preorder~$\lessgr$
reduces to isomorphism. We thus recover the
characterisation~\eqref{eq:charac:ka} for Kleene algebra. This result
also generalises the characterisation provided by Ésik et
al.~\cite{BES95} for Kleene algebra with converse: graphs of
expressions without intersection are just words over a duplicated
alphabet, and the corresponding restriction of the preorder~$\lessgr$
precisely corresponds to the word rewriting system they use.

\subsection{Petri automata}

In order to exploit the above characterisation and obtain decidability
results, one has first to represent graph languages in a finitary way.
We propose for that a new finite automata model, largely based on
Petri nets~\cite{Petri,Petri62,Murata}.  We describe this model below,
ignoring converse, $1$, and $\top$ for the sake of clarity.

Recall that a Petri net consists of
\begin{itemize}
\item a finite set of \emph{places}, denoted with circles; 
\item a set of \emph{transitions}, denoted with rectangles;
\item for each transition, a set of input places and a set of output
  places, denoted with arrows;
\item an \emph{initial place}, denoted by an entrant arrow;
 \item a set of \emph{final markings}.% , denoted by dotted boxes (a
%   marking, or configuration, being a set of places).
\end{itemize}
The execution model is the following: start by putting a token on the
initial place; choose a transition whose input places all contain a
token, remove those tokens and put new tokens in the output places of
the transition; repeat this process until a final marking is
reached. The obtained sequence of transitions is called an
\emph{accepting run}. (We actually restrict ourselves to \emph{safe}
Petri nets, to ensure that there is always at most one token in a
given place when playing this game.)

A \emph{Petri automaton} is just a safe Petri net with variables
labelling the outputs of each transition, and with a single final
marking consisting in the empty set of tokens. The automaton depicted
below is the automaton we will construct for the term
$a\cdot b~\cap ~a\cdot c$. A run must start by firing the left-most
transition, reaching the marking $\set{B,C}$; then we have the choice
of firing the upper transition first, reaching the marking
$\set{D,C}$, or the lower one, reaching the marking $\set{B,E}$. In
both cases we reach the marking $\set{D,E}$ by firing the remaining
transition. We complete the run by firing the last transition, which
leads to the final marking $\emptyset$.
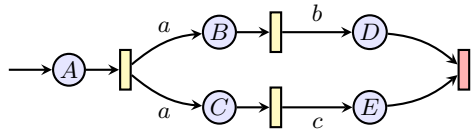
\begin{center}
   \begin{tikzpicture}
    %\draw[thick,dotted,fill=red!10] (3.6,-0.9) rectangle (4.4,0.9);
    \state[A](A)(0,0);
    \state[B](B)(2,0.5);
    \state[C](C)(2,-0.5);
    \state[D](D)(4,0.5);
    \state[E](E)(4,-0.5);
    \trans(0)(0.75,0);
    \trans(1)(2.75,0.5);
    \trans(2)(2.75,-0.5);
    \transf(fn)(5.25,0);
    \initst (A);
    \edge (A) (0) ; 
    \edge[above,out=40,in=-175] (0) (B)[a]; 
    \edge[below,out=-40,in=175] (0) (C)[a]; 
    \edge (B) (1) ; 
    \edge[above] (1) (D)[b]; 
    \edge (C) (2) ; 
    \edge[below] (2) (E)[c]; 
    \edge[out=-5,in=140] (D) (fn);
    \edge[out=5,in=-140] (E) (fn);
  \end{tikzpicture}
\end{center}
To read a graph in such an automaton, we try to find an accepting run
that matches the graph up to homomorphism (Definitions~\ref{def:read}
and~\ref{def:lang}). We do that by using a sequence of functions from
the successive configurations of the run to the vertices of the
graph. We start with the function mapping the unique token in the
initial place, to the input vertex of the graph. To fire a transition,
we must check that all its input tokens are mapped to the same vertex
in the graph, and that this vertex has several outgoing edges,
labelled according to the outputs of the transition. If this is the
case, we update the function by removing the mappings corresponding to
the deleted tokens, and by adding new mappings for each of the created
tokens (using the target vertices of the aforementioned outgoing
edges, according to the labels). If a transition has no output, before
firing it we additionally require that each of its input tokens are
mapped to the output vertex of the graph. The graph is accepted by the
Petri automaton if we can reach the final marking $\emptyset$.

For instance, the previous automaton accepts the graph of
$a\cdot b~\cap ~a\cdot c$ ($F$ in Figure~\ref{fig:GinfF}).  We start
with the function $\set{A\mapsto 0}$. We can fire the first
transition, updating the function into $\set{B\mapsto 1,\,C\mapsto 2}$
(We could also choose to update the function into
$\set{B\mapsto 2,\,C\mapsto 1}$, or $\set{B,C\mapsto 1}$, or
$\set{B,C\mapsto 2}$ but this would lead to a dead-end). Then we can
fire the upper transition, evolving the function into
$\set{D\mapsto 3,\,C\mapsto 2}$, the lower transition leading to
$\set{D,E\mapsto 3}$ and we finish by firing the remaining transition,
thus getting the function with domain $\emptyset$.

We call \emph{language} of $\A$ the set of graphs $\Ln\A$ accepted by
a Petri automaton $\A$. This language is downward-closed:
$\Ln\A=\clgr{\Ln\A}$.
For instance, the previous automaton also accepts the graph $G$ from
Figure~\ref{fig:GinfF}, which is smaller than $F$. Indeed, when we
fire the first transition, we can associate the two newly created
tokens (in places $B$ and $C$) with the same vertex $(5)$. This actually
corresponds to composing the functions used to accept $F$ with the
homomorphism depicted with dotted arrows.

\medskip

This automata model is expressive enough for Kleene allegories: for
every expression $e$, we can construct a Petri automaton $\A(e)$ such
that $\Ln{\A(e)}=\clgr{\G(e)}$ (Sections~\ref{sec:exp:pa}
and~\ref{sec:reading}). We give three other examples of Petri automata
to give more intuition on their behaviour.

\medskip

The first transition in the previous Petri automaton splits the
initial token into two tokens, which are moved concurrently in the
remainder of the run. This corresponds to an intersection in the
considered expression.
% , and to an node with several outgoing arrows in
% the corresponding graph.
%
This is to be contrasted with the behaviour of the following
automaton, which we would construct for the expression
$a\cdot b\,\sumtm\,a\cdot c$. This automaton has two accepting runs:
$\set{A},\set{B},\set{D},\emptyset$ and
$\set{A},\set{C},\set{E},\emptyset$, which can be used to accept the
(graphs of the) terms $a\cdot b$ and $a\cdot c$.
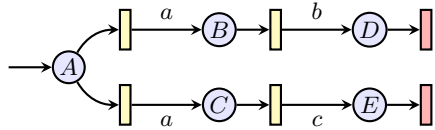
\begin{center}
   \begin{tikzpicture}
    % \draw[thick,dotted,fill=red!10] (3.6,0.1) rectangle (4.4,0.9);
    % \draw[thick,dotted,fill=red!10] (3.6,-0.9) rectangle (4.4,-0.1);
    \state[A](A)(0,0);
    \state[B](B)(2,0.5);
    \state[C](C)(2,-0.5);
    \state[D](D)(4,0.5);
    \state[E](E)(4,-0.5);
    \trans(3)(0.75,-0.5);
    \trans(1)(2.75,0.5);
    \trans(2)(2.75,-0.5);
    \trans(0)(0.75,0.5);
    \transf(f1)(4.75,.5);
    \transf(f2)(4.75,-.5);
    \initst (A);
    \edge[out=60,in=-180] (A) (0) ; 
    \edge[out=-60,in=180] (A) (3) ; 
    \edge[above] (0) (B)[a]; 
    \edge[below] (3) (C)[a]; 
    \edge (B) (1) ; 
    \edge[above] (1) (D)[b]; 
    \edge (C) (2) ; 
    \edge[below] (2) (E)[c]; 
    \edge(D)(f1);
    \edge(E)(f2);
  \end{tikzpicture}
\end{center}
In a sense, two transitions competing for the same tokens represent a
non-deterministic choice, \ie, a union in an expression.

Still in the first example, the two tokens created by the first
transition are later collected in the final transition. Tokens may
also be collected and merged by a non-final transition. Consider for
instance the following automaton for $(a\,\cap\, b\cdot c)\cdot d$. It
has only one accepting run,
$\set{A},\set{B,C},\set{B,D},\set E,\emptyset$, and this run can be
used to read the fourth graph from Figure~\ref{fig:ground:graphs}.
\begin{center}
   \begin{tikzpicture}
    %\draw[thick,dotted,fill=red!10] (5.6,-0.4) rectangle (6.4,0.4);
    \state[A](A)(0,0);
    \state[B](B)(3,0.5);
    \state[C](C)(2,-0.5);
    \state[D](D)(4,-0.5);
    \state[E](E)(6,0);
    \trans(0)(0.75,0);
    \trans(1)(4.75,0);
    \trans(2)(2.75,-0.5);
    \transf(f)(6.75,0);
    \initst (A);
    \edge (A) (0) ; 
    \edge[above,out=40,in=-180] (0) (B)[a]; 
    \edge[below,out=-45,in=180] (0) (C)[b]; 
    \edge[out=0,in=135] (B) (1) ; 
    \edge[out=20,in=-135] (D) (1) ; 
    \edge (C) (2) ; 
    \edge[below] (2) (D)[c]; 
    \edge[above] (1) (E)[d]; 
    \edge(E)(f);
  \end{tikzpicture}
\end{center}

As a last example, consider the following automaton for the expression
$a\;\cap\; b^+\cdot c$. The upper transition introduces a loop, so
that there are infinitely many accepting runs. For all $n>0$, the
graph of the term $a\;\cap\; b^n \cdot c$ is accepted by this
automaton.
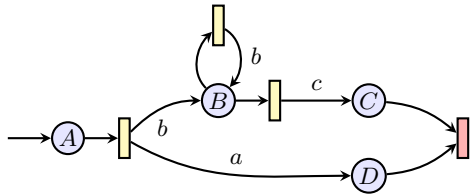
\begin{center}
   \begin{tikzpicture}
    % \draw[thick,dotted,fill=red!10] (3.6,-1.4) rectangle (4.4,0.4);
    \state[A](A)(0,-0.5);
    \state[B](B)(2,0);
    \state[C](C)(4,0);
    \state[D](D)(4,-1);
    \trans(0)(0.75,-0.5);
    \trans(1)(2,1);
    \trans(2)(2.75,0);
    \transf(f)(5.25,-.5);
    \initst (A);
    \edge (A) (0) ; 
    \edge[below,out=45,in=180] (0) (B)[b]; 
    \edge[above,out=-30,in=180] (0) (D)[a]; 
    \edge[out=135,in=-135] (B) (1) ; 
    \edge (B) (2) ; 
    \edge[right,out=-30,in=50] (1) (B)[b]; 
    \edge[above] (2) (C)[c]; 
    \edge[out=-5,in=140](C)(f);
    \edge[out=5,in=-140](D)(f);
  \end{tikzpicture}
\end{center}

%\section{Outline}
\clearpage
\tableofcontents

In Section~\ref{sec:expressions} we define the various sets of
expressions we consider in the paper as well as a few auxiliary
syntactic functions.
We define graphs in Section~\ref{sec:graphs}, and we use them to
characterise the equational theory of representable Kleene allegories.
Series-parallel graphs play a central role, they are defined and
related to fragments of expressions in Section~\ref{sec:sp}.

Petri automata and the sets of graphs they produce are defined in
Section~\ref{sec:pa}. The next two sections are devoted to the Kleene
theorem for Petri automata: in Section~\ref{sec:exp:pa} we translate
expressions into automata, and we obtain the converse translation in
Section~\ref{sec:pa:exp}. The latter translation is technically
involved and requires a large number of preliminary definitions. The
results and definitions in this long section are not required in the
sequel.

Section~\ref{sec:reading} is devoted to the reading of graphs in Petri
automata modulo homomorphism, in a local and incremental way, as
illustrated in the Introduction. This leads us to
Section~\ref{sec:compare}, where we show how to compare Petri automata
modulo homomorphism, using an appropriate notion of simulation.
We provide complexity bounds in Section~\ref{sec:complexity}.

We relate the notion of graph produced by a Petri automaton to the
standard notion of \emph{pomset-trace} of a Petri net in
Section~\ref{sec:relat-with-stand}. We discuss other works on regular
sets of graphs and we compare Petri automata with Lodaya and Weil's
\emph{branching automata} in Section~\ref{sec:relat-with-branch}.

\section{Terms and expressions}
\label{sec:expressions}

We let $a,b$ range over a set $X$ of \emph{variables}. We consider
the following syntax of \emph{expressions}:
\begin{align*}
  e,f\Coloneqq
  e \sumtm f\Mid e \cdot f\Mid e\cap f \Mid 
  e^+ \Mid \conv e \Mid
  0 \Mid 1 \Mid \top \Mid
  a\tag{$a\in X$}
\end{align*}
%We set $e^\star\eqdef 1\sumtm e^+$.
We denote their set by $\Exp$. We write $\Trm$ for the subset of
\emph{terms}: those expressions that do not contain $0$, $\sumtm$ or
$\argument^+$.
\begin{align*}
  u,v\Coloneqq
  u \cdot v\Mid u\cap v \Mid \conv u \Mid 1 \Mid \top \Mid
  a\tag{$a\in X$}
\end{align*}

If $\sigma: X\rightarrow\Rel[S]$ is an interpretation of the variables
into some algebra of relations, we write $\hat\sigma$ for the unique
homomorphism extending $\sigma$ into a function from $\Exp$ to
$\Rel[S]$.
\begin{defi}[Validity of an (in)equation]
  An inequation between two expressions $e$ and $f$ is \emph{valid},
  written $\Rel\models e\leq f$, if for every relational
  interpretation $\sigma$ we have
  $\hat\sigma(e)\subseteq\hat\sigma(f)$. Similarly, we write
  $\Rel\models e= f$, if for every relational interpretation $\sigma$
  we have $\hat\sigma(e)=\hat\sigma(f)$.
\end{defi}

\noindent Expressions actually denote sets of terms.
\begin{defi}[Terms of an expression]
  The \emph{set of terms} of an expression $e\in\Exp$, written
  $\tms{e}$, is the set of terms defined inductively as
  follows.
  \begin{align*}
    \tms{e\cdot f}&\eqdef\setcompr{u\cdot v}{u\in\tms e\text{ and } v\in\tms f}&
    \tms 1&\eqdef \set 1\\
    \tms{e\cap f}&\eqdef\setcompr{u\cap v}{u\in\tms e\text{ and } v\in\tms f}&
    \tms \top&\eqdef \set \top\\
    \tms{\conv e}&\eqdef\setcompr{\conv u}{u\in\tms e}\\
    \tms{e\sumtm f}&\eqdef\tms e \sumtm \tms f&
    \tms 0&\eqdef \emptyset\\
    \tms{e^+}&\eqdef\textstyle{ \bigcup_{n>0}\setcompr{u_1\cdot\cdots\cdot u_n}{\forall i, u_i\in \tms e}}&
    \tms a &\eqdef\set{a}
  \end{align*}
\end{defi}

\noindent As observed by Andréka, Mikulás, and Németi the
(relational) semantics of an expression only depends on the above set
of terms.
\begin{lem}[{\cite[Lemma 2.1]{AMN11}}]
  \label{lem:AMN2+} 
  For every expression $e\in\Exp$, every set $S$ and every relational
  interpretation $\sigma:X\rightarrow \Rel[S]$, we have
  \begin{equation*}
    \hat\sigma(e)=\bigcup_{u\in\tms e}\hat\sigma(u).
  \end{equation*}
\end{lem}

\noindent We will consider the following subsets of terms and
expressions when we focus on identity-free Kleene
lattices. \emph{Simple terms} (resp.\ \emph{simple expressions}) are
those terms (resp.\ expressions) that do not contain converse, $1$ or
$\top$; their set is denoted by $\STrm[X]$ (resp.\ $\SExp[X]$).

Let $\Xb\eqdef X\cup\setcompr{a'}{a\in X}\cup\set{1,\top}$. The
following function $\detype\argument$ on the left-hand side associates
a simple expression over $\Xb$ to every expression over $X$. It is
defined recursively together with the auxiliary function on the
right-hand side, it basically consists in pushing converse to the
leaves.
\begin{align*}
  \detype\argument\colon \Exp &\to \SExp &
  \detype\argument'\colon \Exp &\to \SExp \\
  e\cdot f &\mapsto \detype e\cdot \detype f &
  e\cdot f &\mapsto \detype f'\cdot \detype e' \\
  e\cap f &\mapsto \detype e\cap \detype f &
  e\cap f &\mapsto \detype e'\cap \detype f' \\
  e\sumtm f &\mapsto \detype e\sumtm \detype f &
  e\sumtm f &\mapsto \detype e'\sumtm \detype f' \\
  e^+ &\mapsto \detype e^+ &
  e^+ &\mapsto {\detype e'}^+ \\
  \conv e &\mapsto \detype e' &
  \conv e &\mapsto \detype e \\
  1 &\mapsto 1\in\Xb &
  1 &\mapsto 1\in\Xb \\
  \top &\mapsto \top\in\Xb &
  \top &\mapsto \top\in\Xb \\
  0 &\mapsto 0 &
  0 &\mapsto 0 \\
  a\in X &\mapsto a\in\Xb &
  a\in X &\mapsto a'\in\Xb
\end{align*}
Conversely, we write $\retype\argument\colon\SExp\to\Exp$ for the unique
homomorphism such that $\retype a=a$, $\retype {a'}=\conv a$,
$\retype 1=1$, and $\retype \top=\top$. It is straightforward to check
that for every simple expression $e\in\SExp$, we have
$\detype{\retype e}=e$. In the other direction, the equality is not
syntactic: for every expression $e\in\Exp$, we have
$\Rel\models\retype{\detype e}=e$. (This equation also holds in
language algebras, where converse also commutes with the other
operations in the appropriate way.) Those two functions naturally
restrict to terms and simple terms. 

The diagram below summarises the sets and functions we defined so far.
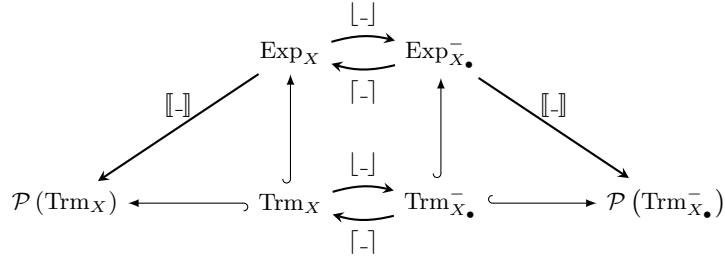
\begin{align*}
  \begin{tikzpicture}
    \node (0) at (0,2) [] {$\Exp$};
    \node (1) at (2,2) [] {$\SExp$};
    \draw[arc, bend left=15] (0) to node [above] {$\detype\argument$} (1); 
    \draw[arc, bend left=15] (1) to node [below] {$\retype\argument$} (0);
    \node (2a) at (-3,0) [] {$\pset\Trm$};
    \node (2) at (0,0) [] {$\Trm$};
    \node (3) at (2,0) [] {$\STrm$};
    \node (3b) at (5,0) [] {$\pset\STrm$};
    \draw[arc, bend left=15] (2) to node [above] {$\detype\argument$} (3); 
    \draw[arc, bend left=15] (3) to node [below] {$\retype\argument$} (2);
    \draw [right hook-latex] (2)--(0);
    \draw [right hook-latex] (2)--(2a);
    \draw [right hook-latex] (3)--(1);
    \draw [right hook-latex] (3)--(3b);
    \draw [arc] (0) to node [above] {$\tms\argument$} (2a);
    \draw [arc] (1) to node [above] {$\tms\argument$} (3b);
  \end{tikzpicture}   
\end{align*}

\section{Graphs}
\label{sec:graphs}

We let $G,H$ range over 2-pointed labelled directed graphs, which we
simply call \emph{graphs} in the sequel. Those are tuples
$\tuple{V,E,\iota,o}$ with $V$ a finite set of vertices,
$E\subseteq V\times X\times V$ a set of edges labelled with $X$,
and $\iota,o\in V$ two distinguished vertices, respectively called
\emph{input} and \emph{output}. We write $\Gph$ for the set of such
graphs. They were introduced independently by Freyd and
Scedrov~\cite[page 208]{FS90}, and Andréka and Bredikhin~\cite{AB95},
for allegories; they form an algebra for the sub-signature of terms:
\begin{itemize}
\item $G\cdot H$ is the \emph{series composition} of the two graphs
  $G$ and $H$, obtained by merging the output of $G$ with the input of $H$;
\item $G\cap H$ is the \emph{parallel composition} of the two graphs
  $G$ and $H$, obtained by merging their inputs and their outputs;
\item $\conv G$ is the graph obtained from $G$ by swapping input and output;
\item $1$ is the graph with a single vertex (both input and output) and no edge;
\item $\top$ is the disconnected graph with two vertices (the input
  and the output) and no edge.
\end{itemize}
See Figure~\ref{fig:graphsdef} for a graphical description of these
operations.

\begin{defi}[Graph of a term]\label{def:grnd-grph}
  The \emph{graph $\Gr u$ of a term $u$} is obtained by using the
  unique homomorphism $\G$ mapping a variable $a$ to the graph with
  two vertices (the input and the output), and a single edge labelled
  $a$ from the input to the output.
\end{defi}
Figure~\ref{fig:ground:graphs} in the Introduction contains some
examples; graphs of non-simple terms are displayed in
Figure~\ref{fig:ground:graphs2}.
\begin{figure}[t]
  \begin{minipage}{\textwidth}
  \noindent
    \begin{align*}
      G\cdot H&\eqdef 
                \begin{tikzpicture}[baseline=(0.south)]
                  \position (0) (0,0);
                  \position (1) (2,0);
                  \position (2) (4,0);
                  \initst (0); \fnst (2);
                  \draw[arc] (0) 
                  to node[midway,fill=white,inner sep = 0] {$G$} (1);
                  \draw[arc] (1)
                  to node[midway,fill=white,inner sep = 0] {$H$} (2);
                \end{tikzpicture}
      &
        1&\eqdef 
           \begin{tikzpicture}[baseline=(0.south)]
             \position (0) (0,0); \initst (0); \fnst (0);
           \end{tikzpicture}
      \\
      G\cap H&\eqdef 
               \begin{tikzpicture}[baseline=(0.south)]
                 \position (0) (0,0);
                 \position (1) (2,0);
                 \initst (0); \fnst (1);
                 \draw[arc] (0) to[bend left] 
                 node[midway,fill=white,inner sep = 0] {$G$} (1);
                 \draw[arc] (0) to[bend right]
                 node[midway,fill=white,inner sep = 0] {$H$} (1);
               \end{tikzpicture}    
      &
        \top&\eqdef 
              \begin{tikzpicture}[baseline=(0.south)]
                \position (0) (0,0);
                \position (1) (1,0);
                \initst (0); \fnst (1);
              \end{tikzpicture}
      \\
      \conv G&\eqdef 
               \begin{tikzpicture}[baseline=(0.south)]
                 \position (0) (0,0);
                 \position (1) (2,0);
                 \draw[arc] (0) to ($ (0) + (-0.8,0) $);
                 \draw[arc] ($ (1) + (0.8,0) $) to (1);
                 \draw[arc] (0) 
                 to node[midway,fill=white,inner sep = 0] {$G$} (1);
               \end{tikzpicture}
      &
        \Gr a&\eqdef 
               \begin{tikzpicture}[baseline=(0.south)]
                 \position (0) (0,0);
                 \position (1) (1,0);
                 \initst (0); \fnst (1);
                 \edge[above] (0)(1)[a];
               \end{tikzpicture}
    \end{align*}
    \caption{The algebra of graphs and the graph of a variable.}
    \label{fig:graphsdef}
  \end{minipage}
  \begin{minipage}{\textwidth}
    \noindent
    \begin{mathpar}
      \G\paren{\paren{a\cdot b}\cap 1} =
      \begin{tikzpicture}[baseline=(0.south)]
        \position(0)(0,0); \position(1)(0,1); 
        \initst[.4] (0); \fnst[.4] (0);
        \edge[left,out=135,in=-135] (0) (1)[a];
        \edge[right,out=-45,in=45] (1) (0)[b];
      \end{tikzpicture}\and
      \G\paren{a\cap \conv b} =
      \begin{tikzpicture}[baseline=(0.south)]
        \position(0)(0,0); \position(1)(1,0); 
        \initst[.4] (0); \fnst[.4] (1);
        \edge[above,out=40,in=140] (0) (1)[a];
        \edge[below,out=-140,in=-40] (1) (0)[b];
      \end{tikzpicture}\and
      \G\paren{\top a\top b\top} =
      \begin{tikzpicture}[baseline=(0.south)]
        \position(0)(0,0); \position(1)(.3,0); \position(2)(1,0);
        \position(3)(1.3,0); \position(4)(2,0); \position(5)(2.3,0);
        \initst[.4] (0);\edge(1)(2)[a];\edge(3)(4)[b]; \fnst[.4] (5);
      \end{tikzpicture}
    \end{mathpar}
    \caption{Graphs of non-simple terms.}
    \label{fig:ground:graphs2}
  \end{minipage}
\end{figure}

\noindent A significant number of the results presented
in~\cite{AB95,AMN11} rely on the following technical lemma.
\begin{lem}[{\cite[Lemma~3]{AB95}}]\label{lem:AB95}
  Let $v\in\Trm$ be a term and write
  $\Gr v=\tuple{V_v,E_v,\iota_v,o_v}$.  Let $S$ be a set and let
  $\sigma\colon X\to \Rel[S]$. For all $i,j\in S$, we have
  $\tuple{i,j}\in\hat\sigma(v)$ if and only if there exists a function
  $\phi: V_v\rightarrow S$ such that
  \begin{enumerate}[label=(\roman*)]
  \item $\phi(\iota_v)=i$, 
  \item $\phi(o_v)=j$, and 
  \item if $\tuple{p,a,q}\in E_v$ then $\tuple{\phi(p),\phi(q)}\in\sigma(a)$.
  \end{enumerate}
\end{lem}

\noindent This lemma actually is the bridge between relational models
and the following natural definition of graph homomorphism.
\pagebreak
\begin{defi}[Graph homomorphism]
  A \emph{graph homomorphism} from $\tuple{V_1,E_1,\iota_1,o_1}$ to
  $\tuple{V_2,E_2,\iota_2,o_2}$ is a map $\phi\colon V_1\rightarrow V_2$
  such that 
  \begin{enumerate}[label=(\roman*)]
  \item $\phi(\iota_1)=\iota_2$, 
  \item $\phi(o_1)=o_2$, and
  \item $\tuple{p,a,q}\in E_1$ entails $\tuple{\phi(p),a,\phi(q)}\in E_2$.
  \end{enumerate}
\end{defi}

\begin{defi}[Preorders $\lessgr$ and $\lesstm$]
  We denote by~$\lessgr$ the
  relation on graphs defined by $G\lessgr G'$ if there exists a graph
  homomorphism from~$G'$ to~$G$. This relation gives rise to a
  relation on terms, written $\lesstm$ and defined by $u\lesstm v$ if
  $\Gr u\lessgr \Gr v$.
\end{defi}

\noindent Those two relations are preorders: they are reflexive and
transitive. As explained in the introduction, they precisely
characterise inclusion under arbitrary relational interpretations:
\begin{thm}[{\cite[Theorem~1]{AB95}, \cite[page 208]{FS90}}]
  \label{thm:AB95}
  For all terms $u,v$, we have 
  \begin{align*}
    \Rel\models u \leq v\text{ if and only if }u\lesstm v.
  \end{align*}
\end{thm}

\noindent To extend this result to the whole syntax of expressions, we
introduce the following generalisation of the language of a regular
expression. Sets of words become sets of graphs.

\begin{defi}[Graphs of an expression]
  The \emph{set of graphs} of an expression $e\in\Exp$ is the set
  $\Gr e \eqdef \setcompr{\Gr u}{u\in \tms e}$
\end{defi}

\noindent
If $e\in\Exp$ is a regular expression (meaning it never uses the
intersection and converse operators, nor the constant $\top$), then
$\Gr e$ is isomorphic to the rational language of words usually
associated with $e$: terms in $\tms e$ are monoid expressions, so that
graphs in $\Gr e$ can be identified with words. We define accordingly
a notion of regular set of graphs.
\begin{defi}[Regular set of graphs]
  A set $S\subseteq\Gph$ of graphs is \emph{regular} if there
  exists an expression $e\in\Exp$ such that $\Gr e = S$.
\end{defi}

\noindent
Coming back to representable Kleene allegories, we need a slight
refinement of Lemma~\ref{lem:AMN2+} to obtain the characterisation
announced in the Introduction:
\begin{lem} 
  \label{lem:AMN2+:refined} 
  For every expression $e\in\Exp$, every set $S$ and every relational
  interpretation $\sigma\colon X\rightarrow \Rel[S]$, we have
  \begin{equation*}
    \hat\sigma(e)=\bigcup_{u\in\cltm{\tms e}}\hat\sigma(u).
  \end{equation*}
\end{lem}
\begin{proof}
  We use Lemma~\ref{lem:AMN2+} and the fact that
  $\hat\sigma(w)\subseteq\hat\sigma(u)$ whenever $w\lesstm u$, thanks
  to Theorem~\ref{thm:AB95}.
\end{proof}

\noindent
Given a set $S$ of graphs, we write $\clgr S$ for its downward
closure: $\clgr S\eqdef\setcompr{G}{G\lessgr G', G'\in S}$.
Similarly, we write $\cltm S$ for the downward closure of a set of
terms w.r.t.~$\lesstm$.
\begin{thm}\label{thm:interlang}
  The following properties are equivalent, for all expressions
  $e,f\in\Exp$:
  \begin{enumerate}[label=(\roman*)]
  \item\label{i:rel} $\Rel\models e \leq f$,
  \item\label{i:tl} $\tms e \subseteq \cltm{\tms f}$,
  \item\label{i:gl} $\Gr e \subseteq \clgr{\Gr f}$.
  \end{enumerate}
\end{thm}

\noindent 
(Note that this theorem amounts to the characterisation announced in
Section~\ref{sec:intro}~\eqref{eq:charac:kal}: for all sets $X,Y$, we
have $\clgr X\subseteq \clgr Y$ if and only if $X\subseteq \clgr
Y$.)
In other words, deciding the (in)equational theory of representable
Kleene allegories reduces to comparing regular sets graphs, modulo
homomorphism.
\begin{proof}
  The implication $\ref{i:tl}\Rightarrow\ref{i:rel}$ follows
  easily from Lemma~\ref{lem:AMN2+:refined}, and
  $\ref{i:gl}\Rightarrow\ref{i:tl}$ is a matter of unfolding
  definitions. For $\ref{i:rel}\Rightarrow\ref{i:gl}$, we mainly
  use Lemma~\ref{lem:AB95}:

  Let $e,f$ be two expressions such that $\Rel\models e\leq f$,
  and $u\in\tms e$ such that $\Gr u=\tuple{V_u,E_u,\iota_u,o_u}$; we
  can build an interpretation
  $\sigma\colon X\rightarrow \Rel[{V_u}]$ by specifying:
  \[\sigma(a)\eqdef\setcompr{\tuple{p,q}}{\tuple{p,a,q}\in E_u}.\]
  It is simple to check that
  $\hat\sigma(u)=\set{\tuple{\iota_u,o_u}}$.  By Lemma~\ref{lem:AMN2+}
  and $\Rel\models e\leq f$, we know that
  \[\hat\sigma(u)\subseteq\hat\sigma (f)=\bigcup_{v\in\tms
      f}\hat\sigma(v).\] Thus there is some $v\in\tms f$ such that
  $\tuple{\iota_u,o_u}\in\hat\sigma(v)$. By Lemma~\ref{lem:AB95} we get that there
  is a map $\phi:V_v\rightarrow V_u$ such that $\phi(\iota_v)=\iota_u$;
  $\phi(o_v)=o_u$ and \[\tuple{p,a,q}\in E_v\Rightarrow \tuple{\phi(p),\phi(q)}\in\sigma(a).\]
  Using the definition of $\sigma$, we rewrite this last condition as
  \[\tuple{p,a,q}\in E_v\Rightarrow \tuple{\phi(p),a,\phi(q)}\in
  E_u.\]
  Thus $\phi$ is a graph homomorphism from $\Gr v$ to $\Gr u$, proving
  that $\Gr u\lessgr \Gr v$, hence $\Gr u\in \clgr{\Gr f}$.
\end{proof}

\section{Series-parallel graphs}
\label{sec:sp}

When focusing on identity-free Kleene lattice, we can forget the
converse operation and the constants $1$ and $\top$, and work with
simple terms and expressions. In such a situation, the manipulated
graphs are \emph{series-parallel}.

More precisely, consider Valdes et al.\ \emph{SP-rewriting
  system}~\cite{Valdes79}. This system, extended to edge-labelled
graphs, is recalled in Figure~\ref{fig:sp-rewriting system}. The
second rule can only be applied if the intermediate vertex on the
left-hand side is not the input and the output, and has no other
adjacent edge. Valdes et al.~\cite{Valdes79} showed that that system
has the Church-Rosser property. This property can be extended to
labelled graphs without difficulty, modulo the congruence generated by
the associativity of $\cdot$ and the associativity and commutativity
of $\cap$.

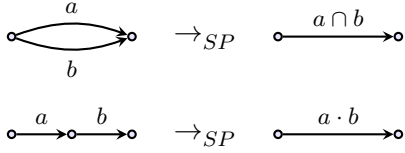
\begin{figure}[t]
  \begin{align*}
    \begin{tikzpicture}[scale=.8,baseline=(x.south)] 
      \state (x) (0,0); \state (y) (2,0); 
      \edge[above,out=20,in=160] (x) (y)[a];
      \edge[below,out=-20,in=-160] (x) (y)[b];
    \end{tikzpicture}\quad&\spred\quad
                            \begin{tikzpicture}[scale=.8,baseline=(z.south)] 
                              \state (z) (5,0); 
                              \state (t) (7,0); % \edge (3,0) (4,0);
                              \edge[above] (z) (t)[a\cap b];
                            \end{tikzpicture}\\
    \begin{tikzpicture}[scale=.8,baseline=(x.south)] 
      \state (x) (0,0); \state (w) (1,0); \state (y) (2,0);
      \edge[above] (x) (w)[a]; \edge[above] (w) (y)[b];
    \end{tikzpicture}\quad&\spred\quad
                            \begin{tikzpicture}[scale=.8,baseline=(z.south)] 
                              \state (z) (5,0); \state (t) (7,0); %\draw[arc] (3,0) -- (4,0);
                              \edge[above] (z) (t)[a\cdot b];
                            \end{tikzpicture} 
  \end{align*}
  \caption{The SP-rewriting system}
  \label{fig:sp-rewriting system}
\end{figure}

A graph~$G$ is \emph{series-parallel}
if it SP-reduces in some number of steps to a graph of the following shape.
\begin{align*}
\begin{tikzpicture}[yscale=.8,xscale=.5,baseline=(z.south)] 
  \state (z) (5,0); 
  \state (t) (7,0); %\draw[arc] (3,0) -- (4,0);
  \edge[above] (z) (t)[e];
  \initst (z); \fnst (t);
\end{tikzpicture}
\end{align*}
In such a case, $e$ is a simple term called a \emph{term
  representation} of~$G$. This term is not unique but we can chose one
which we write $\trm G$.  We write $\SP[Y]$ for the set of
series-parallel graphs labelled in a set $Y$.
Every series-parallel graph has a single source, its input and a
single sink, its output. Also note the series-parallel graphs are
acyclic.

\begin{lem}
  A graph is series-parallel if and only if it is the graph of a
  simple term.
\end{lem}

\noindent
In other words, we have the following situation.
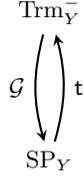
\begin{align*}
  \begin{tikzpicture}
    \node (0) at (0,2) [] {$\STrm[Y]$};
    \node (1) at (0,0) [] {$\SP[Y]$};
    \draw[arc, bend right=15] (0) to node [left] {$\G$} (1); 
    \draw[arc, bend right=15] (1) to node [right] {$\trmc$} (0);
  \end{tikzpicture}     
\end{align*}

\noindent It is thus natural to consider only simple expressions for
defining the notion of regular set of series-parallel graphs.
\begin{defi}[Regular set of series-parallel graphs]
  \label{def:reg:sp}
  A set $S\subseteq\SP[X]$ of series-parallel graphs is \emph{regular}
  if there exists a simple expression $e\in\SExp[X]$ such that
  $\Gr e = S$.
\end{defi}

Recall the functions $\detype\argument\colon \Trm \to \STrm$ and
$\retype\argument\colon \STrm \to \Trm$ from the previous section. A
counterpart to the latter can be defined at the level of graphs.

\begin{defi}[$\retype G$]
  Let $G=\tuple{V,E,\iota,o}\in \Gph[\Xb]$ be a graph labelled with
  $\Xb$.  Let $\equiv$ be the smallest equivalence relation on $V$
  containing all pairs $\tuple{i,j}$ such that $\tuple{i,1,j}\in E$,
  and $[i]$ be the equivalence class of $i$. We associate with $G$ the
  following graph labelled in $X$.
  \begin{align*}
    \retype G&\eqdef \tuple{V/{\equiv},E',[\iota],[o]},\text{ where}\\
    E'&\eqdef\setcompr{\tuple{[i],x,[j]}}{ x\in X,~
        \exists \tuple{k,l}\in[i]\times[j]: \tuple{k,x,l}\in E\text{ or }
        \tuple{l,x',k}\in E }.
  \end{align*}
\end{defi}

\noindent
Informally, $\retype G$ is obtained by merging all vertices related by
an edge labelled with $1$, removing all edges labelled with $\top$,
and reversing all edges labelled with some $a'\in\Xb$. This operation
is defined on arbitrary graphs labelled in $\Xb$, but we shall use it
only on series-parallel graphs.

\begin{prop}
  \label{prop:type:graphs}
  For every term $u\in\Trm$, we have $\Gr u = \retype{\Gr{\detype u}}$.
  For every simple term $u\in\STrm$, we have
  $\Gr{\retype u} = \retype{\Gr u}$. In other words, the following
  diagrams commute.
  \begin{center}
    \begin{tikzcd}    
      \Trm \arrow[r,"\detype\argument"] \arrow[d,"\G"'] &
      \STrm \arrow[d,"\G"] \\
      \Gph & \SP \arrow[l,"\retype\argument"] 
    \end{tikzcd}
    \qquad\qquad
    \begin{tikzcd}    
      \Trm \arrow[d,"\G"'] &
      \STrm \arrow[l,"\retype\argument"'] \arrow[d,"\G"] \\
      \Gph & \SP \arrow[l,"\retype\argument"] 
    \end{tikzcd}
  \end{center}
\end{prop}
%\todo{include proofs in phd.tex?}

\noindent
This result can be extended to expressions and sets of graphs, which
we can depict using the following commutative diagrams.
\begin{center}
  \begin{tikzcd}    
    \Exp \arrow[r,"\detype\argument"] \arrow[d,"\G"'] &
    \SExp \arrow[d,"\G"] \\
    \pset\Gph & \pset \SP \arrow[l,"\retype\argument"] 
  \end{tikzcd}
  \qquad\qquad
  \begin{tikzcd}    
    \Exp \arrow[d,"\G"'] &
    \SExp \arrow[l,"\retype\argument"'] \arrow[d,"\G"] \\
    \pset\Gph & \pset\SP \arrow[l,"\retype\argument"] 
  \end{tikzcd}
\end{center}

\section{Petri automata}
\label{sec:pa}

We fix an alphabet $Y$, which we will later instantiate either with
the set $X$ of variables used in expressions, or with its extended
version $\Xb$.

A Petri automaton is essentially a safe Petri net where the arcs
coming out of transitions are labelled by letters from $Y$.
\begin{defi}[Petri Automaton]\label{def:auto}
  A \emph{Petri automaton} $\A$ over the alphabet $Y$ is a tuple
  $\tuple{P,\T,\iota}$ where:
  \begin{itemize}
  \item $P$ is a finite set of \emph{places},
  \item $\T\subseteq \pset P\times\pset{Y\times P}$ is a set of
    \emph{transitions},
  \item $\iota\in P$ is the \emph{initial place} of the automaton. 
  \end{itemize}
  For each transition $\tr=\trexp\in \T$, $\trin$ is assumed to be
  non-empty; $\trin\subseteq P$ is the \emph{input of $\tr$}; and
  $\trout\subseteq Y\times P$ is the \emph{output of $\tr$}.
  Transitions with empty outputs are called \emph{final}, and
  transitions with input $\set\iota$ are called \emph{initial}.
  We write $\PA[Y]$ for the set of Petri automata over the alphabet $Y$.
\end{defi}
We write $\settrout\eqdef\setcompr p{\exists a, \tuple{a,p}\in\trout}$
for the set of places appearing in the output of~$\tr$. We will add
two constraints on this definition along the way
(Constraints~\ref{cstr:safety} and~\ref{cstr:sp}), but we need more
definitions to state them. An example of such an automaton is
described in Figure~\ref{fig:auto}. The graphical representation used
here draws round vertices for places and rectangular vertices for
transitions, with the incoming and outgoing arcs to and from the
transition corresponding respectively to the inputs and outputs of
said transition. The initial place is denoted by an unmarked incoming
arc.
%
%\FloatBarrier
\begin{figure}[t]
  \centering
  \begin{tikzpicture}[scale=.8]
    \tikzstyle{every node}=[font=\footnotesize]
    % \renewcommand\trans[3][]{\node[trans] (#2) at (#3) {$#1$}}
    % \draw[thick,dotted,fill=red!10] (12.5,0.5) rectangle (13.5,3.5);
    % \draw[thick,dotted,fill=red!10] (8,-0.5) rectangle (9,0.5);
    \state[A](A)(0,2);
    \state[B](B)(3,3);
    \state[C](C)(6,4);
    \state[D](D)(10.5,4);
    \state[E](E)(6,2);
    \state[F](F)(13,3);
    \state[G](G)(13,1);
    \state[H](H)(4,0);
    \state[I](I)(8.5,0);
    \trans[0] (0) (1.5,2);
    \trans[1] (1) (4.5,3);
    \trans[2] (2) (7.5,4);
    \trans[3] (3) (9,3);
    \trans[4] (4) (11.5,3);
    \trans[5] (5) (1,0);
    \trans[6] (6) (5.5,0);
    \transf[7] (7) ($(F)!.5!(G)+(1.5,0)$);
    \transf[8] (8) ($(I)+(1.5,0)$);
    \initst (A);
    \edge (A) (0);
    \edge[bend right] (A) (5);
    \edge (B) (1);
    \edge (C) (2);
    \edge[bend left] (D) (3);
    \edge[bend left] (D) (4);
    \draw[arc] (E) to (E -| 3)
    to[out=0,in=-90] ($(E -|3)!.5!(3)+(.8,0)$)
    to[out=90,in=0] (3);
    \edge[bend right] (E) (4);
    \edge (H) (6);
    \edge[bend left][near start,above] (0) (B)[b];
    \edge[out=-20,in=180][midway,above] (0) (G)[a];
    \edge[bend left][midway,above] (1) (C)[c];
    \edge[bend right][midway,below] (1) (E)[b];
    \edge[][midway,above] (2) (D)[a];
    \edge[out=175,in=-45][near start,above] (3) (C)[c];
    \edge[out=185,in=45][near end,above] (3) (E)[b];
    \edge[][midway,above] (4) (F)[d];
    \edge[][midway,above] (5) (H)[a];
    \edge[][midway,above] (6) (I)[b];
    \edge[bend left] (F) (7);
    \edge[bend right] (G) (7);
    \edge (I) (8);
  \end{tikzpicture}
  \caption{A Petri automaton.}
  \label{fig:auto}
\end{figure}
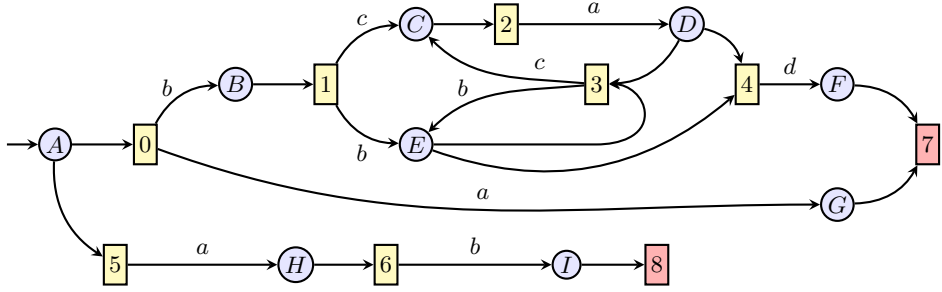

\subsection{Runs and reachable states}
\label{ssec:runs}

We define the operational semantics of Petri automata. Let us fix a
Petri automaton~$\A=\tuple{P,\T,\iota}$ for the remainder of this
section. A \emph{state} of this automaton is a set of places (i.e., a
configuration in the Petri nets terminology). In a given
state~$S\subseteq P$, a transition~$\tr=\trexp$ is \emph{enabled}
if~$\trin\subseteq S$. In that case, we may fire~$\tr$, leading to a
new state~$S'=\paren{S\setminus\trin}\cup\settrout$.  This will be
denoted in the following by~$\atrans S {S'}$. We extend this notation
to sequences of transitions in the natural way:
\begin{equation*}  
  \infer{\atrans[\tr[1];\tr[2];\dots;\tr[n]]{S_0}{S_n}}
  {\atrans[\tr[1]]{S_0}{S_1},
    &\atrans[\tr[2];\dots;\tr[n]]{S_1}{S_n}}
\end{equation*}
In that case we say that~$\tuple{S_0,\tr[1];\tr[2];\dots;\tr[n],S_n}$
is a \emph{valid run}, or simply run, from~$S_0$
to~$S_n$. If~$S_0=\set \iota$ then the run is \emph{initial} and
if~$S_n=\emptyset$ it is \emph{final}. A run that is both initial and
final is called \emph{accepting}. A state~$S$ is reachable in~$\A$ if
there is an initial run leading to~$S$.

We write $\Run$ for the set of runs of an automaton $\A$, and
$\Run\paren{A,B}$ for the set of runs from state $A$ to state $B$ in
$\A$. The set of its accepting runs is then equal
to~$\Run\paren{\set\iota,\emptyset}$, and written $\aRun$.

\begin{exa}[Accepting run]\label{ex:run}
  The triple~$\tuple{\set A,0;1;2;3;2;4;7,\emptyset}$ is a valid run in
  the automaton from Figure~\ref{fig:auto}. It is easy enough to check that:
  \begin{equation*}
    \begin{array}{l@{\quad}l}
    \atrans[0]{\set A}{\set{B,G}};
      &\atrans[1]{\set {B,G}}{\set{C,E,G}};\\
      \atrans[2]{\set {C,E,G}}{\set{D,E,G}};
      &\atrans[3]{\set {D,E,G}}{\set{C,E,G}};\\
      \atrans[4]{\set {D,E,G}}{\set{F,G}};
      &\atrans[7]{\set{F,G}}\emptyset. \\
    \end{array}
  \end{equation*}
  Thus we can prove that~$\atrans[0;1;2;3;2;4;7]{\set A}{\emptyset}$.
  Furthermore, as~$A$ is the initial place, this run is accepting. It
  can be represented graphically as in Figure~\ref{fig:run}.
  \begin{figure}[t]
    \centering
    \begin{tikzpicture}[yscale=.7]
      \node (A) at (0,-0.5) {$A$};
      \node (B) at (2,0) {$B$};
      \node (G1) at (2,-1) {$G$};
      \node (C) at (4,0.5) {$C$};
      \node (E1) at (4,-0.5) {$E$};
      \node (G2) at (4,-1) {$G$};
      \node (D) at (6,0.5) {$D$};
      \node (E2) at (6,-0.5) {$E$};
      \node (G3) at (6,-1) {$G$};
      \node (C1) at (8,0.5) {$C$};
      \node (E3) at (8,-0.5) {$E$};
      \node (G4) at (8,-1) {$G$};
      \node (D1) at (10,0.5) {$D$};
      \node (E4) at (10,-0.5) {$E$};
      \node (G5) at (10,-1) {$G$};
      \node (F) at (12,0) {$F$};
      \node (G6) at (12,-1) {$G$};
      \trans[0] (0) (1,-0.5);
      \trans[1] (1) (3,0);
      \trans[2] (2) (5,0.5);
      \trans[3] (3) (7,0);
      \trans[2] (4) (9,0.5);
      \trans[4] (5) (11,0);
      \transf[7] (6) ($(F)!.5!(G6)+(1,0)$);
      % \node[trans,dotted] (6) at (13,-0.5) {$6$};
      \edge (A) (0);
      \edge[above] (0) (B)[b];
      \edge[below] (0) (G1)[a];

      \edge (B) (1);
      \draw[thick,dotted] (G1) to (G2);
      \edge[above] (1) (C)[c];
      \edge[below] (1) (E1)[b];

      \edge (C) (2);
      \draw[thick,dotted] (E1) to (E2);
      \draw[thick,dotted] (G2) to (G3);
      \edge[above] (2) (D)[a];

      \edge (D) (3);
      \edge (E2) (3);
      \draw[thick,dotted] (G3) to (G4);
      \edge[above] (3) (C1)[c];
      \edge[below] (3) (E3)[b];

      \edge (C1) (4);
      \draw[thick,dotted] (E3) to (E4);
      \draw[thick,dotted] (G4) to (G5);
      \edge[above] (4) (D1)[a];

      \edge (D1) (5);
      \edge (E4) (5);
      \draw[thick,dotted] (G5) to (G6);
      \edge[above] (5) (F)[d];

      \edge (F) (6);
      \edge (G6) (6);
      % \draw[thick,dotted] (F) to (6);
      % \draw[thick,dotted] (G6) to (6);
      
    \end{tikzpicture}
    \caption{An accepting run in the automaton from Figure~\ref{fig:auto}.}
    \label{fig:run}
  \end{figure}
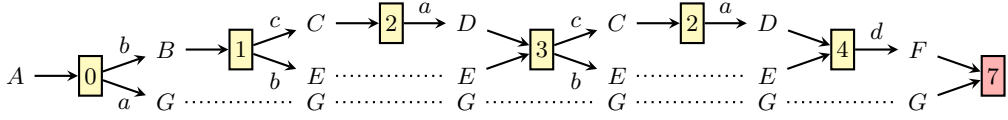
\end{exa}

We may now state the first constraint we impose on Petri automata.
\begin{cstr}[Safety]\label{cstr:safety}
  For every reachable state~$S$ of a Petri automaton $\A$, if for some
  transition $\tr$ we have $\atrans S {S'}$, then
  \begin{align*}
  \paren{S\setminus\trin}\cap\settrout=\emptyset.
  \end{align*}
\end{cstr}
This constraint corresponds to the classic Petri net property of
safety (also called one-boundedness): at any time, there is at most
one token in a given place, thus justifying or use of sets of places
rather than multi-sets of places for configurations. Checking this
constraint is feasible: the set of transitions is finite, and there
are finitely many reachable states (those are subsets of a fixed
finite set). The net in Figure~\ref{fig:non-safe-petri} is not safe.
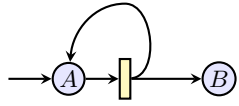
\begin{figure}[t]
  \centering
  \begin{tikzpicture}
    \state[A](A) (0,0);\state[B](B)(2,0);\initst(A);
    \trans(0)(.75,0);\edge(A)(0);\edge(0)(B);
    \draw[arc] (0) to[out=5,in=0] ($(0)+(0,1)$)
    to[out=180,in=85](A);
  \end{tikzpicture}
  \caption{Example of non-safe Petri Net.}
  \label{fig:non-safe-petri}
\end{figure}

\subsection{Graphs produced by a Petri automaton}
\label{ssec:traces}

Petri automata produce series-parallel graphs: to each accepting run
we associate such a graph, called its \emph{trace}. Consider an
accepting run~$\tuple{\set\iota,\tr[0];\dots;\tr[n],\emptyset}$. Its
trace is constructed by creating a vertex~$k$ for each
transition~$\tr[k] = \trexp[k]$ of the run. We add an
edge~$\tuple{k, a, l}$ whenever there is some place~$q$ such
that~$\tuple{a, q}\in \trout[k]$, and~$\tr[l]$ is the first transition
after~$\tr[k]$ in the run with~$q$ among its inputs.
The definition we give here is a generalisation for arbitrary valid
runs, that coincides with the informal presentation we just gave on
accepting runs.

\begin{figure}[t]
  \centering
  \begin{tikzpicture}[yscale=.4,xscale=1.4]
    \state[0](0)(0,0);
    \initst (0);
    \state[1](1)(1,1);
    \state[2](2)(2,2);
    \state[3](3)(3,1);
    \state[4](4)(4,2);
    \state[5](5)(5,1);
    \state[6](6)(6,0);
    \fnst (6);
    \edge[above,in=-170,out=60] (0) (1)[b];
    \edge[below,out=-20,in=-160] (0) (6)[a];
    \edge[above,in=-170,out=60] (1) (2)[c];
    \edge[below,out=-20,in=-160] (1) (3)[b];
    \edge[above,out=-10,in=120] (2) (3)[a];
    \edge[above,in=-170,out=60] (3) (4)[c];
    \edge[below,out=-20,in=-160] (3) (5)[b];
    \edge[above,out=-10,in=120] (4) (5)[a];
    \edge[above,out=-10,in=120] (5) (6)[d];
  \end{tikzpicture}
  \caption{Trace of the run from Example~\ref{ex:run}.}
  \label{fig:ex:trace}
\end{figure}
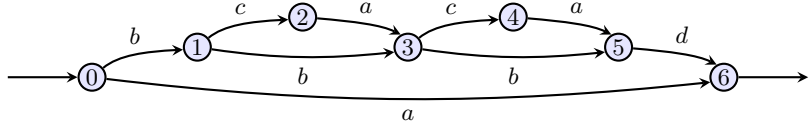

\begin{defi}[Trace of a run]\label{def:trace}
  Let~$\run=\tuple{S,\tr[0];\dots;\tr[n],S'}$ be a run in~$\A$.  For
  every $k$ and~$p\in\settrout[k]$, we define
  $$\nu\paren{k,p}=\setcompr{l}{l>k\text{ and }p\in\trin[l]}.$$
  The \emph{trace} of~$\run$, denoted by~$\Gr\run$, is the graph with
  vertices~$\set{0,\dots,n}\cup S'$ and the set of edges defined by:
  \begin{equation*}
    E_\run=\setcompr{\tuple{k,a,l}}{\tuple{a,p}\in\trout[k]
      \text{ and }\paren{p=l\wedge\nu\paren{k,p}=\emptyset}
      \vee \paren{l=\min{\paren{\nu\paren{k,p}}}}}. 
  \end{equation*}
\end{defi}

\noindent The trace of the run given in Example~\ref{ex:run} is presented in
Figure~\ref{fig:ex:trace}. 
Unfortunately, the trace of an accepting run is not necessarily
series-parallel. Consider for instance the net in
Figure~\ref{fig:bad:sp1}; the trace of its only accepting run is the
graph drawn on the right-hand side, which is not series-parallel.
\begin{figure}[t]
  \hfill
  \begin{tikzpicture}
    \state[A](A)(0,0);\initst(A);\state[B](B)(2,.5);
    \state[C](C)(4,0);\state[D](D)(4,-1);\state[E](E)(6,.5);
    \state[F](F)(6,-.5);
    \trans(0)(.75,0);\trans(1)(2.75,.5);\trans(2)(5,-.5);
    \transf(3)(7,0);
    \edge(A)(0);
    \edge[bend left,above](0)(B)[a];\edge[bend right,below](0)(D)[c];
    \edge(B)(1);
    \edge[bend left,above](1)(E)[b];\edge[bend right,below](1)(C)[e];
    \edge[bend left](C)(2);\edge[bend right](D)(2);\edge(2)(F)[d];
    \edge[bend left](E)(3);\edge[bend right](F)(3);
  \end{tikzpicture}
  \hfill
  \begin{tikzpicture}
    \position(0)(0,0);\position(1)(1.5,.75);
    \position(2)(1.5,-.75);\position(3)(3,0);
    \initst[.5](0);\fnst[.5](3);
    \edge[above,bend left](0)(1)[a];
    \edge[above,bend left](1)(3)[b];
    \edge[below,bend right](0)(2)[c];
    \edge[below,bend right](2)(3)[d];
    \edge[right](1)(2)[e];
  \end{tikzpicture}
  \hfill\mbox{}
  \caption{An automaton yielding a non series-parallel trace.}
  \label{fig:bad:sp1}  
\end{figure}
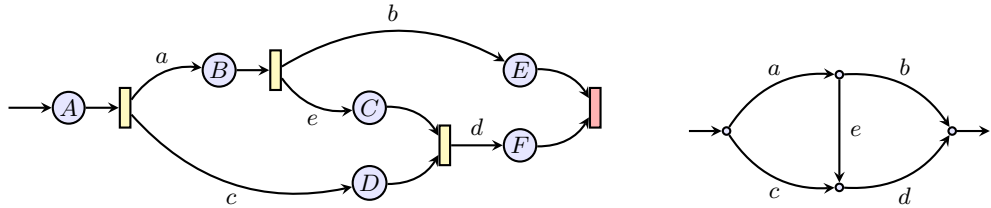
This leads us to the second constraint we impose on Petri automata.
\begin{cstr}[Series-parallel]\label{cstr:sp}
  Every graph $G\in\Gr\A$ is series-parallel.
\end{cstr}
\noindent Despite its infinitary formulation, this property is
decidable. In fact, the procedure we describe in
Section~\ref{sec:pa:exp} to extract a simple expression out of a Petri
automaton fails in finite time if it is given an automaton violating
this constraint, and succeeds otherwise.
Thanks to this constraint, we can restrict our attention to
\emph{proper} runs:
\begin{defi}[Proper run]
  A run $R=\tuple{S_0,\tr[1];\dots;\tr[n],S_n}\in\Run$ is
  \emph{proper} if for all $1\leqslant i\leqslant n$ we have 
$%  \begin{equation*}
    \trout[i]=\emptyset\Rightarrow \paren{i=n\wedge S_n=\emptyset}.
$%  \end{equation*}
\end{defi}
\noindent Sub-runs of a proper run are proper by definition, and we have
\begin{lem}\label{lem:cst2:finaltrans}
  All accepting runs of a Petri automaton (satisfying
  Constraints~\ref{cstr:safety} and~\ref{cstr:sp}) are proper runs.
\end{lem}
\begin{proof}
  Let $R=\tuple{\set\iota,\tr[0];\dots;\tr[n],\emptyset}$ be an
  accepting run (which is series-parallel by
  Constraint~\ref{cstr:sp}). Since the last state is empty we know
  that the run must contain a final transition. Hence we may
  reformulate the definition of proper as $\trout[n]=\emptyset$ and
  $\forall 0\leqslant i<n$, $\trout[i]\neq\emptyset$.

  The proof relies on the following observation: the node $i$ in
  $\Gr R$ is a sink if and only if $\trout[i]=\emptyset$. To prove
  this, remember that the edges coming out of $i$ in $\Gr R$ are:
  \[\setcompr{\tuple{i,a,\min\setcompr{l}{l>i\text{ and
          }p\in\trin[l]}}}{\forall\tuple{a,p}\in\trout[i]}.\]
  Because the run ends in state $\emptyset$, we know that
  ${\forall\tuple{a,p}\in\trout[i]}$ there is a later transition $\tr[j]$
  consuming the token at place $p$ (otherwise $p$ would remain in the last
  state). This entails that if $\tr[i]$ is not final, then $i$ has
  a successor in $\Gr R$. On the other hand, if $\trout[i]$ is empty,
  then there cannot be a transition coming out of node $i$.

  As a series-parallel graph has a single sink, it can only use a
  single final transition. Lastly, notice that $\tuple{i,a,j}\in E_R$
  entails $i<j$. This means that the sink can only be the maximal node
  (for the ordering of natural numbers).
\end{proof}

\noindent We finally define the set of graphs produced by a Petri
automaton.
\begin{defi}
  The \emph{set of (series-parallel) graphs produced by a Petri
    automaton} $\A$ is the set of traces of accepting (proper) runs of
  $\A$:
  \begin{align*}
    \Gr\A\eqdef\setcompr{\Gr\run}{\run\in\aRun}.
  \end{align*}
\end{defi}

\section{From expressions to Petri automata}
\label{sec:exp:pa}

So far we have seen two ways of defining sets of series-parallel
graphs: through simple expressions and through Petri automata.
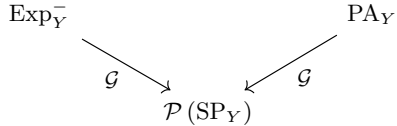
\begin{center}
  \begin{tikzcd}
    \SExp[Y]\arrow[dr,"\G"']&& \PA[Y]\arrow[dl,"\G"]\\ &\pset{\SP[Y]}
  \end{tikzcd}
\end{center}
Now we show how to associate with every simple expression $e\in\SExp[Y]$
a Petri automaton~$\A(e)$ labelled over $Y$ such that:
$\Gr e = \Gr {\A(e)}$. It will follow that regular sets of
series-parallel graphs are recognisable (Theorem~\ref{thm:kl1}).

The construction goes by structural induction on the simple
expressions $e$: we provide Petri automata constructions for each
syntactic entry, thus proving that recognisable languages are closed
under union, iteration and series and parallel compositions.

\subsubsection*{Petri automaton for \texorpdfstring{$0$}{0}.}

We need the automaton for $0$ to have no accepting run. The
following automaton $\underline 0$ has this property.

\begin{center}
  \begin{tikzpicture}[baseline=(0.west)]
    \state[$\quad$](0)(0,0);\initst (0);
  \end{tikzpicture}
\end{center}

\subsubsection*{Petri automaton for a letter.}

For $a\in Y$, we want an automaton whose set of traces is simply the
graph $\Gr{a}$. The following automaton $\underline a$ works.

\begin{center}
  \begin{tikzpicture}[baseline=(0.west)]
    \state[$\quad$](0)(0,0);\initst (0); \state[$\quad$](1)(3,0);
    \trans (2) (1,0); \transf (3) (4,0); \edge (0) (2);\edge (1)
    (3);\edge (2) (1)[a];
  \end{tikzpicture}
\end{center}

\subsubsection*{Union of Petri automata.}

Let $\A_1=\tuple{P_1,\T_1,\iota_1}$ and
$\A_2=\tuple{P_2,\T_2,\iota_2}$ be two Petri automata with disjoint
sets of places.
To get an automation for $\Gr{\A_1}\cup\Gr{\A_2}$, we simply put the
two automata side by side, we add a new initial place $\iota$, and for
every initial transition $\tuple{\set{\iota_j},\trout}$ we add a new
transition $\tuple{\set{\iota},\trout}$. Formally:
\begin{align*}
\A_1\sumtm\A_2\eqdef
\tuple{P_1\cup P_2\cup\set\iota, \T_1\cup\T_2\cup
  \setcompr{\tuple{\set{\iota},\trout}}
  {\tuple{\set{\iota_j},\trout}\in\T_j,j\in\set{1,2}}, \iota}.
\end{align*}
Equivalently the transitions in this automaton follow from the
following rules:
\begin{mathpar}
  \infer[j\in{1,2}]
  {\atrans[\tr][\A_1\sumtm\A_2]S T}{\atrans[\tr][\A_j] S T}
  \and\infer[j\in{1,2}]
  {\atrans[\tuple{\set\iota,S}][\A_1\sumtm\A_2]{\set\iota} S}
  {\atrans[\tuple{\set{\iota_j},S}][\A_j]{\set{\iota_j}} S}
\end{mathpar}

\begin{lem}\label{lem:ok_union}
  We have $\Gr{\A_1\sumtm\A_2}=\Gr{\A_1}\cup\Gr{\A_2}$.
\end{lem}
\begin{proof}
  Let $G\in\Gr{\A_1}\cup\Gr{\A_2}$. There is an accepting run
  $R=\tuple{\set{\iota_j},\tr[0];\tr[1]\dots\tr[n];\emptyset}$ in
  $\A_j$, for either $j=1$ or $j=2$, such that $G=\Gr R$. By
  definition of a valid run, we know that there is a state $S$ such
  that $\atrans[\tr[0]][\A_j]{\set{\iota_j}} S$ and
  $\atrans[\tr[1];\dots;\tr[n]][\A_j] S \emptyset$.  Because
  $\A_1\sumtm\A_2$ contains in particular all places and transitions of
  $\A_j$, we can deduce that
  $\atrans[\tr[1];\dots;\tr[n]][\A_1\sumtm\A_2] S \emptyset$.
  Furthermore~$\atrans[\tr[0]][\A_j]{\set{\iota_j}} S$ entails
  $\tr[0]=\tuple{\set{\iota_j},\trout[0]}$, thus in $\A_1 \sumtm \A_2$
  we have the transition $\tr*[0]=\tuple{\set\iota,\trout[0]}$. Hence:
  $\atrans[\tr*[0]][\A_1\sumtm\A_2]{\set{\iota}} S$. This proves that
  $R'=\tuple{\set{\iota},\tr*[0];\tr[1]\dots\tr[n];\emptyset}$ is an
  accepting run in $\A_1\sumtm\A_2$. As $\Gr{R'}=\Gr{R}$, we have proved
  that $G\in\Gr{\A_1\sumtm\A_2}$.

  Now take $G\in\Gr{\A_1\sumtm\A_2}$, and
  $R=\tuple{\set{\iota},\tr[0];\tr[1]\dots\tr[n];\emptyset}$ such that
  $R=\Gr{G}$. Necessarily, $\tr[0]$ is of the shape
  $\tuple{\set\iota,\trout[0]}$, and thus was produced from
  $\tr*[0]=\tuple{\set{\iota_j},\trout[0]}\in\T_j$, for either $j=1$
  or $j=2$. If $S_0$ is the state reached after the first transition,
  meaning $\atrans[\tr[0]][\A_1 \sumtm\A_2]{\set{\iota_j}} {S_0}$, it
  follows that $S_0\subseteq P_j$. Because we assumed
  $P_1\cap P_2=\emptyset$, it is straightforward that
  $\atrans[\tr][\A_1 \sumtm\A_2] S T$ and $S\subseteq P_j$ entails
  $\atrans[\tr][\A_j] S T$ and $T\subseteq P_j$. This result extends
  to sequences of transitions, allowing us to check that because
  $\atrans[\tr[1];\dots;\tr[n]][\A_1\sumtm\A_2] {S_0} \emptyset$ and
  $S_0\subseteq P_j$ we have
  $\atrans[\tr[1];\dots;\tr[n]][\A_j] {S_0} \emptyset$. Thus
  $R'=\tuple{\set{\iota_j},\tr*[0];\tr[1]\dots\tr[n];\emptyset}$ is an
  accepting run in $\A_j$, and as $\Gr R = \Gr{R'}$, we get
  $G\in \Gr{\A_j}\subseteq\Gr{\A_1}\cup\Gr{\A_2}$.
\end{proof}

\subsubsection*{Series composition of Petri automata.}

Let $\A_1=\tuple{P_1,\T_1,\iota_1}$ and
$\A_2=\tuple{P_2,\T_2,\iota_2}$ be two Petri automata with disjoint
sets of places. To obtain an automaton for
$\Gr{\A_1}\cdot\Gr{\A_2}$, we need that every accepting run~$R_1$
in~$\A_1$ be followed by every accepting run~$R_2$ in~$\A_2$. Because of
Lemma~\ref{lem:cst2:finaltrans}, we know that~$R_1$ must end with a
final transition, and the definition of accepting run imposes
that~$R_2$ starts with an initial transition. Thus it suffices to
put~$\A_1$ in front of~$\A_2$, declare~$\iota_1$ initial, remove the
final transitions~$\tuple{\trin,\emptyset}$ in~$\A_1$ and replace them
with~$\tuple{\trin,\trout}$, for every initial
transition~$\tuple{\set{\iota_2},\trout}$ in~$\A_2$.  This yields the
following definition:
\begin{align*}
  \A_1\cdot \A_2\eqdef\tuple{P_1\cup P_2,\T,\iota_1},
\end{align*}
where
$\T\eqdef\paren{\T_1\setminus \pset {P_1}\times\set{\emptyset}}
\cup\T_2\cup\setcompr\trexp
{\tuple{\tuple{\trin,\emptyset},\tuple{\set{\iota_2},\trout}}\in\T_1\times\T_2}$.
As before, the following inference rules may give more intuition:
\begin{mathpar}
  \infer
  {\atrans[\tr][\A_1\cdot\A_2]S T}
  {\atrans[\tr][\A_1]S T,&T\neq\emptyset}
  \and\infer
  {\atrans[\tr][\A_1\cdot\A_2]S T}
  {\atrans[\tr][\A_2]S T}
  \and\infer
  {\atrans[\trexp][\A_1\cdot\A_2]S T}
  {\atrans[\tuple{\trin,\emptyset}][\A_1]S \emptyset,
    &\atrans[\tuple{\set{\iota_2},\trout}][\A_2]
    {\set{\iota_2}}T}
\end{mathpar}

\begin{rem}
  The inference system above is not entirely faithful to the above
  definition, in the sense that the last rule should be:
  \begin{align*}
    \infer
    {\atrans[\trexp][\A_1\cdot\A_2]S {S'\cup T}}
    {\atrans[\tuple{\trin,\emptyset}][\A_1]S {S'},
      &\atrans[\tuple{\set{\iota_2},\trout}][\A_2]
      {\set{\iota_2}}T}.
  \end{align*}
  However, the application of this rule in the case where
  $S'\neq\emptyset$ cannot yield an accepting run without $\A_1$
  violating Constraint~\ref{cstr:sp}. Indeed if there was an accepting
  run in $\A_1\cdot\A_2$ using this rule, this run would need at least
  another occurrence of the same rule to consume the remaining places
  in $P_1$, thus allowing one to build an accepting run in $\A_1$ with
  several final transitions. This contradicts
  Lemma~\ref{lem:cst2:finaltrans} by violating
  Constraint~\ref{cstr:sp}. The same phenomenon occurs below in the
  constructions we give for $\A^+$ and $\A_1\cap\A_2$, thus explaining
  the slight discrepancies between the inference systems we give to
  drive the intuition, and the formal definitions.
\end{rem}

\begin{lem}\label{lem:ok_seq}
  $\Gr{\A_1\cdot\A_2}=\Gr{\A_1}\cdot\Gr{\A_2}$.
\end{lem}
\begin{proof}
  Let $G\in\Gr{\A_1}\cdot\Gr{\A_2}$. There are two accepting runs
  $R_j=\tuple{\set{\iota_j},\tr[0][j];\cdots;\tr[n_j][j],\emptyset}$
  (with $j\in\set{1,2}$) such that $G=\Gr{R_1}\cdot\Gr{R_2}$.
  Let us name some intermediary states: we call $S_1$ and $S_2$ the
  states such that:
  \[\begin{array}{l@{\quad\quad}l}
    \atrans[\tr[0][1];\dots;\tr[n_1-1][1]][\A_1]{\set{\iota_1}}{S_1};
    &\atrans[\tr[n_1][1]][\A_1]{S_1}{\emptyset};\\
    &\\[-0.5em]
    \atrans[\tr[0][2]][\A_2]{\set{\iota_2}}{S_2};
    &\atrans[\tr[1][2];\dots;\tr[n_2][2]][\A_2]{S_2}\emptyset.
    \end{array}
  \]
  Because of Lemma~\ref{lem:cst2:finaltrans}, we know that in $R_1$, only
  $\tr[n_1][1]$ is a final transition. Using the inference rules
  above, we obtain:
   \[\begin{array}{l@{\quad\quad}l}
    \atrans[\tr[0][1];\dots;\tr[n_1-1][1]][\A_1\cdot\A_2]
      {\set{\iota_1}}{S_1};
      &\atrans[\tr[1][2];\dots;\tr[n_2][2]][\A_1\cdot\A_2]
        {S_2}\emptyset;\\
      &\\[-0.5em]
      \atrans[\tuple{\trin[n_1][1],\trout[0][2]}][\A_1\cdot\A_2]
      {S_1}{S_2}.&
    \end{array}
  \]
  Thus we obtain $R=\tuple{\set{\iota_1},\tr[0][1];\dots;\tr[n_1-1][1];
    \tuple{\trin[n_1][1],\trout[0][2]};
    \tr[1][2];\dots;\tr[n_2][2],\emptyset}$ which is accepting in
  $\A_1\cdot\A_2$ and satisfies
  $\Gr R=\Gr{R_1}\cdot\Gr{R_2}= G$.
  
  For the other direction, let $G\in\Gr{\A_1\cdot\A_2}$ be a graph,
  and $R=\tuple{\set{\iota_1},\tr[0];\dots;\tr[n],\emptyset}$ be the
  corresponding accepting run. We can extract from this run two runs
  $R_1$ and $R_2$ respectively in $\A_1$ and $\A_2$ as follows:
  \begin{align*}
    R_1^0\eqdef&\epsilon&
    R_1^{k+1} \eqdef&\left\{
                 \begin{array}{ll}
                   R_1^k;\tr[k]&\paren{\text{if }\tr[k]\in\T_1}\\
                   R_1^k;\tuple{\trin[k],\emptyset}
                               &\paren{\text{if }\tuple{\tuple{\trin[k],\emptyset},
                                 \tuple{\set{\iota_2},\trout[k]}}
                                 \in\T_1\times\T_2}\\
                   R_1^k&\paren{\text{otherwise}}
                 \end{array}\right.\\
    R_2^0\eqdef&\epsilon&
    R_2^{k+1} \eqdef&\left\{
                 \begin{array}{ll}
                   R_2^k;\tr[k]&\paren{\text{if }\tr[k]\in\T_2}\\
                   R_2^k;\tuple{\set{\iota_2},\trout[k]}
                               &\paren{\text{if }\tuple{\tuple{\trin[k],\emptyset},
                                 \tuple{\set{\iota_2},\trout[k]}}
                                 \in\T_1\times\T_2}\\
                   R_2^k&\paren{\text{otherwise}}
                 \end{array}\right.\\
    R_1\eqdef&R_1^{n+1}&R_2\eqdef&R_2^{n+1}
  \end{align*}
  Using these definitions, we may prove that $\forall k\leqslant n$,
  $\atrans[\tr[0];\dots;\tr[k]][\A_1\cdot\A_2] {\set{\iota_1}} S$
  entails
  $\atrans[R_1^{k+1}][\A_1] {\set{\iota_1}} {\paren{S\setminus P_2}}$.
  This means $R_1$ is an accepting run in
  $\A_1$. By Lemma~\ref{lem:cst2:finaltrans} this entails there is a single
  final transition in $R_1$, hence a single transition $\tr[j]$ in $R$
  such that $\tuple{{\tuple{\trin[j],\emptyset},
    \tuple{\set{\iota_2},\trout[j]}}\in\T_1\times\T_2}$. 
  Again by Lemma~\ref{lem:cst2:finaltrans} we know that $\tr[j]$ marks the
  end of $R_1$, and clearly $R_2$ cannot begin before $\tr[j]$ has
  been fired (no place in $P_2$ can appear before that). From this we
  can deduce that $R$ has almost the shape $R_1;R_2$ (but with the two
  transitions in the middle merged into $\tr[j]$). We can also check
  that $\forall k\geqslant j$,
  $\atrans[\tr[0];\dots;\tr[k]][\A_1\cdot\A_2] {\set{\iota_1}} S$
  entails
  $\atrans[R_2^{k+1}][\A_2] {\set{\iota_2}} {\paren{S\setminus P_1}}$,
  thus proving that $R_2$ is an accepting run in $\A_2$. Finally we
  get that
  $\Gr R =\Gr{R_1}\cdot\Gr{R_2}\in\Gr{\A_1}\cdot\Gr{\A_2}$.
\end{proof}

\subsubsection*{Strict iteration of a Petri automata.}

Then strict iteration of a Petri automaton~$\A=\tuple{P,\T,\iota}$ can
be done by using the previous construction on the automaton itself: we
keep the places and transitions of the automaton, but simply add a
transition~$\trexp$ for every pair of an initial
transition~$\tuple{\set{\iota},\trout}$ and a final
transition~$\tuple{\trin,\emptyset}$ in~$\T$.
\begin{align*}
\A^+\eqdef
  \tuple{P,\T\cup 
    \setcompr\trexp
    {\tuple{\tuple{\trin,\emptyset},
      \tuple{\set{\iota},\trout}}\in\T\times\T},
    \iota}.
\end{align*}
As an inference system, we get:
\begin{mathpar}
  \infer{\atrans[\tr][\A^+]S T}
  {\atrans[\tr][\A]S T}
  \and
  \infer{\atrans[\tr][\A^+]S T}
  {\atrans[\tuple{\trin,\emptyset}][\A]S \emptyset,
    &\atrans[\tuple{\set{\iota},\trout}][\A]{\set{\iota}} T}
\end{mathpar}

\begin{lem}\label{lem:ok_it}
  We have $\Gr{\A^+}=\Gr{\A}^+$.
\end{lem}
\begin{proof}
  The proof follows the same scheme as the previous one.
\end{proof}

\subsubsection*{Parallel composition of Petri automata.}

Let $\A_1=\tuple{P_1,\T_1,\iota_1}$ and
$\A_2=\tuple{P_2,\T_2,\iota_2}$ be two Petri automata with disjoint
sets of places. To obtain a an automaton for $\Gr{\A_1}\cap\Gr{\A_2}$,
we merge the initial transitions of the two automata, and then we
merge their final transitions. This yields the following automaton:
\begin{align*}
  \A_1\cap \A_2\eqdef
  \tuple{P_1\cup P_2\cup\set\iota,
    \paren{\T_1\setminus\pset {P_1}\times\set{\emptyset}}\cup
    \paren{\T_1\setminus \pset {P_2}\times\set{\emptyset}}\cup
  \T^i\cup\T^f,\iota},
\end{align*}
\begin{align*}
\text{where }&\T^i\eqdef\setcompr{\tuple{\set\iota,\trout[1]\cup\trout[2]}}
{\tuple{\tuple{\set{\iota_1},\trout[1]},
               \tuple{\set{\iota_2},\trout[2]}}\in\T_1\times\T_2},\\
\text{and }&\T^f\eqdef\setcompr{\tuple{\trin[1]\cup\trin[2],\emptyset}}
{\tuple{\tuple{\trin[1],\emptyset},\tuple{\trin[2],\emptyset}
  }\in\T_1\times\T_2}.
\end{align*}
Notice that because $P_1$ and $P_2$ are of empty intersection, the set
of states of this automaton (\ie,  $\pset{P_1\cup P_2\cup\set\iota}$)
is isomorphic to $\pset{P_1}\times\pset{P_2}\times
\pset{\set\iota}$.
For clarity, we use the later notation in the sequel.
\begin{mathpar}
    \infer{\atrans[\tr][\A_1\cap\A_2]{\tuple{S,S',\emptyset}}
    {\tuple{T,S',\emptyset}}}
    {\atrans[\tr][\A_1]S T,
    &\trout\neq\emptyset}
    \and
    \infer{\atrans[\tr][\A_1\cap\A_2]{\tuple{S',S,\emptyset}}
      {\tuple{S',T,\emptyset}}}
      {\atrans[\tr][\A_2]S T,
        &\trout\neq\emptyset}
    \and
    \infer{\atrans[\tr][\A_1\cap\A_2]
    {\tuple{\emptyset,\emptyset,\set\iota}} {\tuple{S_1,S_2,\emptyset}}}
    {\atrans[\tuple{\set{\iota_1},
    \trout\setminus\paren{Y\times P_2}}][\A_1]
    {\set{\iota_1}} {S_1},                          
    &\atrans[\tuple{\set{\iota_2},
      \trout\setminus\paren{Y\times P_1}}]
      [\A_2]{\set{\iota_2}} {S_2}}
    \and
    \infer{\atrans[\tr][\A_1\cap\A_2] {\tuple{S_1,S_2,\emptyset}}
    {\tuple{\emptyset,\emptyset,\emptyset}}}
    {\atrans[\tuple{\trin\setminus P_2,\emptyset}][\A_1]
    {S_1}\emptyset,
    &\atrans[\tuple{\trin\setminus P_1,\emptyset}]
      [\A_2] {S_2} \emptyset,
    &\trout=\emptyset} 
\end{mathpar}

\begin{lem}\label{lem:ok_par}
  We have $\Gr{\A_1\cap\A_2}=\Gr{\A_1}\cap\Gr{\A_2}$.
\end{lem}
\begin{proof}
  Let $G\in\Gr{\A_1}\cap\Gr{\A_2}$. There are two accepting runs in
  $\A_1$ and $\A_2$, that we call
  $R_i=\tuple{\set{\iota_i},\tr[0][i];\dots;\tr[n_i][i],\emptyset}$
  (with $i\in\set{1,2}$), such that $G=\Gr{R_1}\cap\Gr{R_2}$.
  We build the following run:
  \[R\eqdef
    \tuple{\set\iota,\tuple{\set\iota,\trout[0][1]\cup\trout[0][2]};
      \tr[1][1];\dots;\tr[n_1-1][1];\tr[1][2];\dots;\tr[n_2-1][2];
      \tuple{\trin[n_1][1]\cup\trin[n_2][2],\emptyset},
      \emptyset}.\]
  It is a simple exercise to check that indeed $R$ is an accepting run
  in $\A_1\cap\A_2$, and that its trace does satisfy
  ${\Gr R= G}$.

  For the other direction, the presentation as an inference system
  simplifies the reasoning: by projecting an accepting run in
  $\A_1\cap\A_2$ on the first (respectively second) component, we get
  an accepting run in $\A_1$ (resp. $\A_2$). From that remark, one can
  deduce that the parallel product of the traces of these two
  projected runs is isomorphic to the trace of the whole run.
\end{proof}

\subsubsection*{Conclusion.}
Now we have all ingredients required to associate a Petri automaton
$\A(e)$ with every expression $e$, by induction on $e$:
\begin{align*}
  \A\paren{e \sumtm f}&\eqdef \A\paren e\cup\A\paren f &
  \A\paren 0&\eqdef\underline 0 \\
  \A\paren{e\cdot f}&\eqdef \A\paren e\cdot\A\paren f &
  \A\paren a&\eqdef\underline a \\
  \A\paren{e\cap f}&\eqdef \A\paren e\cap\A\paren f &
  \A\paren{e^+}&\eqdef \A\paren e^+
\end{align*}
Lemmas~\ref{lem:ok_union}, \ref{lem:ok_seq}, \ref{lem:ok_it}
and~\ref{lem:ok_par} show that this construction is correct, whence
the first half of the Kleene theorem for Petri automata.
\begin{prop}
  \label{prop:ok_e_to_a}
  For every simple expression $e\in\SExp[Y]$, we have
  $\Gr{\A\paren e}=\Gr e$.
\end{prop}

\begin{center}
  \begin{tikzcd}
    \SExp[Y]\arrow[rr,"\A"] \arrow[dr,"\G"']&& \PA[Y]\arrow[dl,"\G"]\\ &\pset{\SP[Y]}
  \end{tikzcd}
\end{center}

\begin{thm}
  \label{thm:kl1:sp}
 Regular sets of series-parallel graphs are recognisable.
\end{thm}

\noindent When a Petri automaton is labelled in $\Xb$, it can be used
to produce graphs labelled in $X$ rather than series-parallel graphs
labelled in $\Xb$, using the function $\retype\argument$. Thus we define:
\begin{defi}
  A set $S\subseteq\Gph$ of graphs is \emph{recognisable} if there
  exists a Petri automaton $\A\in\PA[\Xb]$ such that $S=\retype{\Gr\A}$.
\end{defi}

\noindent We can combine Propositions~\ref{prop:type:graphs}
and~\ref{prop:ok_e_to_a} to obtain the following commutative diagram.
\begin{center}
  \begin{tikzcd}
    \Exp\arrow[rr,"\detype\argument"]\arrow[ddrr,"\G"']&&
    \SExp\arrow[rr,"\A"] \arrow[dr,"\G"']&&
    \PA[\Xb]\arrow[dl,"\G"]\\ 
    &&&\pset{\SP}\arrow[dl,"\retype\argument"]\\
    &&\pset{\Gph}
  \end{tikzcd}
\end{center}
\noindent Whence
\begin{thm}
  \label{thm:kl1}
  Regular sets of graphs are recognisable.
\end{thm}

\begin{rem}
  If $e$ is an expression without intersection, then the transitions
  in $\A(\detype e)$ are all of the form
  $\tuple{\set p,\set{\tuple{x,q}}}$, with only one input, one output
  and a label in $X\cup\set1$. As a consequence, the accessible
  configurations are singletons, and the resulting Petri automaton has
  the structure of a non-deterministic finite-state automaton with
  epsilon transitions (NFA). Actually, in that case, the construction
  we described above is just a variation on Thompson's
  construction~\cite{thompson68}.
\end{rem}

\section{From Petri automata to expressions}
\label{sec:pa:exp}

\noindent Now we prove the converse implications of
Theorems~\ref{thm:kl1:sp} and~\ref{thm:kl1}, thus resulting in full
Kleene theorems for Petri automata and expressions. Like previously,
most of the work is done with simple expressions and series-parallel
graphs; the extension to arbitrary expressions and graphs will follow
from Proposition~\ref{prop:type:graphs}. 

This section is technically more involved than the other parts of the
paper. The following sections (\ref{sec:reading} and beyond) do not
depend on the results and notions introduced here.

\subsection{More notions on graphs}

Given an oriented (but not necessarily labelled and pointed) graph
$G$, its \emph{sources} $\min G$ (resp.\ \emph{sinks} $\max G$) are
the vertices with no incoming (resp. outgoing) edge. A vertex $v$ is
\emph{reachable} (respectively \emph{co-reachable}) from another
vertex $v'$ if there is a path from $v'$ to $v$ (resp.\ from $v$ to
$v'$). A graph is \emph{connected} if there is a non-directed path
between any two vertices. A subgraph is called a \emph{connected
  component} if it is maximal amongst the connected subgraphs.

A \emph{(rooted) tree} is an acyclic graph $T$ such that either
$\min T$ or $\max T$ contains a single node, called the root, and for
any two nodes $x,y\in V$ there exists at most one path from $x$ to
$y$. If $\min T$ is a singleton, then $T$ is a \emph{top-down tree},
otherwise it is a \emph{bottom-up tree}. A tree is a \emph{proper
  tree} if its root has degree one and if it does not contain a node
with exactly one incoming edge and one outgoing edge. There is only a
finite number of proper trees with leaves chosen from a finite set.

The following technical lemma will be instrumental later on. We
illustrate it in Figure~\ref{fig:gluecenter}. 
\begin{lem}\label{lem:gluecenter}
  Let $T=\tuple{V_T,E_T}$ be a proper unlabelled top down tree with
  root $r$ and set of leaves $F\subseteq V_T$, and $G=\tuple{V,E}$ be
  a connected DAG. Let $\phi:F\rightharpoonup V$ be a partial function
  defined on $F'\subseteq F$. The \emph{gluing} of $T$ and $G$ along
  $\phi$ is the graph
  \begin{align*}
    T\cdot_\phi G\eqdef
    \tuple{V_T\cup V,E_T\cup E\cup\set{\tuple{f,\phi\paren f}}\mid f\in F'}.
  \end{align*}
  \noindent
  If $T\cdot_\phi G\spred[\star] T'$ and if $T'$ is a tree, then there
  is a node $c$ in $G$ such that for every
  $\tuple{f,M}\in F'\times\max G$ every path from $\phi\paren f$ to
  $M$ in $G$ visits $c$.
\end{lem}
\begin{proof}
  The maximal elements of $G$ remain in $T'$ as leaves, and because
  $G$ is connected there will be a subtree of $T'$ whose leaves are
  exactly $\max G$. The root of this tree must be a node accessible
  from $F'$, hence coming from $G$. Because paths are preserved during
  SP-rewriting, this vertex has the desired property.
\end{proof}
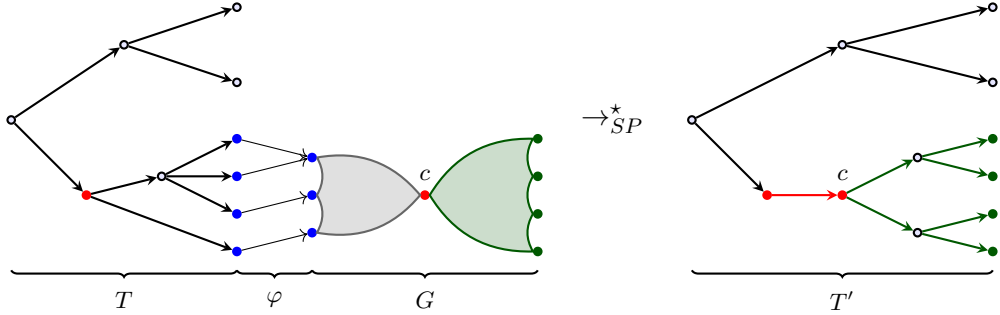
\begin{figure}[t]
  \centering
  \begin{tikzpicture}[baseline=(0.south)]
    \position (0) (0,0);
    \position (1) (1.5,1);
    \node[state,red] (2) at (1,-1) {};
    \position (3) (3,1.5);
    \position (4) (3,0.5);
    \position (5) (2,-.75);
    \node[state,blue] (6) at (3,-.75) {};
    \node[state,blue] (7) at (3,-.25) {};
    \node[state,blue] (8) at (3,-1.25) {};
    \node[state,blue] (9) at (3,-1.75) {};
    \edge (0)(1);\edge(0)(2);
    \edge (1)(3);\edge(1)(4);
    \edge (2)(5);\edge(2)(9);
    \edge (5)(6);\edge(5)(7);\edge(5)(8);
    
    \node[state,blue] (a) at (4,-.5) {};
    \node[state,blue] (b) at (4,-1) {};
    \node[state,blue] (c) at (4,-1.5) {};
    \node[state,red] (C) at (5.5,-1)[label=above:$c$] {};
    \node[state,dkgreen] (x) at (7,-.25) {};
    \node[state,dkgreen] (y) at (7,-.75) {};
    \node[state,dkgreen] (z) at (7,-1.25) {};
    \node[state,dkgreen] (t) at (7,-1.75) {};
    \draw[dimgray,fill,fill opacity=0.2,thick] 
    (a.east)
    to[bend left] (b.east)
    to[bend left] (c.east)
    to[bend right] (C.west)
    to[bend right] (a.east);
    \draw[dkgreen,fill,fill opacity=0.2,thick]
    (x.west)
    to[bend right] (y.west)
    to[bend right] (z.west)
    to[bend right] (t.west)
    to[bend left] (C.east)
    to[bend left] (x.west);
    
    \draw[->] (6) -- (a);
    \draw[->] (7) -- (a);
    \draw[->] (8) -- (b);
    \draw[->] (9) -- (c);

    \coordinate (j) at ($(9)-(0,.25)$);
    \draw[decorate,decoration=brace,thick] (j) to node[below,yshift=-4pt] {$T$} ($(0|-j)$);
    \coordinate (j') at ($(c|-j)$);
    \draw[decorate,decoration=brace,thick] ($(j'-|t)$) to node[below,yshift=-4pt] {$G$} (j');
    \draw[decorate,decoration=brace,thick] ($(j')$) to node[below,yshift=-4pt] {$\varphi$} (j);
    
  \end{tikzpicture}
  $\quad\spred[\star]\quad$
  \begin{tikzpicture}[baseline=(0.south)]
    \position (0) (0,0);
    \position (1) (2,1);
    \node[state,red] (2) at (1,-1) {};
    \position (3) (4,1.5);
    \position (4) (4,0.5);
    \edge (0)(1);\edge(0)(2);
    \edge (1)(3);\edge(1)(4);
    
    \node[state,red] (C) at (2,-1)[label=above:$c$] {};
    
    \position (5) (3,-.5);
    \position (6) (3,-1.5);

    \node[state,dkgreen] (x) at (4,-.25) {};
    \node[state,dkgreen] (y) at (4,-.75) {};
    \node[state,dkgreen] (z) at (4,-1.25) {};
    \node[state,dkgreen] (t) at (4,-1.75) {};
    
    \draw[arc,red] (2) -- (C);
    \draw[arc,dkgreen] (C) -- (5);
    \draw[arc,dkgreen] (C) -- (6);
    \draw[arc,dkgreen] (5) -- (x);
    \draw[arc,dkgreen] (5) -- (y);
    \draw[arc,dkgreen] (6) -- (z);
    \draw[arc,dkgreen] (6) -- (t);

    \coordinate (j) at ($(t)-(0,.25)$);
    \draw[decorate,decoration=brace,thick] (j) to node[below,yshift=-4pt] {$T'$} ($(0|-j)$);
    
  \end{tikzpicture}
  \caption{Illustration of Lemma~\ref{lem:gluecenter}}
  \label{fig:gluecenter}
\end{figure}

\subsection{Boxes}
\label{sec:boxes}

In the standard Kleene theorem for words, an important aspect is that
a partial execution in an automaton always denotes a word, even if the
path is not initial and final. This is not true with graphs and Petri
automata, where partial executions denote slices of a graph. To this
end, we will use \emph{boxes}. Recall the following run, that we used
as an example of run from the automaton displayed in
Figure~\ref{fig:auto}:
\begin{center}
      \begin{tikzpicture}[yscale=.7]
      \node (A) at (0,-0.5) {$A$};
      \node (B) at (2,0) {$B$};
      \node (G1) at (2,-1) {$G$};
      \node (C) at (4,0.5) {$C$};
      \node (E1) at (4,-0.5) {$E$};
      \node (G2) at (4,-1) {$G$};
      \node (D) at (6,0.5) {$D$};
      \node (E2) at (6,-0.5) {$E$};
      \node (G3) at (6,-1) {$G$};
      \node (C1) at (8,0.5) {$C$};
      \node (E3) at (8,-0.5) {$E$};
      \node (G4) at (8,-1) {$G$};
      \node (D1) at (10,0.5) {$D$};
      \node (E4) at (10,-0.5) {$E$};
      \node (G5) at (10,-1) {$G$};
      \node (F) at (12,0) {$F$};
      \node (G6) at (12,-1) {$G$};
      \trans[0] (0) (1,-0.5);
      \trans[1] (1) (3,0);
      \trans[2] (2) (5,0.5);
      \trans[3] (3) (7,0);
      \trans[2] (4) (9,0.5);
      \trans[4] (5) (11,0);
      \transf[7] (6) ($(F)!.5!(G6)+(1,0)$);
      % \node[trans,dotted] (6) at (13,-0.5) {$6$};
      \edge (A) (0);
      \edge[above] (0) (B)[b];
      \edge[below] (0) (G1)[a];

      \edge (B) (1);
      \draw[thick,dotted] (G1) to (G2);
      \edge[above] (1) (C)[c];
      \edge[below] (1) (E1)[b];

      \edge (C) (2);
      \draw[thick,dotted] (E1) to (E2);
      \draw[thick,dotted] (G2) to (G3);
      \edge[above] (2) (D)[a];

      \edge (D) (3);
      \edge (E2) (3);
      \draw[thick,dotted] (G3) to (G4);
      \edge[above] (3) (C1)[c];
      \edge[below] (3) (E3)[b];

      \edge (C1) (4);
      \draw[thick,dotted] (E3) to (E4);
      \draw[thick,dotted] (G4) to (G5);
      \edge[above] (4) (D1)[a];

      \edge (D1) (5);
      \edge (E4) (5);
      \draw[thick,dotted] (G5) to (G6);
      \edge[above] (5) (F)[d];

      \edge (F) (6);
      \edge (G6) (6);
      % \draw[thick,dotted] (F) to (6);
      % \draw[thick,dotted] (G6) to (6);
      
    \end{tikzpicture}
\end{center}
We will abstract each transition by a box, and the run itself by the
sequential composition of those boxes:
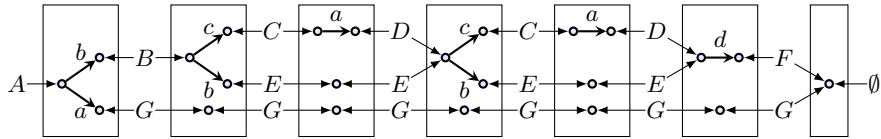
\begin{align*}
  \begin{tikzpicture}[yscale=.7]
    \coordinate (step) at (.5,0);
    \coordinate (up) at (0,.5);
    \coordinate (down) at (0,-.5);
    \coordinate (transspace) at ($(step)+(.6,0)+(.6,0)$);
    \position (0) (1,-0.5);
    \position (B)($(0)+(step)+(up)$);
    \position (G1)($(0)+(step)+(down)$);
    \rect(G1)(0)[1](B)(B);
    \position (1) ($(0)+(up)+(transspace)$);
    \position (C) ($(1)+(step)+(up)$);
    \position (E1) ($(1)+(step)+(down)$);
    \position (G2) ($(1|-E1)!.5!(E1)+(down)$);
    \rect(G2)(1)(C)(C);
    \position (2) ($(1)+(up)+(transspace)$);
    \position (D) ($(2)+(step)$);
    \position (E2) ($(2)!.5!(D)+(down)+(down)$);
    \position (G3) ($(E2)+(down)$);
    \rect(G3)(2)(D)(D);
    \position (3) ($(2)+(down)+(transspace)$);
    \position (C1) ($(3)+(step)+(up)$);
    \position (E3) ($(3)+(step)+(down)$);
    \position (G4) ($(3|-E3)!.5!(E3)+(down)$);
    \rect(G4)(3)(C1)(C1);
    \position (4) ($(3)+(up)+(transspace)$);
    \position (D1) ($(4)+(step)$);
    \position (E4) ($(4)!.5!(D1)+(down)+(down)$);
    \position (G5) ($(E4)+(down)$);
    \rect(G5)(4)(D1)(D1);
    \position (5) ($(4)+(down)+(transspace)$);
    \position (F) ($(5)+(step)$);
    \position (G6) ($(5)!.5!(F)+(down)+(down)$);
    \rect(G6)(5)[1](F)(F);
    \position (6) ($(5)+(down)+(transspace)$);
    \rect[1](6)(6)[1.5](6)(6);
    \iport(A)(0);
    \oport(B)[b](B);\oport(G)[g1](G1);
    \draw[portarc](b)--(1);\draw[portarc](g1)--(G2);
    \oport(G)[g2](G2)[C];\oport(C)[c](C);\oport(E)[e1](E1);
    \draw[portarc](c)--(2);\draw[portarc](e1)--(E2);
    \draw[portarc](g2)--(G3);
    \oport(G)[g3](G3)[D];\oport(D)[d](D);\oport(E)[e2](E2)[D];
    \draw[portarc](d)--(3);\draw[portarc](e2)--(3);
    \draw[portarc](g3)--(G4);
    \oport(G)[g4](G4)[C1];\oport(C)[c1](C1);\oport(E)[e3](E3);
    \draw[portarc](c1)--(4);\draw[portarc](e3)--(E4);
    \draw[portarc](g4)--(G5);
    \oport(G)[g5](G5)[D1];\oport(D)[d1](D1);\oport(E)[e4](E4)[D1];
    \draw[portarc](d1)--(5);\draw[portarc](e4)--(5);
    \draw[portarc](g5)--(G6);
    \oport(F)[f](F);\oport(G)[g6](G6)[F];
    \draw[portarc](f)--(6);\draw[portarc](g6)--(6);
    \oport(\emptyset)[out](6);
    \edge[above] (0) (B)[b];
    \edge[below] (0) (G1)[a];
    \edge[above] (1) (C)[c];
    \edge[below] (1) (E1)[b];
    \edge[above] (2) (D)[a];
    \edge[above] (3) (C1)[c];
    \edge[below] (3) (E3)[b];
    \edge[above] (4) (D1)[a];
    \edge[above] (5) (F)[d];
  \end{tikzpicture}
\end{align*}
These boxes allow us to represent every subrun as a single box. For
instance, the subrun firing transitions $2$ and $3$ in sequence, and
looping from state $\set{C,E,G}$ back to itself can be represented as
the box:
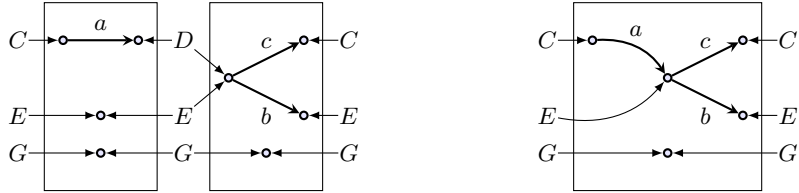
\begin{align*}
  \begin{tikzpicture}
    \coordinate (step) at (1,0);
    \coordinate (up) at (0,.5);
    \coordinate (down) at (0,-.5);
    \coordinate (transspace) at ($(step)+(.6,0)+(.6,0)$);
    \position (2) (0,0);
    \position (D) ($(2)+(step)$);
    \position (E2) ($(2)!.5!(D)+(down)+(down)$);
    \position (G3) ($(E2)+(down)$);
    \rect(G3)(2)(D)(D);
    \position (3) ($(2)+(down)+(transspace)$);
    \position (C1) ($(3)+(step)+(up)$);
    \position (E3) ($(3)+(step)+(down)$);
    \position (G4) ($(3|-E3)!.5!(E3)+(down)$);
    \rect(G4)(3)(C1)(C1);
    \iport(G)[g2](G3)[2];\iport(C)[c](2);\iport(E)[e1](E2)[2];
    \oport(G)[g3](G3)[D];\oport(D)[d](D);\oport(E)[e2](E2)[D];
    \draw[portarc](d)--(3);\draw[portarc](e2)--(3);
    \draw[portarc](g3)--(G4);
    \oport(G)[g4](G4)[C1];\oport(C)[c1](C1);\oport(E)[e3](E3);
    \edge[above] (2) (D)[a];
    \edge[above] (3) (C1)[c];
    \edge[below] (3) (E3)[b];
  \end{tikzpicture}&&
  \begin{tikzpicture}
    \coordinate (step) at (1,0);
    \coordinate (up) at (0,.5);
    \coordinate (down) at (0,-.5);
    \position (2) (0,0);
%    \position (D) ($(2)+(step)$);
%    \position (E2) ($(2)!.5!(D)+(down)+(down)$);
%    \position (G3) ($(E2)+(down)$);
    \position (3) ($(2)+(down)+(step)$);
    \position (C1) ($(3)+(step)+(up)$);
    \position (E3) ($(3)+(step)+(down)$);
    \position (G4) ($(2|-E3)!.5!(E3)+(down)$);
    \rect(G4)(2)(C1)(C1);
    \iport(G)[g2](G4)[2];\iport(C)[c](2);
    \node[portnode] (e) at ($(g2)+(up)$) {$E$};
    \draw[portarc,bend right](e)to(3);
    % \iport(E)[e1](E2)[2];
    % \oport(G)[g3](G3)[D];\oport(D)[d](D);\oport(E)[e2](E2)[D];
    % \draw[portarc](d)--(3);\draw[portarc](e2)--(3);
    % \draw[portarc](g3)--(G4);
    \oport(G)[g4](G4)[C1];\oport(C)[c1](C1);\oport(E)[e3](E3);
    \edge[above,bend left] (2) (3)[a];
    \edge[above] (3) (C1)[c];
    \edge[below] (3) (E3)[b];
  \end{tikzpicture}
\end{align*}

\medskip\noindent We start by defining those boxes formally, and we
study some of their properties.

\subsubsection{The category of boxes.}
\label{sec:categories-boxes}

We fix a finite set $P$ of ports, which will be instantiated in the
next section with the set of places of some Petri automaton.
\begin{defi}[Box]
  Let $S,{S'}\subseteq P$ be two sets of ports. A \emph{box} labelled
  over $Y$ from $S$ to ${S'}$ is a triple $\tuple{\pin,G,\pout}$ where
  $G$ is a directed acyclic graph (DAG) with edges labelled in $Y$,
  $\pin$ is a map from $S$ to the vertices of $G$, and $\pout$ is a
  bijective map from ${S'}$ to $\max G$.
  The set of boxes labelled by $Y$ is written $\boxset[Y]$, and the
  boxes from $S$ to ${S'}$ are denoted by $\boxsetb[Y] S {S'}$, or
  simply $\boxsetb S {S'}$ if the set of labels is clear from the
  context.
\end{defi}
\noindent
As for graphs, we consider boxes up to renaming of internal nodes. We
represent boxes graphically as in Figure~\ref{fig:ex-comp}. Inside the
rectangle is the DAG, with the input ports on the left-hand side and
the output ports on the right-hand side. The maps $\pin$ and $\pout$
are represented by the arrows going from the ports to vertices inside
the rectangle.
\begin{figure}[t]
  \begin{minipage}{.65\textwidth}
    \centering
    \begin{minipage}{.45\linewidth}
      \centering
      \begin{tikzpicture}
        \state (a0) (0,2);
        \state (a1) (1,2);
        \state (a2) (0,1);
        \state (a3) (0,0);
        \state (a4) (2,2.5);
        \state (a5) (2,1.5);
        \state (a6) (2,0.5);
        \edge (a0) (a1)[a];
        \edge (a1) (a4)[b];
        \edge[below] (a1) (a5)[c];
        \edge (a2) (a6)[c];
        \edge[below] (a3) (a6)[a];
        \clbox (a3) (a4);
        \iport[.5] (A) (a0);
        \iport[-.5] (B) (a0);
        \iport (C) (a2);\iport (D) (a3);
        \oport (E) (a4);\oport (F) (a5);\oport (G) (a6);
      \end{tikzpicture}
    \end{minipage}
    \begin{minipage}{.45\linewidth}
      \centering
      \begin{tikzpicture}
        \state (b1) (3,0.5);
        \state (b2) (4,0);
        \state (b3) (4,1);
        \state (b4) (3,2);
        \state (b5) (4,2);
        \edge[below] (b1) (b2)[b];
        \edge (b1) (b3)[a];
        \edge (b4) (b5)[b];
        \clbox[.75] (b1)[.75](b5);
        \iport[.5] (E) (b4);
        \iport[-.5] (F) (b4);
        \iport (G) (b1);
        \oport (A) (b5);
        \oport (D) (b3);
        \oport (C) (b2);
      \end{tikzpicture}    
    \end{minipage}
    \caption{Examples of boxes}
    \label{fig:ex-comp}
  \end{minipage}
  \begin{minipage}{.3\linewidth}
    \centering
    \begin{tikzpicture}
      \state (v1) (1,2.5);\state (v2) (1,2);\state (v3) (1,1.5);
      \state (v4) (1,1);\state (v5) (1,0.5);
      \rect (v5) (v5) (v1) (v1);
      \iport (A) (v1);\oport (A) (v1);\iport (B) (v2);\oport (B)(v2);
      \iport (C) (v3);\oport (C) (v3);\iport (D) (v4);\oport (D)(v4);
      \iport (E) (v5);\oport (E) (v5);
    \end{tikzpicture}
    \caption{$\bid\sigma$.}
    \label{fig:idsig}
  \end{minipage}
\end{figure}

% 

% For any set $S$, we define the identity box on $S$.
\begin{defi}[Identity box]
  If $S=\set{p_1,\dots,p_n}\subseteq P$ is a set of ports, the
  \emph{identity box} on $S$ is defined as
  $\bid S \eqdef \tuple{[p_i\mapsto i],
    \tuple{\set{1,\dots,n},\emptyset}, [p_i\mapsto i]}.$
\end{defi}
\noindent
For instance, the box $\bid{\set{A,B,C,D,E}}$ is represented in
Figure~\ref{fig:idsig}.

Now we define how to compose boxes sequentially. Intuitively, if the
set $S$ of output ports of $\beta_1$ is equal to the set of input
ports of $\beta_2$, we may compose them by putting the graph of
$\beta_1$ to the left of the graph of $\beta_2$, and for every port
$p\in S$, we identify the node $\pout[1]\paren p$ with the node
$\pin[2]\paren p$.
\begin{defi}[Composition of boxes] Let $S_1,S_2,S_3\subseteq P$ and
  for $i\in\set{1,2}$, let $\beta_i$ be a box from $S_i$ to $S_{i+1}$,
  with $\beta_i=\tuple{\pin[i],\tuple{V_i,E_i},\pout[i]}$, such that
  $V_1\cap V_2=\emptyset$. The \emph{composition} of $\beta_1$ and
  $\beta_2$, written $\beta_1\odot\beta_2$ is defined as
  $\tuple{\pin,\tuple{V_1'\cup V_2,E_1'\cup E\cup E_2},\pout}$ with:
  \begin{itemize}
  \item $V_1'\eqdef V_1\setminus \pout[1]\paren{S_2}$, and
    $E_1'\eqdef E_1\cap\paren{V_1'\times Y\times V_1'}$;
  \item
    $E\eqdef\setcompr{\tuple{x,a,\pin[2]\paren{p}}}
    {\tuple{x,a,y}\in
      E_1,\ y=\pout[1](p)}$;
  % \setcompr{(x,a,y)\in E_1}{y\notin\max\tuple{V_1,E_1}}
  \item $\pout\eqdef\pout[2]$ and $\pin(p)\eqdef\left\{
      \begin{array}{l@{,\quad}l}
        \pin[2](q)&\text{if }\pin[1](p)=\pout[1](q),\\
        \pin[1](p)&\text{otherwise}.
      \end{array}\right.$\qedhere
  \end{itemize}
\end{defi}

\noindent
For instance the two boxes in Figure~\ref{fig:ex-comp} can be
composed, and yield the following box:
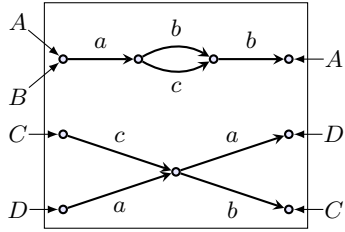
\begin{align*}
  \begin{tikzpicture}
    \state (a0) (0,2);
    \state (a1) (1,2);
    \state (a2) (0,1);
    \state (a3) (0,0);
    \state (b1) (1.5,0.5);
    \state (b2) (3,0);
    \state (b3) (3,1);
    \state (b4) (2,2);
    \state (b5) (3,2);
    \edge (a0) (a1)[a];
    \edge[above,bend left] (a1) (b4)[b];
    \edge[below,bend right] (a1) (b4)[c];
    \edge (a2) (b1)[c];
    \edge[below] (a3) (b1)[a];
    \edge[below] (b1) (b2)[b];
    \edge (b1) (b3)[a];
    \edge (b4) (b5)[b];
    \clbox (a3)[.75] (b5);
    \iport[.5] (A) (a0);\iport[-.5] (B) (a0);
    \iport (C) (a2);\iport (D) (a3);
    \oport (A) (b5);\oport (D) (b3);\oport (C) (b2);
  \end{tikzpicture}    
\end{align*}
We obtain a category of sets of ports and boxes:
\begin{prop}\label{prop:catbox}
  We have a category $\catbox $ whose objects are subsets of $P$,
  and whose morphisms between $S$ and ${S'}$ are the boxes from $S$ to
  $S'$, \ie, $\boxsetb[Y]S {S'}$.
\end{prop}
\begin{proof}
  To give a high level proof of the associativity of box composition,
  notice that the computation of $\beta\odot\gamma$ may be split in
  two steps:
  \begin{enumerate}
  \item compute an intermediate box $\beta\xrightarrow\epsilon\gamma$
    (with $\epsilon$ being a fresh label), build by keeping the input
    port map of $\beta$, the output map of $\gamma$, the disjoint
    union of their vertices and edges, and simply adding edges
    $x\xrightarrow\epsilon y$ whenever
    $\tuple{x,y}\in\max \beta\times\gamma$ such that
    $\pout[\beta](x)=\pin[\gamma](y)$.
  \item collapse all edges labelled by $\epsilon$ by identifying their
    source and target vertices. The effect of this will be the
    destruction of all the nodes from $\max\beta$, and the redirection
    of their incoming arrows to the corresponding vertices in
    $\gamma$.
  \end{enumerate}
  Suppose we have three boxes $\beta,\gamma,\delta$ with the
  appropriate input and output sets of ports (the input ports of
  $\gamma$ should be exactly the output ports of $\beta$...). We may
  now describe the two ways in which to compose them by the following
  diagram:
  \begin{center}
    \begin{tikzcd}
      \beta,\gamma,\delta\ar[r]\ar[d]\ar[rd,phantom,"(1)" description]&
      \beta,\gamma\xrightarrow{\epsilon_2}\delta
      \ar[r]\ar[d]\ar[rd,phantom,"(2)" description]&
      \beta,\gamma\odot\delta\ar[d]\\
      \beta\xrightarrow{\epsilon_1}\gamma,\delta
      \ar[r]\ar[d]\ar[rd,phantom,"(3)" description]&
      \beta\xrightarrow{\epsilon_1}\gamma\xrightarrow{\epsilon_2}\delta
      \ar[r]\ar[d]\ar[rd,phantom,"(4)" description]&
      \beta\xrightarrow{\epsilon_1}\gamma\odot\delta\ar[d]\\
      \beta\odot\gamma,\delta\ar[r]&
      \beta\odot\gamma\xrightarrow{\epsilon_2}\delta\ar[r]&
      \beta\odot\gamma\odot\delta
    \end{tikzcd}
  \end{center}
  The commutation of square $(1)$ is clear: the box
  $\beta\xrightarrow{\epsilon_1}\gamma\xrightarrow{\epsilon_2}\delta$
  could even be described in one step (just put edges with
  $\epsilon_1$ between $\max \beta$ and $\gamma$ and with
  $\epsilon_2$ between $\max\gamma$ and $\delta$.

  Squares $(2)$ and $(3)$ commute as well. For instance for square
  $(3)$, notice that collapsing edges $\epsilon_1$ in
  $\beta\xrightarrow{\epsilon_1}\gamma$ doesn't affect the set
  $\max\beta\xrightarrow{\epsilon_1}\gamma=\max\gamma$. Hence this has
  no bearing on the computation of
  $\beta\odot\gamma\xrightarrow{\epsilon_2}\delta$.

  Finally, for square $(4)$, an additional step could be added:
  replacing both labels $\epsilon_1$ and $\epsilon_2$ by a common
  label $\epsilon$, before collapsing the resulting graph along
  $\epsilon$-edges. This clearly produces the same result. The point
  is then to show that the collapsing operation is confluent. This is
  true, because one could compute the normal form from the beginning,
  as the set of equivalence classes of the smallest equivalence
  relation on vertices containing all pairs $x,y$ of vertices linked
  by an $\epsilon$-edge.

  The fact that for every box
  $\beta=\tuple{\pin,\tuple{V,E},\pout}\in\boxsetb A B$, we have
  $\bid A\odot \beta=\beta$ is straightforward from the
  definitions. % Notice that $\min\bid A=\max\bid A$ is the set of all
  % the vertices of the graph of $\bid A$, and is equivalent to $A$
  % itself (both $\pin$ and $\pout$ are bijections between $A$ and this
  % set).
  The vertices of the composite box are exactly those of $\beta$,
  because the image of $A$ through the output map of $\bid A$ is the
  set of vertices of $\bid A$ itself. As $\bid A$ has no edge, the
  edges of $\bid A\odot \beta$ are again simply those of $\beta$. By
  definition the output map is that of $\beta$, but so is the input
  map, because for every port $p_i\in A$, the images of $p_i$ via the
  input and output maps of $\bid A$ are both equal to $i$.  The fact
  that $\beta\odot\bid B$ is also equal to $\beta$ follows from a
  similar argument.
\end{proof}

\subsubsection{The category of typed boxes.}
\label{sec:categories-typed-boxes}

Since we want to represent slices of series-parallel graphs only, we
actually need to enforce a stronger typing discipline on boxes. Indeed
there are boxes in the above category that are slices of series
parallel graphs, but whose composition cannot be a part of a
series-parallel graph. For instance consider the four boxes below,
called $\beta_1$ through $\beta_4$:
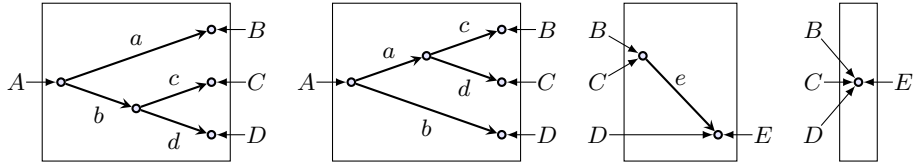
\begin{equation*}
\begin{tikzpicture}[baseline=(pCo.south),yscale=.7]
  \position (0) (0,1);
  \position (1) (2,2);
  \position (2) (1,.5);
  \position (3) (2,1);
  \position (4) (2,0);
  \edge (0) (1)[a];\edge[below](0)(2)[b];\edge(2)(3)[c];\edge[below](2)(4)[d];
  \rect(4)(0)(1)(1);
  \iport(A)(0);\oport(B)(1);\oport(C)(3);\oport(D)(4);
\end{tikzpicture}~~~
\begin{tikzpicture}[baseline=(pCo.south),yscale=.7]
  \position (0) (0,1);
  \position (1) (1,1.5);
  \position (2) (2,0);
  \position (3) (2,2);
  \position (4) (2,1);
  \edge (0) (1)[a];\edge[below](0)(2)[b];\edge(1)(3)[c];\edge[below](1)(4)[d];
  \rect(2)(0)(3)(3);
  \iport(A)(0);\oport(B)(3);\oport(C)(4);\oport(D)(2);
\end{tikzpicture}~~~
\begin{tikzpicture}[baseline=(pCi.south),yscale=.7]
  \position (0) (0,1.5);
  \position (1) (1,0);
  \edge (0) (1)[e];
  \rect(1)(0)[1](0)(1);
  \iport[.5](B)(0);\iport[-.5](C)(0);\iport(D)(1)[0];\oport(E)(1);
\end{tikzpicture}~~~
\begin{tikzpicture}[baseline=(pCi.south),yscale=.7]
  \position (0) (0,0);
  \rect[1.5](0)(0)[1.5](0)(0);
  \iport[1](B)(0);\iport(C)(0);\iport[-1](D)(0);\oport(E)(0);
\end{tikzpicture}
\end{equation*}
\noindent
According to their interfaces, $\beta_1$ and $\beta_2$ can both
compose with $\beta_3$ and $\beta_4$. Indeed $\beta_1\odot\beta_4$,
$\beta_2\odot\beta_3$ and $\beta_2\odot\beta_4$ all yield boxes
containing series-parallel graphs. However the composition
$\beta_1\odot\beta_3$ produces the box shown in Figure~\ref{fig:bad}.
\begin{figure}[t]
  \centering
  $\begin{tikzpicture}[baseline=(pCo.south),yscale=.7]
    \position (0) (0,1);
    \position (1) (2,2);
    \position (2) (1,.5);
    \position (3) (2,1);
    \position (4) (2,0);
    \edge (0) (1)[a];\edge[below](0)(2)[b];\edge(2)(3)[c];\edge[below](2)(4)[d];
    \rect(4)(0)(1)(1);
    \iport(A)(0);\oport(B)(1);\oport(C)(3);\oport(D)(4);
  \end{tikzpicture}\odot
\begin{tikzpicture}[baseline=(pCi.south),yscale=.7]
  \position (0) (0,1.5);
  \position (1) (1,0);
  \edge (0) (1)[e];
  \rect(1)(0)[1](0)(1);
  \iport[.5](B)(0);\iport[-.5](C)(0);\iport(D)(1)[0];\oport(E)(1);
\end{tikzpicture}=
  \begin{tikzpicture}[baseline=(pAi.south),yscale=.7]
    \position (0) (0,1); \position (1) (1,2); \position (2) (1,0);
    \position (4) (2,1); \edge (0)(1)[a];\edge[below](0)(2)[b];
    \edge[right](2)(1)[c];\edge[below](2)(4)[d];\edge (1)(4)[e]; 
    \rect(2)(0)(1)(4); \iport(A)(0);\oport(E)(4);
  \end{tikzpicture}$
  \caption{An example of ``bad'' composition}
  \label{fig:bad}
\end{figure}
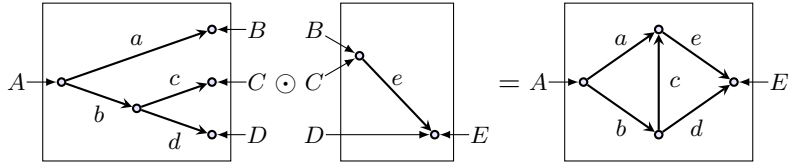
The graph in this box is not series-parallel, and is in fact the
forbidden subgraph of the class of series-parallel
graphs~\cite{Valdes79}. To prevent this situation statically, we
introduce a stronger notion of typing. The types are trees with
leaves labelled with ports:
\begin{defi}[Type]
  A \emph{type} over $S\subseteq P$ is a triple
  $\tau=\tuple{V,E,\lambda}$ such that $\tuple{V,E}$ is a proper
  unlabelled top-down tree, and $\lambda$ is a bijective function from
  $S$ to $\max\tuple{V,E}$.  The set of types over subsets of $P$ is
  written $\typ[P]$.
\end{defi}
\noindent
As before, types are considered up to bijective renaming of internal
vertices. It is then a simple observation to notice that $\typ[P]$ is
finite (recall that $P$ was assumed to be finite as well). We present
two examples of such types in Figure~\ref{fig:extyp}.

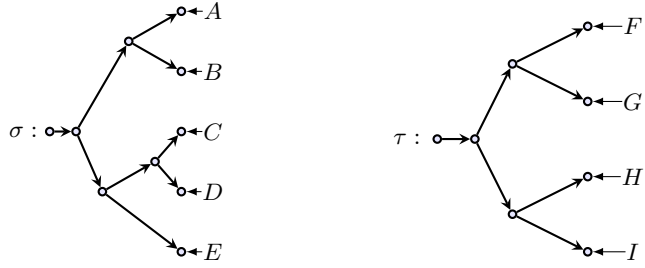
\begin{figure}[t]
  \centering
  \begin{tikzpicture}[yscale=.8,xscale=0.7]
        \state(t0) (0,2); \state(t1) (.5,2); \state(t2) (1.5,3.5);
        \state(t6) (1,1); \state(t9) (2,1.5); \state(t3) (2.5,4);
        \state(t4) (2.5,3); \state(t5) (2.5,2); \state(t7) (2.5,1);
        \state(t8) (2.5,0);
        \edge (t0) (t1); \edge (t1) (t2); \edge (t2) (t3);
        \edge (t2) (t4); \edge (t1) (t6); \edge (t6) (t9);
        \edge (t6) (t8); \edge (t9) (t5); \edge (t9) (t7);
        \oport (A) (t3); \oport (B) (t4); \oport (C) (t5);
        \oport (D) (t7); \oport (E) (t8);
        \node (lblt) at ($(t0) - (0.5,0)$) {$\sigma :$};
      \end{tikzpicture}\hspace{2cm}
      \begin{tikzpicture}
        \state(t0) (0,2); \state(t1) (.5,2);
        \state(t6) (1,1); \state(3)(2,1.5); \state(5)(2,0.5);
        \state(6)(1,3); \state(7)(2,2.5); \state(8)(2,3.5);
        \edge (t0) (t1); \edge (t1) (6); \edge (t1) (t6);
        \edge (t6) (3); \edge (t6) (5); \edge[below] (6) (7);
        \edge (6) (8);
        \oport (F) (8);\oport (G) (7);\oport (I) (5);\oport (H) (3);
        \node (lblt) at ($(t0) - (.4,0)$) {$\tau :$};
      \end{tikzpicture}
  \caption{Examples of types}
  \label{fig:extyp}
\end{figure}

We may forget about the label information in a box: the $a$-erasing of
a box $\beta$ is the box $\ers\beta$ where all labels are replaced by
an arbitrary letter $a\in Y$.  When the labels are not relevant
(or, in the next section when we label edges with expressions rather
than variables), we can define SP-reductions of boxes. Given a box
$\beta=\tuple{\pin,G,\pout}$, if $G\spred G'$, and if no vertex in the
image of $\pin$ was erased in the reduction, we write
$\beta\spred\tuple{\pin,G',\pout}$. Composition commutes with
SP-reductions: if $\beta\spred\beta'$, then
$\beta\odot\gamma\spred\beta'\odot\gamma$ and
$\gamma\odot\beta\spred\gamma\odot\beta'$.

A type $\tau=\tuple{V,E,\lambda}$ over $S$ may be seen as a box from
$\set\iota$ to $S$: if $r$ is the root of $\tau$ and $a\in Y$, we
can build the box $\typtobox{\tau}$ from~$\set\iota$ to~$S$ as
$\tuple{[\iota\mapsto r],\tuple{V,E'},\lambda}$, with
$E'=\setcompr{\tuple{x,a,y}}{\tuple{x,y}\in E}$.
For every type $\tau$, $\typtobox\tau=\ers{\typtobox\tau}$; for all
boxes $\beta\in\boxsetb {S_1} {S_2}$ and
$\gamma\in\boxsetb {S_2} {S_3}$, we have
$\ers{\beta\odot\gamma}=\ers\beta\odot\ers\gamma$.

Now we can define typed boxes.
\begin{defi}[Typed boxes]\label{sec:definitions-3}
  Let $\beta\in\boxsetb S {S'}$ be a box, $\sigma$ and $\tau$ be types
  respectively over $S$ and ${S'}$. \emph{$\beta$ has the type
    $\sigma\to\tau$} if
  $\ers{\typtobox\sigma \odot \beta}\spred[\star]\typtobox\tau$. We
  write $\tboxsetb[Y]\sigma\tau$ for the set of boxes over $Y$ of
  type $\sigma\to\tau$, and $\tboxset[Y]$ for the set of typed boxes
  over $Y$.
\end{defi}

\begin{exa}\label{ex:typing}
  Consider the following type $\sigma$ and box $\beta$:
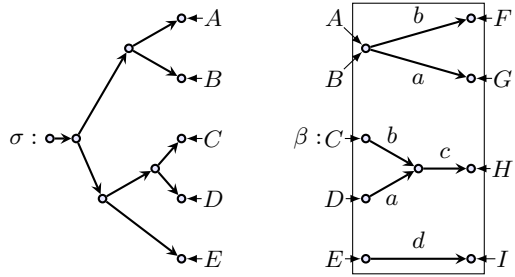
  \begin{align*}
    \begin{tikzpicture}[yscale=.8,xscale=0.7]
      \begin{scope}
        \state(t0) (0,2);
        \state(t1) (.5,2);
        \state(t2) (1.5,3.5);
        \state(t6) (1,1);
        \state(t9) (2,1.5);
        \state(t3) (2.5,4);
        \state(t4) (2.5,3);
        \state(t5) (2.5,2);
        \state(t7) (2.5,1);
        \state(t8) (2.5,0);
        \edge (t0) (t1);
        \edge (t1) (t2);
        \edge (t2) (t3);
        \edge (t2) (t4);
        \edge (t1) (t6);
        \edge (t6) (t9);
        \edge (t6) (t8);
        \edge (t9) (t5);
        \edge (t9) (t7);
        \oport (A) (t3);
        \oport (B) (t4);
        \oport (C) (t5);
        \oport (D) (t7);
        \oport (E) (t8);
        \node (lblt) at ($(t0) - (0.5,0)$) {$\sigma :$};
      \end{scope}
      \begin{scope}[xshift=5cm]
        \state(0)(1,1);
        \state(1)(1,2);
        \state(2)(2,1.5);
        \state(3)(3,1.5);
        \state(4)(1,0);
        \state(5)(3,0);
        \state(6)(1,3.5);
        \state(7)(3,3);
        \state(8)(3,4);
        \clbox (4) (8);
        \edge[below] (0) (2)[a];
        \edge (1) (2)[b];
        \edge (2) (3)[c];
        \edge (4) (5)[d];
        \edge[below] (6) (7)[a];
        \edge (6) (8)[b];       
        \iport[.5] (A) (6);\iport[-.5] (B) (6);
        \iport (C) (1);\iport (D) (0);\iport (E) (4);
        \oport (F) (8);\oport (G) (7);\oport (I) (5);\oport (H) (3);
        \node (lblt) at ($(pCi) - (0.5,0)$) {$\beta :$};
      \end{scope}
    \end{tikzpicture}
  \end{align*}
  First, we glue them together, and erase the labels. This yields the
  graph:
  
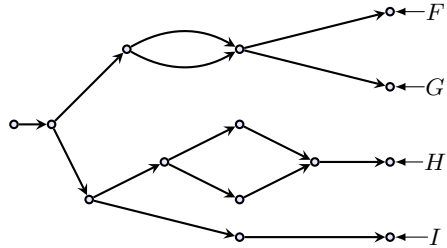
\begin{align*}
    \begin{tikzpicture}
      \state(t0) (0,2);
      \state(t1) (.5,2);
      \state(t2) (1.5,3);
      \state(t6) (1,1);
      \state(t9) (2,1.5);
      \state(0)(3,1);
      \state(1)(3,2);
      \state(2)(4,1.5);
      \state(3)(5,1.5);
      \state(4)(3,0.5);
      \state(5)(5,0.5);
      \state(6)(3,3);
      \state(7)(5,2.5);
      \state(8)(5,3.5);
      \edge (t0) (t1);
      \edge (t1) (t2);
      \edge[bend left] (t2) (6);
      \edge[bend right] (t2) (6);
      \edge (t1) (t6);
      \edge (t6) (t9);
      \edge (t6) (4);
      \edge (t9) (1);
      \edge (t9) (0);
      \edge[below] (0) (2);
      \edge (1) (2);
      \edge (2) (3);
      \edge (4) (5);
      \edge[below] (6) (7);
      \edge (6) (8);
      \oport (F) (8);\oport (G) (7);\oport (I) (5);\oport (H) (3);
    \end{tikzpicture}
  \end{align*}
  Then, we apply as many reductions steps as much as possible:
  
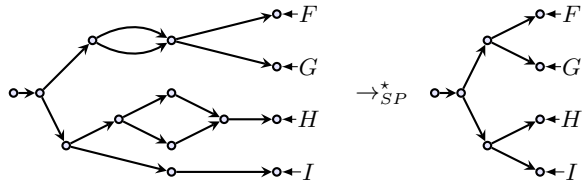
\begin{align*}
    \begin{tikzpicture}[scale=.7]
      \begin{scope}
        \state(t0) (0,2);
        \state(t1) (.5,2);
        \state(t2) (1.5,3);
        \state(t6) (1,1);
        \state(t9) (2,1.5);
        \state(0)(3,1);
        \state(1)(3,2);
        \state(2)(4,1.5);
        \state(3)(5,1.5);
        \state(4)(3,0.5);
        \state(5)(5,0.5);
        \state(6)(3,3);
        \state(7)(5,2.5);
        \state(8)(5,3.5);
        \edge (t0) (t1);
        \edge (t1) (t2);
        \edge[bend left] (t2) (6);
        \edge[bend right] (t2) (6);
        \edge (t1) (t6);
        \edge (t6) (t9);
        \edge (t6) (4);
        \edge (t9) (1);
        \edge (t9) (0);
        \edge[below] (0) (2);
        \edge (1) (2);
        \edge (2) (3);
        \edge (4) (5);
        \edge[below] (6) (7);
        \edge (6) (8);
        \oport (F) (8);\oport (G) (7);\oport (I) (5);\oport (H) (3);
      \end{scope}
      \begin{scope}[shift={(8cm,0cm)}]
        \state(t0) (0,2);
        \state(t1) (.5,2);
        \state(t6) (1,1);
        \state(3)(2,1.5);
        \state(5)(2,0.5);
        \state(6)(1,3);
        \state(7)(2,2.5);
        \state(8)(2,3.5);
        \edge (t0) (t1);
        \edge (t1) (6);
        \edge (t1) (t6);
        \edge (t6) (3);
        \edge (t6) (5);
        \edge[below] (6) (7);
        \edge (6) (8);
        \oport (F) (8);\oport (G) (7);\oport (I) (5);\oport (H) (3);
        \node (lblt) at ($(t0) - (1,0)$) {$\spred[\star]$};
      \end{scope}
    \end{tikzpicture}
  \end{align*}
  The last graph is the box corresponding to a tree $\tau$. Hence we
  can state that $\beta$ has the type $\sigma\rightarrow\tau$.
\end{exa}

%%% Local Variables:
%%% mode: latex
%%% TeX-master: "main"
%%% End:

\begin{rem}
  Notice that the type of a box is not unique: a single box may have
  multiple types. However, given an input type $\sigma$ and a box
  $\beta$, there is at most one output type $\tau$ such
  that~$\beta\in\tboxsetb\sigma\tau$.
\end{rem}

\noindent When $\sigma$ is a type over $S$, we set
$\bid\sigma\eqdef\bid S$.
\begin{prop}\label{prop:typcat}
  We have a category $\tcatbox$ of typed boxes, with $\typ[P]$ as the
  set of objects, and $\tboxsetb[Y]\tau\sigma$ as the set of
  morphisms from $\tau$ to $\sigma$.
\end{prop}
\begin{proof}
  It suffices to show that the identity box can be typed:
  $\ers{\typtobox\sigma \odot \bid
    S}=\ers{\typtobox\sigma}=\typtobox\sigma$;
  and that composition of typed boxes is a typed box. Let
  ${\beta,\gamma\in\tboxsetb \sigma\tau
    \times\tboxsetb{\tau}{\theta}}$, we show that %their composite
  $\beta\odot\gamma\in\tboxsetb{\sigma}{\theta}$:
  \begin{align*}
    \ers{\typtobox\sigma \odot (\beta\odot\gamma)}
    =\ers{\typtobox\sigma \odot \beta}\odot\ers{\gamma}
    \spred[\star]\typtobox\tau\odot\ers{\gamma}
    =\ers{\typtobox\tau\odot\gamma}
    \spred[\star]\typtobox\theta%\tag*\qedhere
  \end{align*}
\end{proof}

% \begin{rem}
%   We can extend the product of boxes to sets of boxes: let
%   ${\Gamma\subseteq\boxsetb A B}$ and ${\Delta\subseteq\boxsetb B C}$
%   be two sets of boxes. We define the pointwise composition of
%   $\Gamma$ and $\Delta$:
%   \[{\Gamma\odot\Delta\eqdef\setcompr{\gamma\odot\delta}{\gamma,\delta\in\Gamma\times\Delta}\subseteq\boxsetb
%       A C}.\]
%   By Lemma~\ref{lem:typstable}, we get that if
%   ${\Gamma\subseteq\tboxsetb\tau\sigma}$ and
%   ${\Delta\subseteq\tboxsetb\sigma\chi}$, then
%   ${\Gamma\odot\Delta\subseteq\tboxsetb \tau\chi}$.
% \end{rem}
\begin{rem}
  \label{rem:tkleene}
  Given a category $\mathcal C$, one can form a Kleene category
  $\mathcal K$ (also called typed Kleene
  algebra~\cite{kozen1998typed,moller1999typed}) whose
  types are the objects of the category, and where the morphisms from
  $a$ to $b$ are subsets of the set $\C\bracket{a,b}$ of morphisms
  from $a$ to $b$ in $\C$. The product in this algebra is the
  pointwise composition of morphisms, and sum is union. In particular,
  $\pset{\boxset}$ and $\pset{\tboxset}$ are Kleene categories.
\end{rem}

\subsubsection{Templates.}
\label{sec:generalised-boxes}

Another important notion we need for the second direction of our
Kleene theorem is that of template. In the proof of the classical
Kleene theorem, one moves from automata to generalised automata, which
are labelled with regular expressions rather than with letters. This
allows for the label of a single transition to represent many paths in
the original automaton. Finite templates will serve this function in
our proof. A finite template is a finite set of boxes sharing the same
input and output ports, and labelled with simple expressions.
\begin{defi}[(Finite) Template]
  Let $S$ and $S'$ be sets of of ports. A \emph{template}
  (respectively \emph{finite template}) from $S$ to ${S'}$ is a set
  (resp. finite set) $\Gamma\subseteq\boxsetb[{{\SExp[Y]}}] S {S'}$ of
  boxes from $S$ to ${S'}$ labelled with simple expressions.  If
  furthermore $\sigma$ is a type over $S$, $\tau$ is a type over $S'$
  and $\Gamma\subseteq\tboxsetb[{{\SExp[Y]}}]\sigma\tau$, we say that
  $\Gamma$ has type $\sigma\to\tau$.  We write $\gboxset S {S'}$ for
  the set of finite templates from $S$ to ${S'}$, and
  $\gtboxset \sigma\tau$ for finite templates of type $\sigma\to\tau$.
\end{defi}
\noindent One extracts boxes out of templates as follows.
\begin{defi}[Boxes generated by a template]
  Let $\Gamma\in\gboxset S {S'}$ be a template from $S$ to
  ${S'}$. Then $\beta\in\boxsetb[Y] S {S'}$ is \emph{generated} by
  $\Gamma$ if it can be obtained from a box
  $\tuple{\pin,G,\pout}\in\Gamma$ by replacing each edge
  $x\xrightarrow{~e~}y$ by a series-parallel graph $G'\in\Gr e$ with
  input vertex $x$ and output vertex $y$. We write $\rlang \Gamma$ for
  the set of boxes generated by $\Gamma$.
\end{defi}

\noindent
For instance in the example below the template $\Gamma$ generates
all boxes of the shapes $\beta_n$ and $\delta_n$, for $n>0$:
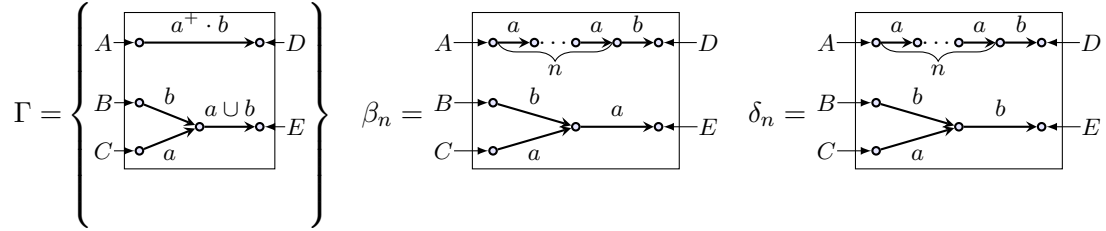
\begin{align*}
  \Gamma=&\left\{
  \begin{tikzpicture}[scale=.8,baseline=(2.north)]
    \state (0) (1,0.2);\state (1) (1,1);\state (2) (2,.6);
    \state (3) (3,.6);
    \state (4) (1,2);\state (5) (3,2);
    \edge[below] (0) (2)[a];\edge (1) (2)[b];
    \edge (2) (3)[a\sumtm b];\edge (4) (5)[a^+\cdot b];
    \rect[.3](0)(0)(5)(5);
    \iport (A) (4);\iport (B) (1);\iport (C) (0);
    \oport (D) (5);\oport (E) (3);
  \end{tikzpicture}\right\}
  &\beta_n=~&
  \begin{tikzpicture}[yscale=.8,baseline=(2.north),xscale=1.1]
    \state (0) (1,0.2);\state (1) (1,1);\state (2) (2,.6);
    \state (3) (3,.6);
    \state (4) (1,2);\state (41) (1.5,2);\state (43) (2,2);
    \state (44) (2.5,2);
    \state (5) (3,2);
    \edge[below] (0) (2)[a];\edge (1) (2)[b];\edge (2) (3)[a];
    \edge (4) (41)[a];
    \edge (43) (44)[a];\edge (44) (5)[b];
    \rect[.3](0)(0)(5)(5);
    \iport (A) (4);\iport (B) (1);\iport (C) (0);
    \oport (D) (5);\oport (E) (3);
    \node()at($(41)!.5!(43)$){$\dots$};
    \draw [decorate,decoration={brace,amplitude=1.5ex}] (44) -- (4)
    node[midway,below,yshift=-1ex]{$n$};
  \end{tikzpicture}
  &\delta_n=~&
  \begin{tikzpicture}[yscale=.8,baseline=(2.north),xscale=1.1]
    \state (0) (1,0.2);\state (1) (1,1);\state (2) (2,.6);
    \state (3) (3,.6);
    \state (4) (1,2);\state (41) (1.5,2);\state (43) (2,2);
    \state (44) (2.5,2);
    \state (5) (3,2);
    \edge[below] (0) (2)[a];\edge (1) (2)[b];\edge (2) (3)[b];
    \edge (4) (41)[a];
    \edge (43) (44)[a];\edge (44) (5)[b];
    \rect[.3](0)(0)(5)(5);
    \iport (A) (4);\iport (B) (1);\iport (C) (0);
    \oport (D) (5);\oport (E) (3);
    \node()at($(41)!.5!(43)$){$\dots$};
    \draw [decorate,decoration={brace,amplitude=1.5ex}] (44) -- (4)
    node[midway,below,yshift=-1ex]{$n$};
  \end{tikzpicture}
\end{align*}
\begin{fact}
  If $\Gamma$ has type $\sigma\to\tau$, then
  $\rlang \Gamma\subseteq \tboxsetb[Y]\sigma\tau$.
\end{fact}

% We can extend the product of boxes to templates: let
% ${\Gamma\in\gboxset {S_1} {S_2}}$ and
% ${\Delta\in\gboxset {S_2} {S_3}}$ be two templates. We define their
% composition $\Gamma\odot\Delta$ as
% $\setcompr{\gamma\odot\delta}{\paren{\gamma,\delta}\in\Gamma\times\Delta}$.
% By Lemma~\ref{lem:typcomp}, we get that if
% ${\Gamma\in\gtboxset\sigma\tau}$ and ${\Delta\in\gtboxset\tau\chi}$,
% then ${\Gamma\odot\Delta\in\gtboxset \sigma\chi}$.

\noindent
According to Remark~\ref{rem:tkleene}, templates and typed templates
form Kleene categories. The function $\rlang\argument$ used to extract
boxes from templates is a functor from the Kleene category $\gboxsetc$
to the Kleene category $\pset\boxset$, and it preserves the regular
operations.

If $\Gamma,\Delta$ are finite templates of appropriate sort, then
$\Gamma\sumtm\Delta$ and $\Gamma\cdot\Delta$ are also finite templates.
However the Kleene star of a finite set of boxes yields an infinite
set of boxes. Even worse, this set of boxes cannot always be generated
by a finite template. For instance, the star of the box $\delta$ below
yields all boxes $\delta^n$, for $n\in\Nat$:
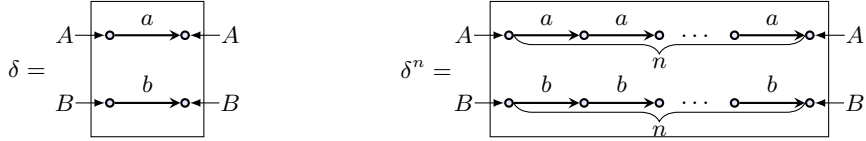
\begin{align*}
  \begin{tikzpicture}[yscale=.6,baseline=(A.south)]
    \state (x1) (0.5,1);\state (x0) (0.5,-0.5);
    \state (x2) (1.5,1);\state (x3) (1.5,-0.5);
    \edge (x1) (x2)[a];\edge (x0) (x3)[b];
    \rect[.75] (x0) (x1)[.75] (x2) (x3); 
    \iport (A)[i] (x1);\iport (B)[j] (x0);
    \oport (A) (x2); \oport (B) (x3);
    \node () at ($(i)!.5!(j)-(.5,0)$) {$\delta=$};
  \end{tikzpicture}\hspace{2cm}
  \begin{tikzpicture}[yscale=.6,baseline=(A.south)]
    \state (a1) (0.5,1); \state (b1) (0.5,-0.5);
    \state (a2) ($(a1)+(1,0)$);\state (b2) ($(b1)+(1,0)$);
    \state (a3) ($(a2)+(1,0)$);\state (b3) ($(b2)+(1,0)$);
    \state (a4) ($(a3)+(1,0)$);\state (b4) ($(b3)+(1,0)$);
    \state (a5) ($(a4)+(1,0)$);\state (b5) ($(b4)+(1,0)$);
    \edge (a1) (a2)[a];\edge (a2) (a3)[a];\edge (a4) (a5)[a];
    \node () at ($(a3)!.5!(a4)$) {$\dots$};
    \draw [decorate,decoration={brace,amplitude=1.5ex}] (a5) -- (a1)
    node[midway,below,yshift=-1ex]{$n$};
    \edge (b1) (b2)[b];\edge (b2) (b3)[b];\edge (b4) (b5)[b];
    \node () at ($(b3)!.5!(b4)$) {$\dots$};
    \draw [decorate,decoration={brace,amplitude=1.5ex}] (b5) -- (b1)
    node[midway,below,yshift=-1ex]{$n$};
    \rect[.75] (b1) (b1)[.75] (a5) (a5); 
    \iport (A)[i] (a1);\iport (B)[j] (b1);
    \oport (A) (a5); \oport (B) (b5);
    \node () at ($(i)!.5!(j)-(.5,0)$) {$\delta^n=$};
  \end{tikzpicture}
\end{align*}
This set of boxes cannot be represented by a finite number of boxes,
as the two branches have no way to synchronise. (This means that we
cannot ensure there will be exactly the same number of iterations on
both branches.) Fortunately, it is possible to define Kleene star for
a rich enough class of templates, namely for \emph{atomic} templates.

\subsubsection{Atomic boxes and templates.}
\label{sec:atomic-boxes}
% As our goal is to define the star of typed templates, we will mostly
% consider in this section boxes and templates with type
% $\sigma\to\sigma$. 
\begin{defi}[Support]
  The \emph{support} of a box $\beta\in\boxsetb S S$ is the following set
  \begin{align*}
    \supp\beta\eqdef\setcompr {p} {\pin\paren p\neq\pout\paren p}
  \end{align*}
  The support of a template $\Gamma\in\gboxset S S$ is defined as
  $\bigcup_{\beta\in\Gamma}\supp\beta$.
\end{defi}
\noindent
Intuitively, the support constitutes the irreflexive part of a
box. In particular, ${\supp{\bid\sigma}}$ is always empty. 

\begin{fact}
  If $\Gamma,\Gamma'$ are templates of type $\sigma\to\sigma$ with
  disjoint support, then $\Gamma\odot\Gamma'=\Gamma'\odot\Gamma$.
\end{fact}

\begin{defi}[Atomic box, atomic template]
  A box $\alpha=\tuple{\pin,G,\pout}\in\boxsetb S S$ is \emph{atomic}
  if its graph has a single non-trivial connected component $C$, and
  if for every vertex $v$ outside $C$, there is a unique port $p\in S$
  such that $\pin\paren p=\pout\paren p=v$. An \emph{atomic template}
  is a finite template exclusively composed of atomic boxes.
\end{defi}

\begin{figure}[t]
  \centering
  \begin{tikzpicture}
    \position (0) (0,1.75);
    \position (1) (0,1);
    \position (2) (.5,1.5);
    \position (x) (1.25,1.75);
    \position (3) (2,2);
    \position (4) (2,1.5);
    \position (5) (2,1);
    \position (7) ($(1)!.5!(5)-(0,.5)$);
    \position (8) ($(7)-(0,.5)$);
    \rect(8)(1)(3)(3);
    \edge(0)(2)[a];\edge[below](1)(2)[b];\edge(2)(x)[c];
    \edge[below](2)(5)[a];\edge(x)(3)[b];\edge[below](x)(4)[a];
    \iport[.25](A)(0);\iport[-.25](B)(0);
    \iport(C)(1);\iport(D)(7)[1];\iport(E)[ai](8)[1];
    \oport(A)(3);\oport(B)(4);\oport(C)(5);
    \oport(D)(7)[5];\oport(E)[af](8)[5];

    \state (x1) (4,1.75);\state (x0) (4,0.25);
    \state (x2) (5,1.75);\state (x3) (5,0.25);
    \edge (x1) (x2)[a];\edge (x0) (x3)[b];
    \rect[.75] (x0) (x1)[.75] (x2) (x3); 
    \iport (A)[nai] (x1);\iport (B) (x0);
    \oport (A) (x2); \oport (B) (x3);

    \state (y1) (7,2);\state (y2) (8,2);
    \state(y2') ($(y2)-(0,.5)$);
    \state(y3) ($(y1)!.5!(y2)-(0,1)$);
    \state(y4) ($(y3)-(0,.5)$);
    \state(y5) ($(y4)-(0,.5)$);
    \edge (y1) (y2)[a];\edge[below](y1)(y2')[b];
    \rect (y5) (y1) (y2) (y2); 
    \iport (A) (y1);\iport[.5](B) (y3)[y1];
    \iport(C) (y3)[y1];\iport (D) (y4)[y1];
    \iport(E)(y5)[y1];
    \oport (A) (y2); \oport (B) (y2');
    \oport(C) (y3)[y2];\oport (D) (y4)[y2];
    \oport(E)(y5)[y2];

    \state (z1) (10,1.75);\state (z2) (11,2);
    \state(z2') ($(z2)-(0,.5)$);
    \state(z3) ($(z1|-z2')!.5!(z2')-(0,.5)$);
    \state(z4) ($(z3)-(0,.5)$);
    \state(z5) ($(z4)-(0,.5)$);
    \edge (z1) (z2)[a];\edge[below](z1)(z2')[b];
    \rect (z5) (z1) (z2) (z2); 
    \iport[.25] (A) (z1);\iport[-.25](B) (z1);
    \iport(C) (z3)[z1];\iport (D) (z4)[z1];
    \iport(E)(z5)[z1];
    \oport (A) (z2); \oport (B) (z2');
    \oport(D) (z3)[z2];\oport (C) (z4)[z2];
    \oport(E)[naf](z5)[z2];

    \draw[decorate,decoration=brace,thick] 
    ($(af)-(0,1)$) to node[below,yshift=-4pt] {atomic}
    ($(ai)-(0,1)$);
    \draw[decorate,decoration=brace,thick] 
    ($(naf)-(0,1)$) to node[below,yshift=-4pt] {non-atomic}
    ($(nai|-naf)-(0,1)$);

  \end{tikzpicture}
  \caption{Atomic and non-atomic boxes}
  \label{fig:ill:atomic}
\end{figure}
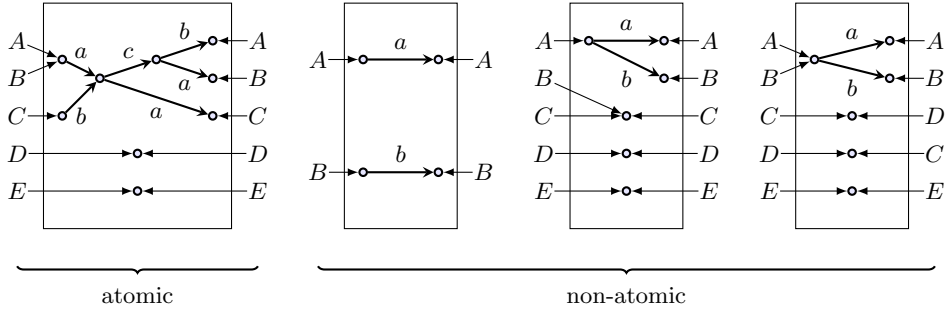

Atomic templates with singleton support can be easily
iterated. Indeed, the non-trivial connected component of an atomic box
with singleton support is series-parallel, and may thus be SP-reduced
to a graph with only two vertices, joined by a single edge labelled
with some simple expression $e$. If we then replace $e$ with $e^+$,
and put the resulting box $\alpha'$ in a template
$\alpha^\star\eqdef\set{\bid \sigma,\alpha'}$, we easily get
$\rlang{\alpha^\star}=\rlang\alpha^\star$. Now if
$\Gamma=\set{\alpha_1,\dots,\alpha_n}$ is an atomic template with
singleton support, it must be that every $\alpha_i$ has the same
singleton support. We may thus reduce their graphs, yielding
expressions $e_1,\dots,e_i$, and use the label
$\paren{e_1\sumtm\dots\sumtm e_i}^+$ to do the same construction.
\begin{lem}\label{lem:typboxbowtie}
  The non-trivial connected component of an atomic box of type
  $\sigma\to\sigma$ always contains a vertex $c$, such that for every
  port $p$ mapped inside that component, all paths from $\pin\paren p$
  to a maximal vertex visit $c$.
\end{lem}
\begin{proof}
  This is a direct consequence of Lemma~\ref{lem:gluecenter}
\end{proof}
\noindent
This lemma allows us to split an atomic box into the product
$\alpha=\alpha^1\odot\alpha^2$ of two typed boxes such that
$\alpha^2\odot\alpha^1$ has singleton support.  Furthermore if
$\supp\Gamma\subseteq\supp\alpha$ then
$\set{\alpha^2}\odot\Gamma\odot\set{\alpha^1}$ has singleton support.

Another important property is that the supports of atomic boxes of the
same type are either disjoint or comparable:
\begin{lem}
  For all atomic boxes $\beta,\gamma\in\tboxsetb\sigma\sigma$,
  we have either
  \begin{itemize}
  \item $\supp\beta\subseteq\supp\gamma$, or
  \item $\supp\gamma\subseteq\supp\beta$, or
  \item $\supp\beta\cap\supp\gamma=\emptyset$.
  \end{itemize}
\end{lem}
\begin{proof}
  It suffices to observe that the support of every atomic box
  $\beta\in\tboxsetb\sigma\sigma$ has to be the set of leaves
  reachable from some vertex in $\sigma$.
\end{proof}

\begin{prop}
  For every atomic template $\Gamma$ with type $\sigma\to\sigma$, there
  is a finite template~$\Gamma^\star$ with the same type such that
  $\rlang\Gamma^\star=\rlang{\Gamma^\star}$.
\end{prop}
\begin{proof}
  Let $\Gamma=\set{\alpha_1,\dots,\alpha_n}$ be an atomic template of
  type $\sigma\to\sigma$, indexed in such a way that $i < j$ entails
  either $\supp{\alpha_i}\subseteq\supp{\alpha_j}$ or
  $\supp{\alpha_i}\cap\supp{\alpha_j}=\emptyset$. We define
  $\emptyset^\star=\gbid\sigma$. Then for every $k\leqslant n$, we
  split $\set{\alpha_1,\dots,\alpha_{k-1}}$ into $\Gamma_1$ and
  $\Gamma_2$, such that $\supp{\Gamma_1}\subseteq\supp{\alpha_k}$ and
  $\supp{\Gamma_2}\cap\supp{\alpha_k}=\emptyset$. We obtain:
  \begin{align}
    \rlang{\set{\alpha_1,\dots,\alpha_{k-1},\alpha_k}}^\star
    =&\rlang{\paren{\Gamma_1\sumtm\alpha_k}\sumtm\Gamma_2}^\star\tag{commutativity}\\
    =&\paren{\rlang{\Gamma_1}\cup\rlang{\alpha_k}}^\star
       \odot\rlang{\Gamma_2}^\star
       \tag{$\supp{\Gamma_2}\cap\supp{\Gamma_1\sumtm\alpha_k}=\emptyset$}\\
    =&\rlang{\Gamma_1^\star}\rlang{\alpha_k\cdot\Gamma_1^\star}^\star
       \odot\rlang{\Gamma_2^\star}\tag{regular laws}\\ 
    =&\rlang{\Gamma_1^\star}\paren{\rlang{\gbid\sigma}
       \cup\rlang{\alpha_k\Gamma_1^\star}\rlang{\alpha_k\Gamma_1^\star}^\star}
       \odot\rlang{\Gamma_2^\star}\tag{regular laws}\\
    =&\rlang{\Gamma_1^\star}\paren{\rlang{\gbid\sigma}
       \cup\rlang{\alpha^1_k\alpha^2_k\Gamma_1^\star}\rlang{\alpha^1_k\alpha^2_k\Gamma_1^\star}^\star}
       \odot\rlang{\Gamma_2^\star}\tag{$\alpha_k=\alpha^1_k\alpha^2_k$}\\
    =&\rlang{\Gamma_1^\star}\paren{\rlang{\gbid\sigma}
       \cup\rlang{\alpha^1_k}\rlang{\alpha^2_k\Gamma_1^\star\alpha^1_k}^\star\rlang{\alpha^2_k\Gamma_1^\star}}
       \odot\rlang{\Gamma_2^\star}\tag{regular laws}
  \end{align}
  As we noticed earlier, because we have
  $\supp{\Gamma_1}\subseteq\supp{\alpha_k}$ the template
  $\alpha^2_k\Gamma_1^\star\alpha^1_k$ has a singleton support. This
  means we can compute its star, thus reduce the last expression into
  a single finite template.
\end{proof}

\subsection{Extracting simple expressions from Petri Automata}
\label{sec:main-theorem}

Now we have enough material to embark in the proof that recognisable
languages of series-parallel graphs are regular.

Let us fix an automaton $\A=\tuple{P,\T,\iota}$. Assume that $\iota$
never appears in the output of a transition and that $P$ contains a
place $f$ not connected to any transition. (It is easy to modify $\A$
to enforce these two properties).

We start by building a finite state automaton whose states are types,
and transitions are typed boxes. Then, using a procedure similar to
the proof of the classical Kleene's Theorem, we reduce it into a
single box template from which we can extract a simple expression.

\subsubsection{A regular language of runs.}
\label{sec:kleene-theorem}

We can associate a box with every transition of a proper run, as
illustrated in Example~\ref{ex:boxes-run}.

\begin{defi}[Box of a transition]
  Suppose $\tr\in\T$ is a non-final transition in $\A$, and
  $S,{S'}\subseteq P$ are two states such that $\atrans S {S'}$.  The
  \emph{box of $\tr$ from $S$} (written~$\boxtrans\tr S$) is built as
  follows. Its set of input ports (respectively output ports) is $S$
  (resp. ${S'}$). The vertices of its graph are the places in ${S'}$,
  together with an additional node $*$. Places in $\trin$ are mapped
  by $\pin$ to $*$, and the others are mapped to themselves. $\pout$
  sends every place in ${S'}$ (seen as a port) to itself (seen as a
  vertex). Finally, we put edges $\tuple{*,a,q}$ whenever
  $\tuple{a,q}\in \pout$.
  For final transitions~$\tr=\tuple{\trin,\emptyset}$ we adapt the
  construction from state~$\trin$ to reach the state $\set f$, by
  defining $\boxtrans\tr \trin$ to be
  $\tuple{[\_ \mapsto *],\tuple{\set *,\emptyset},[f\mapsto *]}$.
\end{defi}

We extend this construction to runs in a straightforward way: if
$S_0,\dots,S_n$ are states and $R=\tuple{S_0,\tr[1];\dots;\tr[n],S_n}$
is a proper run in $\A$ we define
$\boxrun R\eqdef
\boxtrans{\tr[1]}{S_0}\odot\boxtrans{\tr[2]}{S_1}\odot
\cdots\odot\boxtrans{\tr[n]}{S_{n-1}}$. With this definition,
accepting runs yield boxes in $\boxsetb{\set\iota}{\set f}$.
This actually amounts to computing the trace of $R$:
\begin{lem}\label{lem:boxtrace}
  If $\boxrun R = \tuple{\pin,G,\pout}$, then $G$ is isomorphic to $\Gr{R}$.
\end{lem}
\begin{proof}
  The vertex $k$ in $\Gr R$ (produced by $\tr[k]$) is equivalent to
  the vertex $*$ coming from the same transition.
\end{proof}
\begin{exa}\label{ex:boxes-run}
  Recall the accepting run from Example~\ref{ex:run}:
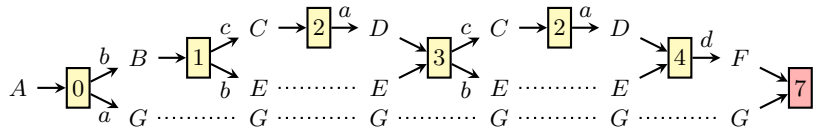
  \begin{align*}
    \begin{tikzpicture}[scale=.8]
      \node (A) at (0,-0.5) {$A$};
      \node (B) at (2,0) {$B$};
      \node (G1) at (2,-1) {$G$};
      \node (C) at (4,0.5) {$C$};
      \node (E1) at (4,-0.5) {$E$};
      \node (G2) at (4,-1) {$G$};
      \node (D) at (6,0.5) {$D$};
      \node (E2) at (6,-0.5) {$E$};
      \node (G3) at (6,-1) {$G$};
      \node (C1) at (8,0.5) {$C$};
      \node (E3) at (8,-0.5) {$E$};
      \node (G4) at (8,-1) {$G$};
      \node (D1) at (10,0.5) {$D$};
      \node (E4) at (10,-0.5) {$E$};
      \node (G5) at (10,-1) {$G$};
      \node (F) at (12,0) {$F$};
      \node (G6) at (12,-1) {$G$};
      \trans[0] (0) (1,-0.5);
      \trans[1] (1) (3,0);
      \trans[2] (2) (5,0.5);
      \trans[3] (3) (7,0);
      \trans[2] (4) (9,0.5);
      \trans[4] (5) (11,0);
      \transf[7] (6) ($(F)!.5!(G6)+(1,0)$);
      % \node[trans,dotted] (6) at (13,-0.5) {$6$};
      \edge (A) (0);
      \edge[above] (0) (B)[b];
      \edge[below] (0) (G1)[a];
      \edge (B) (1);
      \draw[thick,dotted] (G1) to (G2);
      \edge[above] (1) (C)[c];
      \edge[below] (1) (E1)[b];
      \edge (C) (2);
      \draw[thick,dotted] (E1) to (E2);
      \draw[thick,dotted] (G2) to (G3);
      \edge[above] (2) (D)[a];
      \edge (D) (3);
      \edge (E2) (3);
      \draw[thick,dotted] (G3) to (G4);
      \edge[above] (3) (C1)[c];
      \edge[below] (3) (E3)[b];
      \edge (C1) (4);
      \draw[thick,dotted] (E3) to (E4);
      \draw[thick,dotted] (G4) to (G5);
      \edge[above] (4) (D1)[a];
      \edge (D1) (5);
      \edge (E4) (5);
      \draw[thick,dotted] (G5) to (G6);
      \edge[above] (5) (F)[d];
      \edge (F) (6);
      \edge (G6) (6);
      % \draw[thick,dotted] (F) to (6);
      % \draw[thick,dotted] (G6) to (6);      
    \end{tikzpicture}
  \end{align*}
  The transitions of this run yield the following
  boxes:
  \begin{align*}
    \boxtrans 0 {\set A}=
    &\quad
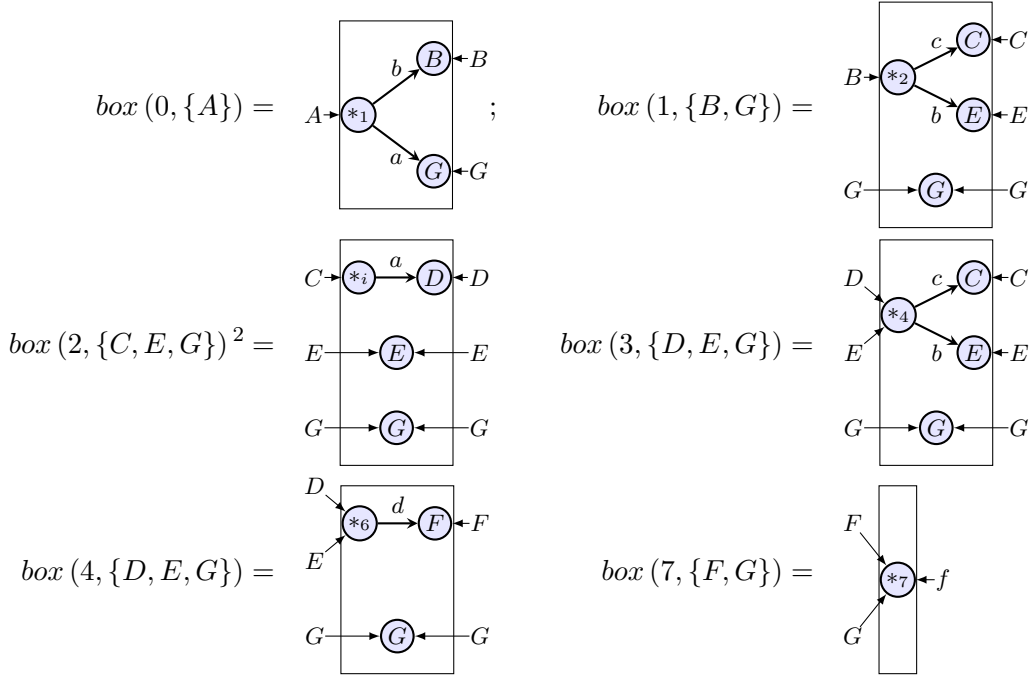
\begin{tikzpicture}[baseline=(0)]
      \state[*_1] (0) (0,0);
      \state[B] (b0) (1,.75);
      \state[G] (g0) (1,-.75);
      \edge (0) (b0)[b];
      \edge[below] (0) (g0)[a];
      \rect (g0) (0) (b0) (b0);
      \iport (A) (0);
      \oport (B) (b0);
      \oport (G) (g0);
    \end{tikzpicture}
    ;&\boxtrans 1 {\set{B,G}}=
    &\quad\begin{tikzpicture}[baseline=(e1)]
      \state[*_2] (1) (0,0);
      \state[C] (c1) (1,.5);
      \state[E] (e1) (1,-.5);
      \state[G] (g1) ($(1|-e1)!.5!(e1)-(0,1)$);
      \edge (1) (c1)[c];
      \edge[below] (1) (e1)[b];
      \rect (g1) (1) (c1) (c1);
      \iport (B) (1);
      \oport (C) (c1);
      \oport (E) (e1);
      \iport (G) (g1)[1];
      \oport (G) (g1)[e1];
    \end{tikzpicture}\\
    \boxtrans 2 {\set{C,E,G}}\footnotemark=
    &\quad\begin{tikzpicture}[baseline=(e2)]
      \state[*_i] (2) (0,0);
      \state[D] (d2) (1,0);
      \state[E] (e2) ($(2)!.5!(d2)-(0,1)$);
      \state[G] (g2) ($(e2)-(0,1)$);
      \edge (2) (d2)[a];
      \rect (g2) (2) (d2) (d2);
      \iport (C) (2);
      \oport (D) (d2);
      \iport (E) (e2)[2];
      \oport (E) (e2)[d2];
      \iport (G) (g2)[2];
      \oport (G) (g2)[d2];
    \end{tikzpicture}
    &\boxtrans 3 {\set{D,E,G}}=
    &\quad\begin{tikzpicture}[baseline=(e3)]
      \state[*_4] (3) (0,0);
      \state[C] (c3) (1,.5);
      \state[E] (e3) (1,-.5);
      \state[G] (g3) ($(1|-e3)!.5!(e3)-(0,1)$);
      \edge (3) (c3)[c];
      \edge[below] (3) (e3)[b];
      \rect (g3) (3) (c3) (c3);
      \iport[.5] (D) (3);
      \iport[-.5] (E) (3);
      \oport (C) (c3);
      \oport (E) (e3);
      \iport (G) (g3)[3];
      \oport (G) (g3)[e3];
    \end{tikzpicture}\\
    \boxtrans 4 {\set{D,E,G}}=
    &\quad\begin{tikzpicture}[baseline=(anchor)]
      \state[*_6] (4) (0,0);
      \state[F] (f4) (1,0);
      \state[G] (g4) ($(1|-f4)!.5!(f4)-(0,1.5)$);
      \edge (4) (f4)[d];
      \rect (g4) (4) (f4) (f4);
      \iport[.5] (D) (4);
      \iport[-.5] (E) (4);
      \oport (F) (f4);
      \iport (G) (g4)[4];
      \oport (G) (g4)[f4];
      \coordinate (anchor) at ($(f4)!.5!(g4)$);
    \end{tikzpicture}
     &\boxtrans 7 {\set{F,G}}=
    &\quad\begin{tikzpicture}[baseline=(7)]
      \state[*_7] (7) (0,0);
      \rect[1.25] (7)(7)[1.25](7)(7);
      \iport[.75] (F) (7);
      \iport[-.75] (G) (7);
      \oport (f) (7);
    \end{tikzpicture}
  \end{align*}
  \footnotetext{Two instances of this box have to be used, one with
    $i=3$ and one with $i=5$.}

  \noindent Accordingly, the box of the run is the one below.
  \begin{align*}
    \boxrun R=
    & 
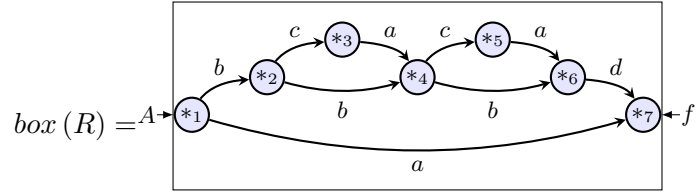
\begin{tikzpicture}[yscale=.5,baseline=(0.south)]
      \state[*_1] (0) (0,0);
      \state[*_2] (1) (1,1);
      \state[*_3] (2) (2,2);
      \state[*_4] (3) (3,1);
      \state[*_5] (2') (4,2);
      \state[*_6] (4) (5,1);
      \state[*_7] (7) (6,0);
      \edge[bend left,above] (0) (1)[b];
      \edge[bend right,below] (0) (7)[a];
      \edge[bend left,above] (1) (2)[c];
      \edge[bend right,below] (1) (3)[b];
      \edge[bend left,above] (2) (3)[a];
      \edge[bend left,above] (3) (2')[c];
      \edge[bend right,below] (3) (4)[b];
      \edge[bend left,above] (2') (4)[a];
      \edge[bend left,above] (4) (7)[d];
      \rect[2] (7)(0)[1](2)(7);
      \iport (A) (0);
      \oport (f) (7);
    \end{tikzpicture}
  \end{align*}
  The graph in this box is exactly the trace of the run
  (Figure~\ref{fig:ex:trace}). Also notice that the only vertices in
  the composite box are the $*_i$ vertices, one for each transition.
\end{exa}

%%% Local Variables:
%%% mode: latex
%%% TeX-master: "main"
%%% End:

This equivalence with traces yields another property: because of
Constraint~\ref{cstr:sp}, we know that every trace in $\A$ is
series-parallel. This means that the box of every run can be typed as
going from
$\tau_\iota=\begin{tikzpicture}[baseline=(0.south)] \state (0)
  (0,0);\state (1) (1,0); \edge (0) (1);\oport (\iota)[i] (1);
\end{tikzpicture}$
to
type~$\tau_f=\begin{tikzpicture}[baseline=(0.south)] \state (0)
  (0,0);\state (1) (1,0); \edge (0) (1);\oport (f)[f] (1);
\end{tikzpicture}$.
\begin{lem}\label{lem:runboxesaretypable}
  Let $R=\tuple{\set\iota,\tr[1];\dots;\tr[n],\emptyset}$ be an
  accepting run, with intermediary
  states ${\set\iota=S_0,S_1,\dots,S_{n-1},S_n=\emptyset}$. There
  exists a sequence $\tau_0,\dots,\tau_n$ of types over~$P$ such that
  $\tau_0=\tau_\iota$, $\tau_n=\tau_f$ and
  $\forall 1\leqslant i\leqslant n$,
  $\boxtrans{\tr[i]}{S_{i-1}}\in\tboxsetb{\tau_{i-1}}{\tau_i}$
\end{lem}
\begin{proof}
  Simply put, since $\Gr R$ is series-parallel, then for every $k$ it
  must be the case that
  $\Gr{\tuple{\set\iota,\tr[1];\dots;\tr[k],S_k}}$ SP-reduces to a
  tree $\tau_k$. Using Lemma~\ref{lem:boxtrace} we can use these trees
  to type the boxes of every transition in $R$.
 
  The fact
  that~$\tau_n=\tau_f=\begin{tikzpicture}[baseline=(0.south),xscale=.7]
    \state (0) (0,0);\state (1) (1,0); \edge (0) (1);\oport (f) (1);
  \end{tikzpicture}$
  is straightforward, as this is the only type over~$\set f$.
\end{proof}
Consider the finite state automaton $Aut$ with states $\typ$ (the set
of types over $P$), initial state~$\tau_\iota$, a single final
state~$\tau_f$ and two kinds of transitions. For every non-final
transition $\tr$, states $S,{S'}$ and types $\tau,\sigma$ such that
$\atrans S {S'}$ and $\boxtrans\tr S\in\tboxsetb\tau\sigma$, there is
a transition $\tuple{\tau,\boxtrans \tr S,\sigma}$ in~$Aut$. There are
also transitions $\tuple{\tau,\boxtrans\tr \trin,\tau_f}$ for every
final transition $\tuple{\trin,\emptyset}$ and every type $\tau$ over
$\trin$.

This automaton captures exactly the accepting runs of $\A$. Indeed,
Lemma~\ref{lem:runboxesaretypable} assures us that for every accepting
run in $\A$ we can find a corresponding accepting run in $Aut$. On the
other hand every accepting run in $Aut$ stems by construction from an
accepting run in~$\A$.

This means we can extract a regular expression $e$ with letters in
$\tboxset$ from this automaton using the standard Kleene theorem. We
can also get the language of $\A$ by extracting the graphs of boxes
$\beta_1\odot\cdots\odot\beta_n$, whenever $\beta_1\dots\beta_n$ is a
word in the language of $e$.  However, this is not yet what we are
looking for, that is, a simple expression over $Y$. To get this
expression, we replay the standard proof of Kleene's theorem, with
some modifications to suit our needs.

\begin{rem}\label{sec:regul-lang-runs}
  The automaton above is constructed from the state $\tau_\iota$, by
  exploring all initial runs. This construction fails if either
  Constraint~\ref{cstr:safety} or Constraint~\ref{cstr:sp} are not
  satisfied. This shows that the two constraints we imposed are
  decidable.
\end{rem}

\subsubsection{Computing the expression.}
\label{sec:comp-expr}

A standard way of proving Kleene's theorem consists in moving from
automata to \emph{generalised automata}, that is automata with regular
expressions labelling transitions. We perform a similar step here, by
replacing boxes with box templates in transition labels. The
transformation is straightforward: for every pair of
states~$\sigma,\tau$, we put in the generalised automaton a transition
labelled with~$\setcompr\beta{\tuple{\sigma,\beta,\tau}\in Aut}$. This
ensures that there is exactly one transition between every pair of
states. Notice that in the resulting generalised automaton, if the
transition from~$\sigma$ to~$\tau$ is labelled with~$\Gamma$, all
boxes in~$\Gamma$ have type~$\sigma\to\tau$, therefore~$\Gamma$ itself
has this type.

Then, we proceed to remove non-initial non-final states, in an
arbitrary order. Notice that the automaton~$Aut$ we are considering
has a single initial state and a single final state,
respectively~$\tau_\iota$ and~$\tau_f$. These states are distinct,
there is no outgoing transition from~$\tau_f$, nor is there an
incoming transition in~$\tau_\iota$ (remember that~$\iota$ does not
appear in the output of transitions). At the end of this procedure,
the only remaining states will thus be~$\tau_\iota$
and~$\tau_f$. There will be exactly one transition from~$\tau_\iota$
to~$\tau_f$ of the
shape~$\tuple{\tau_\iota,\set{\beta_1,\cdots,\beta_n},\tau_f}$. As
every~$\beta_i$ has type~$\tau_\iota\to\tau_f$, its graph~$G_i$ must be
series-parallel, thus we will get that
$e\eqdef\trm{G_1}\sumtm\dots\sumtm\trm{G_n}$ is a simple expression such that
$\rlang{\set{\beta_1,\cdots,\beta_n}}=\Gr{e}$.

We maintain an invariant throughout the construction: the boxes
generated by templates labelling each transition of the automaton
should stem from runs in $\A$. More precisely, we require every
template labelling a transition to be $\A$-valid:
\begin{defi}[$\A$-validity]
  A template $\gamma\in\gtboxset\tau\sigma$ is \emph{$\A$-valid} if
  for every $\beta\in\rlang{\gamma}$ there exists a run $R$ in $\A$
  such that $\boxrun{R}=\beta$.
\end{defi}

Now we only need to show how to remove one state, preserving the
language of the automaton. The idea when removing a state $\tau$ is to
add transitions to replace every run going through $\tau$. For every
pair of states $\sigma,\chi$ with transitions labelled with
$\beta,\delta,\gamma$ and $\Gamma'$ as below, we will define a
template $\Gamma$. We will then replace $\Gamma'$ with
$\Gamma\sumtm\Gamma'$ on the transition going from $\sigma$ to $\chi$.
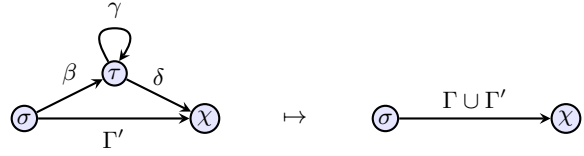
\begin{center}
  \begin{tikzpicture}[yscale=.6,xscale=1.2]
    \state[\sigma] (s) (0,0);
    \state[\tau] (t) (1,1);
    \coordinate (loop) at ($(t)+(0,1)$);
    \node[above] () at (loop) {$\gamma$};
    \state[\chi] (c) (2,0);
    \edge (s) (t)[\beta];
    \draw[arc] (t) to[out=110,in=180] (loop)
    to[out=0,in=70] (t);
    \edge (t) (c)[\delta];
    \edge[below] (s) (c)[\Gamma'];

    \state[\sigma] (s1) (4,0);
    % \state[\tau] (t1) (5,1);
    % \coordinate (loop1) at ($(t1)+(0,1)$);
    % \node[above] () at (loop1) {$\gamma$};
    \state[\chi] (c1) (6,0);
    % \edge (s1) (t1)[\beta];
    % \draw[arc] (t1) to[out=110,in=180] (loop1)
    % to[out=0,in=70] (t1);
    % \edge (t1) (c1)[\delta];
    \edge (s1) (c1)[\Gamma\sumtm\Gamma'];

    \node () at ($(c)!.5!(s1)$) {$\mapsto$};
  \end{tikzpicture}
\end{center}
We would like $\Gamma$ to be
$\beta\cdot\gamma^\star\cdot\delta$. However, remember that we can
only compute the star of atomic templates.  But in this case, we can
approximate $\gamma$ with a good-enough atomic template:
\begin{lem}\label{lem:exists-star}
  For every $\A$-valid template $\gamma\in\gtboxset\tau\tau$ there
  exists an $\A$-valid atomic template $At(\gamma)\in\gtboxset\tau\tau$
  such that $\rlang{\gamma}\subseteq \rlang{At(\gamma)^\star}$.
\end{lem}
\begin{proof}
  From each connected component of the graph of each box in $\gamma$
  stems an $\A$-valid atomic box. Every box in $\gamma$ is equal to
  the product (in any order) of the boxes corresponding to its
  connected components. We then take $At(\gamma)$ to be the set of all
  these atomic boxes.
\end{proof}
We then define $\Gamma$ to be
$\beta\cdot At(\gamma)^\star\cdot\delta$.  Hence we get
$\rlang{\beta}\odot\rlang{\gamma}^\star\odot\rlang\delta\subseteq
\rlang{\Gamma}$. The lemma also ensures that for every graph produced
by a run in the automaton where $\tau$ is removed, there is a run in
$\A$ yielding the same graph. Both theses properties allow us to
conclude that this step is valid, preserving both the invariant and
the language of the automaton.

\begin{prop}
  \label{prop:ok_a_to_e}
  There exists a function $\E\colon \PA[Y]\to\SExp[Y]$ such that 
  for every Petri automaton $\A$, we have
  $\Gr{\E\paren\A}=\Gr\A$.
\end{prop}

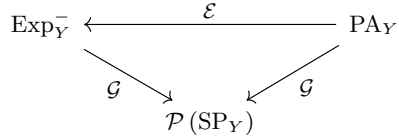
\begin{center}
  \begin{tikzcd}
    \SExp[Y] \arrow[dr,"\G"']&& \PA[Y]\arrow[dl,"\G"]\arrow[ll,"\E"']\\ &\pset{\SP[Y]}
  \end{tikzcd}
\end{center}

\begin{thm}
  \label{thm:kl2:sp}
  Recognisable sets of series-parallel graphs are regular.
\end{thm}

\noindent As in Section~\ref{sec:exp:pa}, we can use
Proposition~\ref{prop:type:graphs} to extend this result to sets of
arbitrary graphs: the following diagram commutes.
\begin{center}
  \begin{tikzcd}
    \Exp\arrow[ddrr,"\G"']&&
    \SExp\arrow[ll,"\retype\argument"'] \arrow[dr,"\G"']&&
    \PA[\Xb]\arrow[dl,"\G"]\arrow[ll,"\E"']\\ 
    &&&\pset{\SP}\arrow[dl,"\retype\argument"]\\
    &&\pset{\Gph}
  \end{tikzcd}
\end{center}
\noindent Whence
\begin{thm}
  \label{thm:kl2}
  Recognisable sets of graphs are regular.
\end{thm}

\section{Reading graphs modulo homomorphism}
\label{sec:reading}

In Section~\ref{sec:exp:pa} we have shown how to associate a Petri
automaton $\A\paren e$ with every expression $e$, such that
$\Gr{\A\paren e}=\Gr e$. However, remember from
Section~\ref{sec:graphs} that the set of graphs we are interested in
for representatble Kleene allegories is not directly $\Gr e$, but its
downward closure $\clgr{\Gr e}$ w.r.t.\ homomorphism
(Theorem~\ref{thm:interlang}). Given a Petri automaton $\A\in\PA$, we
show how to read the graphs in $\clgr{\retype {\Gr\A}}$ in a local
and incremental way.

\begin{defi}[Reading, language of a run]\label{def:read}
  A \emph{reading} of $G=\tuple{V,E,\iota,o}$ along a run
  $\run=\tuple{S_0,\tr[0];\dots;\tr[n],S_{n+1}}$ (with intermediary
  states $S_1,\dots$) is a sequence
  $(\rho_k)_{0\leqslant k\leqslant {n+1}}$ such that for all $k$,
  $\rho_k$ is a map from $S_k$ to $V$, $\rho_0(S_0)=\set {\iota}$, and
  $\forall 0\leqslant k\leqslant n$, the following holds:
  \begin{itemize}
  \item all tokens in the input of the transition are mapped to the
    same vertex in the graph:
    \[\forall p,q\in \trin[k],\ \rho_{k}(p)= \rho_{k}(q);\]
  \item a final transition may only be fired from the output of a graph:
    \[\trout[k]=\emptyset\Rightarrow \forall
      p\in\trin[k],\rho_{k}(p)=o;\]
  \item the images of tokens in $S_{k+1}$ that are not in the input of
    the transition are unchanged:
    \[\forall p\in S_{k+1} \setminus \trin[k],\ \rho_{k}(p)=
      \rho_{k+1}(p);\]
  \item each pair in the output of the transition can be validated
    by the graph: 
    \begin{align*}
      \forall p\in \trin[k],~\forall \tuple{x,q}\in \trout[k],~
      \tuple{\rho_{k}(p),x, \rho_{k+1}(q)}\in E.
    \end{align*}
  \end{itemize}
  The \emph{language of a run} $\run$, denoted by $\Ln\run$, is the set
  of graphs that can be read along~$\run$.
\end{defi}

\begin{exa}
  The following automaton on the left-hand side can read the
  graph on the right-hand side.
  
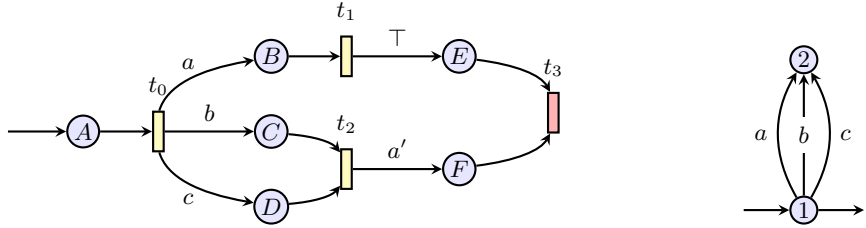
\begin{center}
    \begin{tikzpicture}[xscale=2.5]
      \state[A](A)(0,0);\state[B](B)(1,1);
      \state[C](C)(1,0);\state[D](D)(1,-1);
      \state[E](E)(2,1);\state[F](F)(2,-.5);
      \transb[t_0](0)(0.4,0);\transb[t_1](1)(1.4,1);\transb[t_2](2)(1.4,-.5);
      \transfb[t_3](3)(2.5,.25);
      \edge(A)(0);\edge[bend left,above](0)(B)[a];\edge(0)(C)[b];
      \edge[bend right,below](0)(D)[c];
      \edge(B)(1);\edge(1)(E)[\top];
      \edge[bend left](C)(2);\edge[bend right](D)(2);\edge(2)(F)[a'];
      \edge[bend left](E)(3);\edge[bend right](F)(3);
      \edge($(A)-(.4,0)$)(A);
    \end{tikzpicture}
    \hspace{2cm}
    \begin{tikzpicture}
      \state[1](1)(0,0);\state[2](2)(0,2);
      \initst(1);\fnst(1);
      \edge[bend left,left](1)(2)[a];
      \edge[][fill=white](1)(2)[b];\edge[bend right,right](1)(2)[c];
    \end{tikzpicture}
  \end{center}
  Let $R$ be the run $\tuple{\set A,t_0;t_1;t_2;t_3,\emptyset}$. To
  read the graph above along $R$, one should use the following reading:
  \begin{equation*}
    \left[A\mapsto 1\right]; 
    \left[
      \begin{array}{c}
        B\mapsto 2\\
        C\mapsto 2\\
        D\mapsto 2
      \end{array}
    \right]; 
    \left[
      \begin{array}{c}
        E\mapsto 1\\
        C\mapsto 2\\
        D\mapsto 2
      \end{array}
    \right]; 
    \left[
      \begin{array}{c}
        E\mapsto 1\\
        F\mapsto 1
      \end{array}
    \right];
    \left[\right].
  \end{equation*}
  One can easily check that this is indeed a valid reading along $R$.
\end{exa}

\noindent
The language of a Petri automaton is finally obtained by considering
all accepting runs.

\begin{defi}[Language of a Petri automaton]\label{def:lang}
  The \emph{language of $\A$}, written $\Ln\A$, is the set of graphs
  readable from accepting runs:
  \begin{align*}
    \Ln\A&\eqdef\bigcup_{R\in\aRun}\Ln\run.\qedhere
  \end{align*}
\end{defi}

\noindent 
To avoid confusions between the languages $\Ln\A$ and $\Gr\A$, we
write ``$G$ is produced by $\A$'' when $G\in\Gr\A$, reserving language
theoretic terminology like ``$G$ is accepted by $\A$'' to cases where
we mean $G\in\Ln\A$.

\begin{lem}\label{lem:langex}
  For every accepting run $\run$, we have $G\in\Ln\run$ if and only if
  $G\lessgr\retype{\Gr\run}$.
\end{lem}
\begin{proof} Let us fix a Petri automaton~$\A=\tuple{P,\T,\iota_\A}$.
  Let $G=\tuple{V,E,\iota,o}$ and
  $\run=\tuple{\set\iota,\tr[0];\dots;\tr[n],\emptyset}$, with
  intermediary states $S_1,\dots$.

  Remember that for every $k$ and~$p\in\settrout[k]$, we defined
  (Definition~\ref{def:trace}):
  $$\nu\paren{k,p}=\setcompr{l}{l>k\text{ and }p\in\trin[l]}.$$
  $\Gr\run$, is then the graph with vertices~$\set{0,\dots,n}$
  and the set of edges defined by:
  $$E_\run=\setcompr{\tuple{k,a,l}}{\tuple{a,p}\in\trout[k]
    \text{ and }l=\min{\nu\paren{k,p}}}. $$
  Finally, we get
  $\retype{\Gr\run}=\tuple{\setcompr{[i]}{0\leqslant i\leqslant
      j},E',[0],[n]}$.
  (It is easy to check that the source of $\Gr\run$ is the vertex $0$,
  and that its sink is $n$.)
   
  It will prove convenient in the following to use the notation
  $\Nu(k,p)=\min\nu(k-1,p)=\min\setcompr{l}{l\geqslant k\text{ and
    }p\in\trin[l]}$. Notice that
  $\forall k,p\in S_k, k\leqslant\Nu(k,p)\leqslant n$, and that
  whenever $p\in\trin[k]$, we have $\Nu(k,p)=k$.

  Suppose there exists a graph homomorphism $\phi$ from
  $\retype{\Gr{\run}}$ to $G$. We build a reading $(\rho_k)_k$ of $G$
  along $\run$ by letting
  $\rho_k(p)\eqdef \phi([\Nu(k,p)])$ for
  $0\leqslant k\leqslant n\text{ and }p\in S_k$. We now have to
  check that $\rho$ is truly a reading of $G$ in $\A$:
  \begin{itemize}
  \item for the initialisation of the reading:
    \begin{align*}
      \rho_0(\iota_\A)
      &=\phi([\Nu(0,\iota_\A)])\tag{by definition}\\
      &=\phi([0])=\set\iota\tag{$\phi$ is a homomorphism}
    \end{align*}
  \item for the final transition:
    \begin{align*}
      p\in S_{n},~\rho_{n}(p)
      &=\phi([\Nu(n,p)])\\
      &=\phi([n])=\set{o}.\tag{$\phi$ is a homomorphism}
    \end{align*} 
  \item for all
    $p\in \trin[k],\rho_{k}(p)=\phi([\Nu(k,p)])=\phi([k])$
    which does not depend on $p$.
  \item for all $p\in S_{k+1}\setminus \trin[k]$, we have
    $\Nu(k,p)=\Nu(k+1,p)$ (since $p\notin \trin[k]$). Hence 
    \begin{align*}
      \rho_{k}(p)=\phi([\Nu(k,p)])
      &=\phi([\Nu(k+1,p)])\\
      &=\rho_{k+1}(p).
    \end{align*}

  \item for all $p\in \trin[k]$ and $\tuple{x,q}\in \trout[k]$, we know that
    $\rho_{k}(p)=\phi([k])$ and that
    $\tuple{k,x,\Nu(k+1,q)}\in E_{\run}$.
    \begin{itemize}
    \item If $x\in X$, we also have
      $\tuple{[k],x,[\Nu(k+1,q)]}\in E'$. Because $\phi$ is a
      homomorphism we can deduce that:
      \[\tuple{\phi([k]),x,\phi([\Nu(k+1,q)])}\in E,\]
      which can be rewritten
      $\tuple{\rho_k(p),x,\rho_{k+1}(q)}\in E$.
    \item If $x=y',\ y\in X$, we also have
      $\tuple{[\Nu(k+1,q)],y,[k]}\in E'$. Because $\phi$ is a
      homomorphism we can get like before
      $\tuple{\rho_{k+1}(q),y,\rho_k(p)}\in E$.
    \item If finally $x=1$, then we know that
      $k\equiv\Nu(k+1,q)$, thus proving
      that
      \[\rho_k(p)=\phi([k])=\phi([\Nu(k+1,q)])=\rho_{k+1}(q).\]
    \end{itemize}

  \end{itemize}

  If on the other hand we have a reading
  $(\rho_k)_{0\leqslant k\leqslant n}$ of $G$, we define
  $\phi: \set{0,\dots,n}\to V$ by $\phi([k])\eqdef\rho_k(p)$ for
  all $p\in \trin[k]$. As $(\rho_k)_k$ is a reading, $\phi$ is well
  defined\footnote{It is not difficult to check that
    $k\equiv l\Rightarrow \forall \tuple{p,q}\in \trin[k]\times
    \trin[l], \rho_k(p)=\rho_l(q)$.}. Let us check that $\phi$ is a
  homomorphism from $\retype{\Gr\run}$ to $G$:
  \begin{itemize}
  \item $\phi([0])=\rho_0(\iota_\A)=\iota$;
  \item $\phi([n])=\rho_n(p)$ with $p\in S_n$, and since
    $(\rho_k)_k$ is a reading and $\tr[n]$ is final, $\rho_n(p)=o$.
  \item if $\tuple{[k],x,[l]}\in E'$ is an edge of $\retype{\Gr\run}$, then
    it was produced from some edge $(i,y,j)\in E_\run$, with either
    $x=y$ and $\tuple{i,j}\in[k]\times[l]$ or $y=x'$ and
    $\tuple{i,j}\in[l]\times[k]$. There is some $p\in \trin[{i}]$ and
    $q$ such that $\tuple{y,q}\in \trout[{j}]$ and $j=\Nu(i+1,q)$.

    By definition of $\Nu$ we know that $\forall{i+1\leqslant}m <j,
    q\notin \trin[m]$. Thus, because $(\rho_k)_k$ is a reading,
    $\rho_{i+1}(q)=\rho_j(q)$ and either
    $\tuple{\rho_i(p),x,\rho_{i+1}(q)}\in E$ or
    $\tuple{\rho_{i+1}(q),x,\rho_i(p)}\in E$, and thus
    $\tuple{\phi([k]),x,\phi([l])}\in E$.\qedhere
  \end{itemize}
  % Suppose there exists a graph homomorphism $\phi$ from
  % $\retype{\Gr{\run}}$ to $G$. Then we can build a reading by defining
  % $\rho_k(p)=\phi([\min\nu(k,p)]_\run)$ for $0\leqslant k\leqslant n$
  % and $p\in S_k$.  On the other hand, if we have a reading
  % $(\rho_k)_{0\leqslant k\leqslant n}$ of $G$, we can build a
  % homomorphism $\phi$ by letting $\phi([k]_\run)=\rho_k(p)$ for
  % any $p\in \trin[k]$. As $(\rho_k)_k$ is a reading, $\phi$ is well
  % defined.
\end{proof}

\noindent 
As an immediate consequence, we obtain the following characterisations of
the language of a Petri automaton.
\begin{thm}\label{thm:petri-automata}
  For every Petri automaton $\A\in\PA$, we have $\Ln\A=\clgr{\retype{\Gr\A}}$.
\end{thm}
\begin{thm}\label{thm:petri-automata:simple}
  For every Petri automaton $\A\in\PA[X]$, we have $\Ln\A=\clgr{\Gr\A}$.
\end{thm}

\noindent The left-hand side language is defined through readings
along accepting runs, which is a local and incremental notion and
which allows us to define \emph{simulations} in the following
section. By contrast, the right-hand side language is defined
globally.

\section{Comparing Petri automata modulo homomorphism}
\label{sec:compare}

While we can use Petri automata over $\Xb$ to read arbitrary graphs,
and thus reason about arbitrary expressions, we do not know how to
compare the languages of these automata in general. We thus restrict
to Petri automata over $X$ and to simple expressions. At the algebraic
level, this means we consider identity-free Kleene lattices.  We prove
in this section that the containment problem is decidable.

Assembling the results from Sections~\ref{sec:graphs},
\ref{sec:exp:pa} and~\ref{sec:reading}, we have the following
proposition.
\begin{prop}
  \label{prop:reduce}
  For all simple expressions $e,f\in\SExp[X]$, the following are
  equivalent:
  \begin{enumerate}[label=(\roman*)]
  \item\label{i2:rel} $\Rel\models e \leq f$,
  \item\label{i2:aut} $\Gr{\A(e)} \subseteq \Ln{\A(f)}$.
  \end{enumerate}
  For all Petri automata $\A,\B\in\PA[X]$, the following are
  equivalent:
  \begin{enumerate}[label=(\roman*)]
  \item\label{i3:rel} $\Ln\A\subseteq\Ln\B$,
  \item\label{i3:aut} $\Gr\A\subseteq\Ln\B$.
  \end{enumerate}
\end{prop}
\begin{proof}
  We have
  \begin{align*}
    & \Rel\models e \leq f\\
    \Leftrightarrow~ \tag{Theorem~\ref{thm:interlang}}
    & \Gr e \subseteq \clgr{\Gr f} \\
    \Leftrightarrow~ \tag{Proposition~\ref{prop:ok_e_to_a}}
    & \Gr{\A(e)} \subseteq \clgr{\Gr{\A(f)}} \\
    \Leftrightarrow~ \tag{Theorem~\ref{thm:petri-automata:simple}}
    & \Gr{\A(e)} \subseteq \Ln{\A(f)}
  \end{align*}
  The second point follows from
  Theorem~\ref{thm:petri-automata:simple} and the fact that
  $S \subseteq \clgr T$ iff $\clgr S \subseteq \clgr T$.
\end{proof}
\noindent
As a consequence, to decide the (in)equational theory of identity-free
Kleene lattices, or to compare two Petri automata, it suffices to find
a way to decide whether $\Gr\A\subseteq\Ln\B$, given two Petri
automata $\A$ and $\B$.

\subsection{Intuition}
\label{sec:intuitions}

In this section, we show how the notion of simulation relation, that
allows one to compare NFA, can be adapted to handle simple Petri
automata. Consider two automata $\A_1=\tuple{P_1,\T_1,\iota_1}$ and
$\A_2=\tuple{P_2,\T_2,\iota_2}$, we try to show that for every accepting
run $\run$ in $\A_1$, $\Gr{\run}$ is recognised by some accepting run
$\run'$ in $\A_2$. Leaving non-determinism aside, the first idea that
comes to mind is to find a relation between the states in $\A_1$ and
the states in $\A_2$, that satisfy some conditions on the initial and
final states, and such that if $S_k\duck S'_k$ and
$\atrans[\tr][\A_1]{S_k}{S_{k+1}}$, then there is a state $S'_{k+1}$
in $\A_2$ such that $S_{k+1}\duck S'_{k+1}$,
$\atrans[\tr*][\A_2]{S'_k}{S'_{k+1}}$, and these transition steps are
compatible in some sense. However, such a definition will not give us
the result we are looking for. Consider these two runs:
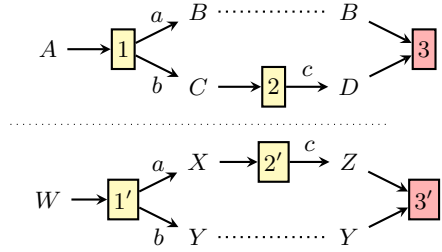
\begin{center}
  \begin{tikzpicture}[baseline=(a.north)]
    \node (a) at (0,1) {$A$};
    \node (b1) at (2,1.5) {$B$};
    \node (b2) at (4,1.5) {$B$};
    \node (c) at (2,0.5) {$C$};
    \node (d) at (4,0.5) {$D$};
    \trans[1] (1) (1,1);
    \trans[2] (2) (3,0.5);
    \edge (a) (1);
    \edge[above] (1) (b1)[a];
    \edge[below] (1) (c)[b];
    \edge (c) (2);
    \draw[thick,dotted] (b1) to (b2);
    \edge[above] (2) (d)[c];
    \draw[dotted] (-0.5,0) -- (4.5,0);
    \node (w) at (0,-1) {$W$};
    \node (y1) at (2,-1.5) {$Y$};
    \node (y2) at (4,-1.5) {$Y$};
    \node (x) at (2,-0.5) {$X$};
    \node (z) at (4,-0.5) {$Z$};
    \trans[1'] (3) (1,-1);
    \trans[2'] (4) (3,-0.5);
    \edge (w) (3);
    \edge[below] (3) (y1)[b];
    \edge[above] (3) (x)[a];
    \edge (x) (4);
    \draw[thick,dotted] (y1) to (y2);
    \edge[above] (4) (z)[c];
    \transf[3] (f1) ($(b2)!.5!(d)+(1,0)$);
    \transf[3'] (f2) ($(y2)!.5!(z)+(1,0)$);
    \edge(b2) (f1);\edge(d)(f1);
    \edge(y2) (f2);\edge(z)(f2);
  \end{tikzpicture}
\end{center}
The graphs produced by the first and the second runs correspond
respectively to the simple terms $a\cap(b\cdot c)$ and
$(a\cdot c)\cap b$. These two terms are incomparable, but the relation
$\duck$ depicted below satisfies the previously stated conditions.

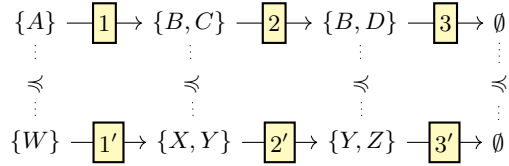
\begin{center}
  \begin{tikzpicture}
    \node (A) at (0,0) {$\set A$};
    \node[right=of A] (BC) {$\set{B,C}$};
    \node[right=of BC] (BD) {$\set{B,D}$};
    \node[right=of BD] (f1) {$\emptyset$};
    \node[below=of A] (W) {$\set W$};
    \node[right=of W] (XY) {$\set{X,Y}$};
    \node[right=of XY] (YZ) {$\set{Y,Z}$};
    \node[right=of YZ] (f2) {$\emptyset$};
    \draw[dotted] (A) to node[midway,fill=white] {$\duck$} (W);
    \draw[dotted] (BC) to node[midway,fill=white] {$\duck$} (XY);
    \draw[dotted] (BD) to node[midway,fill=white] {$\duck$} (YZ);
    \draw[dotted] (f1) to node[midway,fill=white] {$\duck$} (f2);
    \path[->] (A) edge node[midway,trans] {$1$} (BC)
    (W) edge node[midway,trans] {$1'$} (XY)
    (BC) edge node[midway,trans] {$2$} (BD)
    (XY) edge node[midway,trans] {$2'$} (YZ)
    (BD) edge node[midway,trans] {$3$} (f1)
    (YZ) edge node[midway,trans] {$3'$} (f2);
  \end{tikzpicture}
\end{center}
% \[\SelectTips{cm}{}\xymatrix@R=3ex@C=8ex{
%   \set{A}\ar[r]^-{\fbox 1}&\set{B,C}\ar[r]^-{\fbox 2}&\set {B,D}\\
%   \set{W}\ar@{-}[u]|\duck\ar[r]^-{\fbox
%     {1'}}&\set{X,Y}\ar@{-}[u]|\duck\ar[r]^-{\fbox {2'}}&\set
%   {Y,Z}\ar@{-}[u]|\duck}\]

The problem here is that in Petri automata, runs are token firing
games. To adequately compare two runs, we need to closely track the
tokens. For this reason, we will relate a state $S_k$ in
$\A_1$ not only to a state $S'_k$ in $\A_2$, but to a map
$\eta_k$ from $S'_k$ to $S_k$. This will enable us to associate
with each token situated on some place in $P_2$ another token placed
on $\A_1$.

We want to find a reading of $\Gr{\run}$ in $\A_2$, \ie,~a run in
$\A_2$ together with a sequence of maps associating places in $\A_2$
with vertices in $\Gr{\run}$. Consider the picture below. Since we
already have a reading of $\Gr\run$ along $\run$ (by defining
$\rho_k(p)=\Nu(k,p)$, as in the proof of Lemma~\ref{lem:langex}), it
suffices to find maps from the places in $\A_2$ to the places in
$\A_1$ (the maps $\eta_k$): the reading of $\Gr{\run}$ in $\A_2$ will
be obtained by composing $\eta_k$ with $\rho_k$.
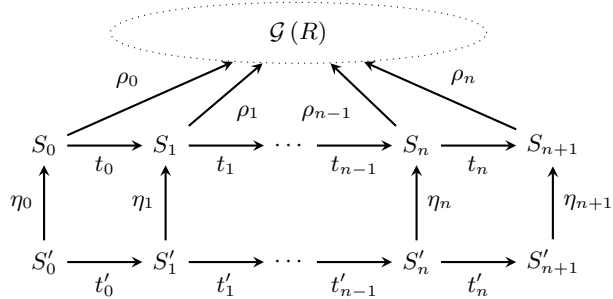
\begin{center}
  \begin{tikzpicture}
    \node (a0) at (0,0) {$S_0$};
    \node[right=of a0] (a1) {$S_1$};
    \node[right=of a1] (a2) {$\cdots$};
    \node[right=of a2] (a3) {$S_n$};
    \node[right=of a3] (a4) {$S_{n+1}$};
    \node[below=of a0] (b0) {$S'_0$};
    \node[right=of b0] (b1) {$S'_1$};
    \node[right=of b1] (b2) {$\cdots$};
    \node[right=of b2] (b3) {$S'_n$};
    \node[right=of b3] (b4) {$S'_{n+1}$};
    \edge[below] (a0)(a1)[\tr[0]];
    \edge[below] (a1)(a2)[\tr[1]];
    \edge[below] (a2)(a3)[\tr[n-1]];
    \edge[below] (a3)(a4)[\tr[n]];
    \edge[below] (b0)(b1)[\tr*[0]];
    \edge[below] (b1)(b2)[\tr*[1]];
    \edge[below] (b2)(b3)[\tr*[n-1]];
    \edge[below] (b3)(b4)[\tr*[n]];
    \edge[left](b0)(a0)[\eta_0];
    \edge[left](b1)(a1)[\eta_1];
    \edge[right](b3)(a3)[\eta_{n}];
    \edge[right](b4)(a4)[\eta_{n+1}];

    \node[draw,dotted,ellipse,minimum width=5cm] 
    (G) at ($(a0)!.5!(a4)+(0,1.5)$) {$\Gr\run$};
    \edge[above left] (a0)(G)[\rho_0];
    \edge[below right] (a1)(G)[\rho_1];
    \edge[below left] (a3)(G)[\rho_{n-1}];
    \edge[above right] (a4)(G)[\rho_n];
  \end{tikzpicture}
\end{center}

% \[
% \SelectTips{cm}{}\xymatrix@=2ex{ &&&&\save
%   [llll].[rrrr]*+++[o][F.]\frm{}!C="g1"\restore{\Gr{\xi}}
%   &&&&\\
%   &&&&&&&&\\
%   &&&&&&&&\\
%   \xi_0\ar[rr]_{\tr_0}&&\xi_1\ar[rr]_{\tr_1}&&\ \cdots\ \ar[rr]_{\tr_{n-1}}&&\xi_n\ar[rr]_{\tr_n}&&\xi_{n+1}\\
%   &&&&&&&&\\
%   \xi'_0\ar[uu]^{\eta_0}\ar[rr]_{\tr_0'}&&\xi'_1\ar[rr]_{\tr_1'}\ar[uu]^{\eta_1}&&\
%   \cdots\
%   \ar[rr]_{\tr'_{n-1}}&&\xi'_n\ar[uu]_{\eta_n}\ar[rr]_{\tr'_n}&&\xi'_{n+1}\ar[uu]_{\eta_{n+1}}
%   \ar "4,1";"g1" ^-{\rho_0} \ar "4,3";"g1" _-{\rho_1} \ar "4,7";"g1"
%   ^-{\rho_n} \ar "4,9";"g1" _-{\rho_{n+1}} }\] 

We need to impose some constraints on the maps $(\eta_k$) to ensure
that $\paren{\rho_k\circ\eta_k}_{0\leqslant k\leqslant n}$ is indeed a
correct reading in $\A_2$. First, we need to ascertain that a
transition $\tr*[k]$ in $\A_2$ may be fired from the reading state
$\rho_k\circ\eta_k$ to reach the reading state
$\rho_{k+1}\circ\eta_{k+1}$. Furthermore, as for NFA, we want
transitions $\tr[k]$ and $\tr*[k]$ to be related: specifically, we
require $\tr*[k]$ to be included (via the homomorphisms $\eta_k$ and
$\eta_{k+1}$) in the transition $\tr[k]$. This is meaningful because
transition $\tr[k]$ contains a lot of information about the vertex $k$
of $\Gr{\run}$ and about $\rho$: the labels of the outgoing edges from
$k$ are the labels on the output of $\tr[k]$, and the only places that
will ever be mapped to $k$ in the reading $\rho$ are exactly the
places in the input of~$\tr[k]$.

This already shows an important difference between the simulations for
NFA and Petri automata. For NFA, we relate a transition
$p\xrightarrow a p'$ to a transition $q\xrightarrow a q'$ with the
same label $a$. Here the transitions
${\atrans[\tr[k]][\A_1]{S_k}{S_{k+1}}}$ and
${\atrans[\tr*[k]][\A_2]{S'_k}{S'_{k+1}}}$ may have different
labels. Consider the step represented below, corresponding to a square
in the above diagram.  The output of transition ${0}$ has a label $b$
that does not appear in ${0'}$, and ${0'}$ has two outputs labelled by
$a$. Nevertheless this satisfies the conditions informally stated
above, indeed, $a\cap b\leq a\cap a$ holds.
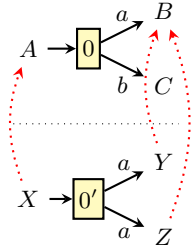
\begin{center}
  \begin{tikzpicture}[baseline=(a.north)]
    \node (a) at (0.2,1) {$A$}; \node (b) at (2,1.5) {$B$}; \node (c)
    at (2,0.5) {$C$}; \trans[0] (1) (1,1); \edge (a) (1); \edge[above]
    (1) (b)[a]; \edge[below] (1) (c)[b]; \draw[dotted] (0,0) --
    (2.2,0); \node (x) at (0.2,-1) {$X$}; \node (y) at (2,-0.5) {$Y$};
    \node (z) at (2,-1.5) {$Z$}; \trans[0'] (3) (1,-1); \edge (x) (3);
    \edge[below] (3) (z)[a]; \edge[above] (3) (y)[a];

    \draw[arc,dotted,color=red] (x) to[out=110,in=-110] (a) ;
    \draw[arc,dotted,color=red] (y) to[out=110,in=-110] (b) ;
    \draw[arc,dotted,color=red] (z) to[out=70,in=-70] (b) ;
  \end{tikzpicture}
\end{center}

However this definition is not yet satisfactory. Consider the two
runs below:
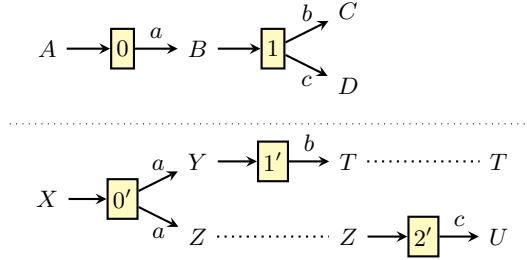
\begin{center}
  \begin{tikzpicture}[baseline=(a.north)]
    \node (a) at (0,1) {$A$};
    \node (b) at (2,1) {$B$};
    \node (c) at (4,1.5) {$C$};
    \node (d) at (4,0.5) {$D$};
    \trans[0] (0) (1,1);
    \trans[1] (1) (3,1);
    \edge (a) (0);
    \edge[above] (0) (b)[a];
    \edge (b) (1);
    \edge[above] (1) (c)[b];
    \edge[below] (1) (d)[c];
    \draw[dotted] (-0.5,0) -- (6.5,0);
    \node (x) at (0,-1) {$X$};
    \node (y) at (2,-0.5) {$Y$};
    \node (z1) at (2,-1.5) {$Z$};
    \node (z2) at (4,-1.5) {$Z$};
    \node (t1) at (4,-0.5) {$T$};
    \node (t2) at (6,-0.5) {$T$};
    \node (u) at (6,-1.5) {$U$};
    \trans[0'] (2) (1,-1);
    \trans[1'] (3) (3,-0.5);
    \trans[2'] (4) (5,-1.5);
    \edge (x) (2);
    \edge[below] (2) (z1)[a];
    \edge[above] (2) (y)[a];

    \edge (y) (3);
    \draw[thick,dotted] (z1) to (z2);
    \edge[above] (3) (t1)[b];

    \edge (z2) (4);
    \draw[thick,dotted] (t1) to (t2);
    \edge[above] (4) (u)[c];
  \end{tikzpicture}
\end{center}
Their produced graphs correspond respectively to the simple terms
$a\cdot(b\cap c)$ and $(a\cdot b)\cap (a\cdot c)$. The problem is that
$a\cdot(b\cap c)\leq (a\cdot b)\cap (a\cdot c)$, but with the previous
definition, we cannot relate these runs: they do not have the same
length. The solution here consists in grouping the transitions ${1'}$
and ${2'}$ together, and considering these two steps as a single step
in a \emph{parallel run}. This last modification gives us a notion of
simulation that suits our needs.

\subsection{Simulations}
\label{sec:simulations}

Before getting to the notion of simulation, we need to define what is
a parallel run, and a parallel reading.

A set of transitions $T \subseteq\T$ is \emph{compatible} if their
inputs are pairwise disjoint. If furthermore all transitions in $T$
are enabled in a state $S$, one can observe that the
state $S'$ reached after firing them successively does not
depend on the order in which they are fired. In that case we write
$\atrans[T]S{S'}$.

A \emph{parallel run} is a sequence
$\mathcal R=\tuple{S_0,T_0;\dots;T_n,S_{n+1}}$, where the
$T_k\subseteq \T$ are compatible sets of transitions such that
$\atrans[T_k]{S_{k}}{S_{k+1}}$. We define a \emph{parallel reading}
$\rho$ along some parallel run
$\left.\mathcal R=\tuple{S_0,T_0;\dots;T_n,\emptyset}\right.$ by
requiring that: $\rho_0(S_0)=\set {\iota}$, $\rho_n(S_{n})=\set{o}$,
and $\forall k\leqslant n$ the following holds:
\begin{itemize}
\item $\forall p\in S_k \setminus \bigcup_{\tr\in
    T_k}\trin,\ \rho_{k+1}(p)= \rho_k(p)$;
\item $\forall \tr\in T_k,\forall p,q\in \trin, 
  \rho_{k}(p)= \rho_{k}(q)$;
\item $\forall \tr\in T_k,\forall p\in \trin,\forall \tuple{x,q}\in \trout, \tuple{\rho_k(p),x, \rho_{k+1}(q)}\in E$
\end{itemize}
\begin{defi}[Simulation]\label{def:sim}
  A relation
  ${\duck}\subseteq \pset{P_1}\times\pset{P_2\rightharpoonup P_1}$
  between the states of $\A_1$ and the partial maps from the
  places of $\A_2$ to the places of $\A_1$ is called a
  \emph{simulation} between $\A_1$ and $\A_2$ if:
  \begin{itemize}
  \item if $S\duck E$ and $\eta\in E$ then the range of $\eta$ must
    be included in $S$;
  \item $\set{\iota_1}\duck\set{[\iota_2\mapsto \iota_1]}$;
  \item if $S\duck E$ and
    $\atrans[\tr][\A_1]S{S'}$, then $S'\duck E'$
    where $E'$ is the set of all $\eta'$ such that there is some
    $\eta\in E$ and a compatible set of transitions $T\subseteq\T_2$
    such that:
    \begin{itemize}
    \item $\atrans[T][\A_2]{\dom\eta}{\dom{\eta'}}$;
    \item $\forall \tr*\in T, \eta(\trin*)\subseteq \trin$ and
      ${\forall \tuple{x,q}\in \trout*,} {\tuple{x,\eta'(q)}\in \trout}$;
      %%% Présentation
    \item $\forall p\in\dom{\eta},\paren{\forall \tr*\in
        T, p\notin \trin*} \Rightarrow \eta(p)=\eta'(p)$.
    \end{itemize}
  \item if $S\duck E$ and $S=\emptyset$, then there must be some
    $\eta\in E$ such that $\dom\eta=\emptyset$. \qedhere
  \end{itemize}
\end{defi}

We will now prove that the language of $\A_1$ is contained in the
language of $\A_2$ if and only if there exists such a simulation.  We
first introduce the following notion of embedding.
\begin{defi}[Embedding]\label{def:embed}
  Let $\run=\tuple{S_0,\tr[0];\dots;\tr[n-1],S_n}$ be a run in $\A_1$,
  and let ${\mathcal R=\tuple{S'_0,T_0;\dots;T_{n-1},S'_n}}$ be a
  parallel run in $\A_2$.  An \emph{embedding} of $\mathcal R$ into
  $\run$ is a sequence $\paren{\eta_i}_{0\leqslant i\leqslant n}$ of
  maps such that for all $i<n$, we have:

  \noindent
  \begin{minipage}{.78\linewidth}
    \begin{itemize}
    \item $\eta_i$ is a map from $S'_i$ to $S_i$;
    \item the image of $T_i$ by $\eta_i$ is included in $\tr[i]$,
      meaning that for all $\tr\in T_i$, for all $p\in \trin$ and
      $\tuple{x,q}\in \trout$, $\eta_i(p)$ is contained in the input of
      $\tr[i]$ and $\tuple{x,\eta_{i+1}(q)}$ is in the output of $\tr[i]$;
    \item the image of the tokens in $S_i$ that do not appear in the
      input of $T_i$ are preserved ($\eta_i(p)=\eta_{i+1}(p)$) and
      their image is not in the input of $\tr[i]$.
    \end{itemize}
  \end{minipage}%
  \begin{minipage}{.22\linewidth}\centering
    \begin{tikzpicture}
      \node (0) at (0,0) {$S_i$}; \node[right=of 0] (1) {$S_{i+1}$};
      \node[below=of 0] (2) {$S'_i$}; \node[right=of 2] (3)
      {$S'_{i+1}$}; \edge (0) (1)[\tr[i]]; \edge[below] (2)
      (3)[T_i]; \edge[left] (2) (0)[\eta_i]; \edge[right] (3)
      (1)[\eta_{i+1}];
    \end{tikzpicture}
  \end{minipage}
\end{defi}

\noindent
Figure~\ref{fig:embed}~illustrates the embedding of some parallel run, 
% , producing
% $G\paren{((b\cdot c\cdot a\cdot b)\cap(b\cdot b\cdot c\cdot a))\cdot
% d}$
into the run presented in Figure~\ref{fig:run}. Notice that it is
necessary to have a parallel run instead of a simple one: to find
something that matches the second transition in the upper run, we need
to fire two transitions in parallel in the lower run.

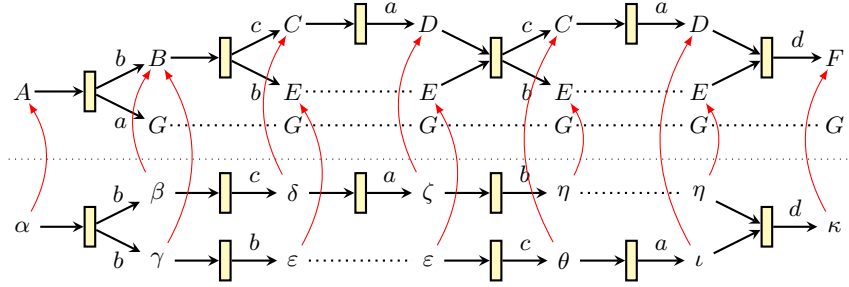
\begin{figure}[t]
  \centering 

  \begin{tikzpicture}[scale=0.9]
    \node[inner sep=1] (A) at (0,-0.5) {$A$};
    \node[inner sep=1] (B) at (2,0) {$B$};
    \node[inner sep=1] (G1) at (2,-1) {$G$};
    \node[inner sep=1] (C) at (4,0.5) {$C$};
    \node[inner sep=1] (E1) at (4,-0.5) {$E$};
    \node[inner sep=1] (G2) at (4,-1) {$G$};
    \node[inner sep=1] (D) at (6,0.5) {$D$};
    \node[inner sep=1] (E2) at (6,-0.5) {$E$};
    \node[inner sep=1] (G3) at (6,-1) {$G$};
    \node[inner sep=1] (C1) at (8,0.5) {$C$};
    \node[inner sep=1] (E3) at (8,-0.5) {$E$};
    \node[inner sep=1] (G4) at (8,-1) {$G$};
    \node[inner sep=1] (D1) at (10,0.5) {$D$};
    \node[inner sep=1] (E4) at (10,-0.5) {$E$};
    \node[inner sep=1] (G5) at (10,-1) {$G$};
    \node[inner sep=1] (F) at (12,0) {$F$};
    \node[inner sep=1] (G6) at (12,-1) {$G$};
    \transb (0) (1,-0.5);
    \transb (1) (3,0);
    \transb (2) (5,0.5);
    \transb (3) (7,0);
    \transb (4) (9,0.5);
    \transb (5) (11,0);
    \edge (A) (0);
    \edge[above] (0) (B)[b];
    \edge[below] (0) (G1)[a];

    \edge (B) (1);
    \draw[thick,dotted] (G1) to (G2);
    \edge[above] (1) (C)[c];
    \edge[below] (1) (E1)[b];

    \edge (C) (2);
    \draw[thick,dotted] (E1) to (E2);
    \draw[thick,dotted] (G2) to (G3);
    \edge[above] (2) (D)[a];

    \edge (D) (3);
    \edge (E2) (3);
    \draw[thick,dotted] (G3) to (G4);
    \edge[above] (3) (C1)[c];
    \edge[below] (3) (E3)[b];

    \edge (C1) (4);
    \draw[thick,dotted] (E3) to (E4);
    \draw[thick,dotted] (G4) to (G5);
    \edge[above] (4) (D1)[a];

    \edge (D1) (5);
    \edge (E4) (5);
    \draw[thick,dotted] (G5) to (G6);
    \edge[above] (5) (F)[d];
    
    \draw[dotted] (-0.2,-1.5) -- (12.2,-1.5);

    \node (a) at (0,-2.5) {$\alpha$};
    \node (b) at (2,-2) {$\beta$};
    \node (c) at (2,-3) {$\gamma$};
    \node (d) at (4,-2) {$\delta$};
    \node (e1) at (4,-3) {$\epsilon$};
    \node (f) at (6,-2) {$\zeta$};
    \node (e2) at (6,-3) {$\epsilon$};
    \node (g1) at (8,-2) {$\eta$};
    \node (h) at (8,-3) {$\theta$};
    \node (g2) at (10,-2) {$\eta$};
    \node (i) at (10,-3) {$\iota$};
    \node (j) at (12,-2.5) {$\kappa$};
    \trans (k) (1,-2.5);
    \trans (l) (3,-2);
    \trans (m) (3,-3);
    \trans (n) (5,-2);
    \trans (o) (7,-2);
    \trans (p) (7,-3);
    \trans (q) (9,-3);
    \trans (r) (11,-2.5);

    \draw[thick,dotted] (e1) to (e2);
    \draw[thick,dotted] (g1) to (g2);
    
    \edge (a) (k);
    \edge (b) (l);
    \edge (c) (m);
    \edge (d) (n);
    \edge (f) (o);
    \edge (e2) (p);
    \edge (h) (q);
    \edge (g2) (r);
    \edge (i) (r);

    \edge[above] (k) (b)[b];
    \edge[below] (k) (c)[b];
    \edge[above] (l) (d)[c];
    \edge[above] (m) (e1)[b];
    \edge[above] (n) (f)[a];
    \edge[above] (o) (g1)[b];
    \edge[above] (p) (h)[c];
    \edge[above] (q) (i)[a];
    \edge[above] (r) (j)[d];

    \draw[->,>=latex,color=red] (a) to[bend right] (A);
    \draw[->,>=latex,color=red] (b) to[bend left] (B);
    \draw[->,>=latex,color=red] (c) to[bend right] (B);
    \draw[->,>=latex,color=red] (d) to[bend left] (C);
    \draw[->,>=latex,color=red] (e1) to[bend right] (E1);
    \draw[->,>=latex,color=red] (f) to[bend left] (D);
    \draw[->,>=latex,color=red] (e2) to[bend right] (E2);
    \draw[->,>=latex,color=red] (g1) to[bend right] (E3);
    \draw[->,>=latex,color=red] (h) to[bend left] (C1);
    \draw[->,>=latex,color=red] (g2) to[bend right] (E4);
    \draw[->,>=latex,color=red] (i) to[bend left] (D1);
    \draw[->,>=latex,color=red] (j) to[bend left] (F);

  \end{tikzpicture}
  \caption{Embedding of a parallel run into the run
    from Figure~\ref{fig:run}.}
  \label{fig:embed}
\end{figure}

There is a close relationship between simulations and embeddings:
\begin{lem}\label{lem:comparing-automata:sim_vs_emb}
  Let $\A_1$ and $\A_2$ be two Petri automata, the following are
  equivalent:
  \begin{enumerate}[label=(\roman*)]
  \item there exists a simulation $\duck$ between $\A_1$ and $\A_2$;
  \item\label{item:1} for every accepting run $\run$ in $\A_1$, there is
    an accepting parallel run $\mathcal R$ in $\A_2$ that can be
    embedded into $\run$.\qedhere
  \end{enumerate}
\end{lem}
\begin{proof}
  If we have a simulation $\duck$, let
  $\run=\tuple{S_0,\tr[0];\dots;\tr[n],\emptyset}$ be an accepting run
  in $\A_1$. By the definition of simulation, we can find a sequence
  of sets of maps $\paren{E_k}_{0\leqslant k\leqslant n+1}$ such that
  $E_0=\set{[\iota_2\mapsto \iota_1]}$ and $\forall k, S_k\duck E_k$.
  Furthermore, we can extract from this a sequence of maps
  $\paren{\eta_k}_{0\leqslant k\leqslant n+1}$ and a sequence of
  parallel transitions $\paren{T_k}_{0\leqslant k\leqslant n}$ such
  that:
  \begin{itemize}
  \item $\forall k, \atrans[T_k][\A_2]{\dom{\eta_k}}{\dom{\eta_{k+1}}}$
  \item the parallel run
    $\tuple{\dom{\eta_0},T_0;\dots;T_n;\dom{\eta_{n+1}}}$ is accepting,
    and can be embedded into $\run$ using $\paren{\eta_k}$.
  \end{itemize}
  This follows directly from the definitions of embedding and
  simulation.

  On the other hand, if we have property~(\ref{item:1}), then we can
  define a relation $\duck$ by saying that $S\duck E$ if there is an
  accepting run $\run=\tuple{S_0,\tr[0];\dots;\tr[n],S_{n+1}}$ in
  $\A_1$ such that there is an index $k_0$: $S=S_{k_0}$; and the
  following holds: $\eta\in E$ if there is an accepting parallel run
  $\mathcal R$ in $\A_2$ visiting successively the states
  $S'_0,\dots,S'_{n+1}$ and
  $\paren{\eta'_k}_{0\leqslant k\leqslant n+1}$ an embedding of
  $\mathcal R$ into $\run$ such that $\eta=\eta'_{k_0}$.  It is then
  immediate to check that $\duck$ is indeed a simulation.
\end{proof}

\noindent
If $\eta$ is an embedding of $\mathcal R$ into $\run$, and $\rho$ is a
reading of $\Gr{\run}$ along $\run$, then we can easily check that
$\paren{\rho_i\circ\eta_i}_{0\leqslant i\leqslant n}$ is a parallel
reading of $\Gr{\run}$ along $\mathcal R$ in $\A_2$.  Thus, it is
clear that once we have such a run $\mathcal \run$ with the sequence
of maps $\eta$, we have that $\Gr{\run}$ is indeed in the language of
$\A_2$. The more difficult question is the completeness of this
approach: if $\Gr{\run}$ is recognised by $\A_2$, is it always the
case that we can find a run $\mathcal R$ that may be embedded into
$\run$?  The answer is affirmative, thanks to Lemma~\ref{lem:exord}
below. If $(\rho_j)_{0\leqslant j\leqslant n+1}$ is a reading of $G$
along $\run=\tuple{S_0,\tr[0];\dots;\tr[n],S_{n+1}}$, we write
$active(j)$ for the only position in
$\rho_{j}(\trin[j])$\footnote{Recall that if
  $(\rho_j)_{0\leqslant j\leqslant n+1}$ is a reading along $\run$
  then for all $p,q\in \trin[j]$, we have
  $\rho_{j}(p) = \rho_{j}(q)$.}. A binary relation $\sqsubseteq$ is a
\emph{topological ordering} on $G=\tuple{V,E,\iota,o}$ if
$\tuple{V,\sqsubseteq}$ is a linear order and $(p,x,q)\in E$ entails
$p\sqsubseteq q$.

\begin{lem}\label{lem:exord}
  Let $G\in \Ln{\A_2}$ and let $\sqsubseteq$ be a topological ordering
  on $G$. Then there exists a run $\run$ and a reading
  $(\rho_j)_{0\leqslant j\leqslant n+1}$ of $G$ along $\run$ such that
  $\forall k,active(k)\sqsubseteq active(k+1)$.
\end{lem}

\noindent
The proof of this result is achieved by taking a run $\run$
accepting $G$, and then exchanging transitions in $\run$ according to
$\sqsubseteq$, while preserving the existence of a reading. However we
need to introduce some lemmas first.

Let us fix $\A=\tuple{P,\T,\iota}$ a Petri automaton, and
$\run=\tuple{S_0,\tr[0];\dots;\tr[n],S_{n+1}}$ a run of $\A$.
\begin{defi}[Exchangeable transitions]
  Two transitions $\tr[k]$ and $\tr[k+1]$ are \emph{exchangeable} in
  $\run$ if for all $p\in S_k$, $p$ is in $\trin[k+1]$ implies that
  there is no $x\in X$ such that $(x,p)\in \trout[k]$.
\end{defi}
\noindent
As the name might suggest, two exchangeable transitions may be
exchanged in a run, and every graph read along the initial run can still
be read along the permuted run.
\begin{lem}\label{lem:exread}
  Suppose $\tr[k]$ and $\tr[k+1]$ are exchangeable for some
  $0\leqslant k<n$. We write
  $C'=S_{k}\setminus \trin[{k+1}]\cup\settrout[{k+1}]$. Then
  $\atrans[\tr[k+1]]{S_{k}} C'$ and $\atrans[\tr[k]]{C'}{S_{k+2}}$.
  Furthermore, for every graph $G$, if $G\in \Ln\run$, then
  $G\in \Ln{\run[k\leftrightarrow k+1]}$, where:
  \[
    \run[k\leftrightarrow k+1]\eqdef\tuple{S_0,
      \tr[0];\cdots;\tr[k+1];\tr[k];\cdots;\tr[n],S_n}.\]
\end{lem}
\begin{proof}
  The fact that $S_{k}\xrightarrow{\tr[k+1]} _\A C'$ and
  $C'\xrightarrow{\tr[k]} _\A S_{k+2}$ is trivial to check, with the
  definition of exchangeable.

  Let $\paren{\rho_j}_{0\leqslant j\leqslant n+1}$ be a reading of $G$
  along $\run$. If $\paren{\rho'_j}_{0\leqslant j\leqslant n+1}$ is
  defined by:
  \begin{align*}
    \rho'_j(p)&=
    \begin{cases}
      \rho_j(p)&\text{if }j\neq k+1,\\
      \rho_{k+2}(p)&\text{if }j=k+1\text{ and }(x,p)\in \trout[{k+1}]
      \text{ for some }x,\\
      \rho_{k}(p)&\text{otherwise}.
    \end{cases}
  \end{align*}
  Then $\paren{\rho'_j}_{0\leqslant j\leqslant n+1}$ is a reading of $G$
  along $\run[k\leftrightarrow k+1]$.
\end{proof}

\noindent
Recall that if $(\rho_j)_{0\leqslant j\leqslant n}$ is a reading of
$G$ along $\xi$ we write $active(j)$ for the only position in
$\rho_{j}(\trin[j])$.

\begin{lem}\label{sec:omitted-proof:-lemmaexord}
  Let $\sqsubseteq$ be a topological ordering %~(a linear ordering
  % compatible with the directed edges of the graph)
  on $G$.  If $(\rho_j)_{0\leqslant j\leqslant n+1}$ is a reading of
  $G$ along $\run$, and if ${active(k+1)\sqsubseteq active(k)}$ for
  some $k$, then $\tr[k]$ and $\tr[k+1]$ are exchangeable.
\end{lem}
\begin{proof} 
  Let $G=\tuple{V,E,\iota,o}$.  since
  $(\rho_i)_{0\leqslant j\leqslant n+1}$ is a reading, for every
  $\tuple{x,p}\in \trout[k]$, we have
  $\tuple{active(k),x,\rho_{k+1}(p)}\in E$, and thus
  \[active(k)\sqsubset\rho_{k+1}(p).\]
  We know that $active(k+1)\sqsubseteq active(k)$, meaning by
  transitivity that $active(k+1)\sqsubset\rho_{k+1}(p)$.
  Hence \[active(k+1)\neq \rho_{k+1}(p)\]
  and because $(\rho_i)_{0\leqslant j\leqslant n+1}$ is a reading we
  can infer that $p\notin \trin[{k+1}]$, thus proving that $\tr[k]$
  and $\tr[{k+1}]$ are exchangeable.
\end{proof}

\begin{proof}[Proof of Lemma~\ref{lem:exord}]
  Because $G$ is in $\Ln\A$, we can find a reading $\rho'$ along some
  run $\run$. If that reading is not in the correct order, then by
  Section~\ref{sec:omitted-proof:-lemmaexord} we can exchange two transitions
  and Lemma~\ref{lem:exread} ensures that we can find a corresponding
  reading. We repeat this process until we get a reading in the
  correct order.
\end{proof}
\noindent
(Notice that if $G$ contains cycles, this lemma cannot apply because 
of the lack of a topological ordering.)

Lemma~\ref{lem:exord} enables us to build an embedding from every
reading of $\Gr{\run}$ in $\A_2$.
\begin{lem}\label{lem:recplong} 
  Let $\run$ be an accepting run of $\A_1$. Then
  $\Gr{\run}\in\Ln{\A_2}$ if and only if there is an accepting
  parallel run in $\A_2$ that can be embedded into $\run$.
\end{lem} 
\begin{proof}
  Let $\run=\tuple{S_0,\tr[0];\dots;\tr[n],S_{n+1}}$ be a run.
  For all indexes $k$ and places $p\in S_k$, define
  $\rho^1_{k}(p)\eqdef\Nu(k,p)=\min\setcompr{l}{l\geqslant k\text{ and
    }p\in\trin[l]}$. We have that
  $\paren{\rho^1_k}_{0\leqslant k\leqslant n+1}$ is a reading of
  $\Gr\run$ along $\run$.

  \begin{itemize}
  \item Assume an embedding $\paren{\eta_k}_{0\leqslant k\leqslant n+1}$
    of an accepting parallel run $\mathcal R$ into $\run$. We define a
    parallel reading $(\rho^2_k)$ of $\Gr\run$ in $\A_2$ by letting
    $\rho^2_k(p)\eqdef \rho^1_k(\eta_k(p))$.
  \item On the other hand, notice that the natural ordering on $\Nat$
    is a topological ordering on $\Gr{\run}$, and that
    $\forall 0\leqslant k\leqslant n$, $\rho^1_{k}(\trin[k])=\set k$.
    By Lemma~\ref{lem:exord} we gather that $\Gr{\run}$ is in $\Ln{\A_2}$
    if and only if there exists a reading
    $(\rho^2_j)_{0\leqslant j\leqslant n'}$ of $\Gr{\run}$ along some
    run $R'=\tuple{S'_0,\tr*[0];\dots;\tr*[n'],S'_{n'+1}}$ such
    that $\forall j,active(j)\leqslant active(j+1)$ (with $active(j)$
    the only position in $\rho^2_{j}\paren{\trin[j]}$).

    Now, suppose we have such a reading; we can build an embedding
    $(\eta_k)_{0\leqslant k\leqslant n+1}$ as follows. For
    $k\leqslant n$, define
    $T_k\eqdef\setcompr{\tr*[j]}{active(j)=k}.$
    We describe the construction incrementally:
    \begin{itemize}
    \item $\eta_0=[\iota_2\mapsto \iota_1]$.
    \item For all
      $p\in \dom{\eta_{k}}\setminus\bigcup_{\tr\in T_k} \trin$ we
      simply set $\eta_k(p)=\eta_{k-1}(p)$.
    \item Otherwise, $\forall \tr*[j]\in T_k$, let
      $q\in \trin*[j]$.  Then, for all $\tuple{x,p}$ in $\trout*[j]$,
      because $\rho^2$ is a reading and by construction of $\Gr{\run}$
      we also know that there is some $p'\in S_{k+1}$ that satisfies
      $\tuple{x,p'}\in \trout[k]$ and $\rho^1_{k+1}(p')=\rho^2_{j+1}(p)$.
      That $p'$ is a good choice for $p$, hence we define
      $\eta_k(p)=p'$.
    \end{itemize}
    It is then routine to check that~$\paren{\eta_k}_{0\leqslant
      k\leqslant n+1}$ is indeed an embedding.
  \end{itemize}
  \noindent
  For the if direction, we build a parallel reading from the
  embedding, as explained above. For the other direction, we consider
  a reading of $\Gr\run$ in $\A_2$ along some run $\run'$. Notice that
  the natural ordering on $\Nat$ is a topological ordering on
  $\Gr{\run}$; we may thus change the order of the transitions in
  $\run'$ (using Lemma~\ref{lem:exord}) and group them adequately to obtain
  a parallel reading $\mathcal R$ that embeds in $\run$.
\end{proof}
\noindent
So we know that the existence of embeddings is equivalent to the
inclusion of languages, and we previously established that it is also
equivalent to the existence of a simulation relation. Hence, the
following characterisation holds:
\begin{prop}\label{thm:inclsim}
  Let $\A,\B\in\PA[X]$ be two Petri automata. We have
  $\Gr{\A}\subseteq \Ln{\B}$ if and only if there exists a simulation
  relation $\duck$ between $\A$ and $\B$.
\end{prop}
\begin{proof}
  By Lemmas~\ref{lem:comparing-automata:sim_vs_emb}
  and~\ref{lem:recplong}.
\end{proof}
\noindent
As Petri automata are finite, there are finitely many relations in
$\pset{\pset{P_1}\times\pset{P_2\rightharpoonup P_1}}$. The existence
of a simulation thus is decidable, allowing us to prove the main
result:
\begin{thm}\label{thm:sim:pre} Given two Petri automata
  $\A,\B\in\PA[X]$, testing whether $\Gr\A\subseteq \Ln\B$ is
  decidable.
\end{thm}
\noindent
In practice, we can build the simulation on-the-fly, starting from the
pair $\tuple{\set{\iota_1},\set{\left[\iota_2\mapsto\iota_1\right]}}$
and progressing from there. We have implemented this algorithm in
\textsc{OCaml}~\cite{rklm:web}. Even though its theoretical worst case
time complexity is huge\footnote{A quick analysis gives a
  $\O{2^{n+\paren{n+1}^{m}}}$ complexity bound, where $n$ and $m$ are
  the numbers of places of the automata.}, we get a result almost
instantaneously on simple one-line examples.

By Proposition~\ref{prop:reduce}, the equational theory of
identity-free Kleene lattices is thus decidable, and so is language
inclusion of Petri automata.

\subsection{The problems with converse, \texorpdfstring{$1$}{unit},
  and \texorpdfstring{$\top$}{top}}
\label{ssec:problems:cycles}

The previous algorithm is not complete in presence of converse, $1$,
or $\top$. More precisely, Lemma~\ref{lem:recplong} does not hold for
general automata. Indeed, it is not possible to compare two runs just
by relating the tokens at each step, and checking each transition
independently. Consider the automaton from Figure~\ref{fig:1etab+c}.
\begin{figure}[t]
 \centering
 \begin{tikzpicture}[scale=0.8]
   \state[A] (a) (0,0);
   \state[B] (b) (2,0.5);
   \state[C] (c) (4,0.5);
   \state[D] (d) (6.2,0.5);
   \state[E] (e) (6.2,-0.5);
   \trans (0) (1,0);
   \trans (1) (3,0.5);
   \trans (2) (4,1.3);
   \trans (3) (5,0.5);
   \transf (4) ($(e)!.5!(d)+(2,0)$);
   \initst (a);
   \draw[arc] (a) to (0);
   \draw[arc] (b) to (1);
   \draw[arc] (c) to[out=130,in=-130] (2.west);
   \draw[arc] (c) to (3);
   \edge[above,bend left] (0)(b)[a];
   \edge[below,out=-20,in=180] (0)(e)[1];
   \edge[above] (1)(c)[b];
   \edge[right,out=-50,in=50] (2.east)(c)[b];
   \edge[above] (3)(d)[c];
   \edge[out=0,in=-140] (e) (4);
   \edge[out=0,in=140] (d) (4);
 \end{tikzpicture}
 \caption{A Petri automaton for $1\cap a\cdot b^+\cdot c$.}
 \label{fig:1etab+c}
\end{figure}
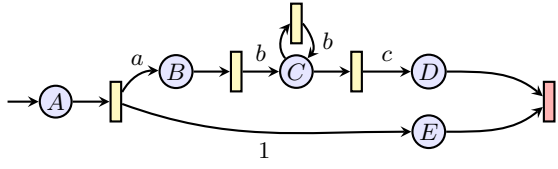
This automaton has in particular an accepting run recognising $1\cap
abc$. 
% It ends in the final state $\set{D,E}$.
Let us try to test if this is smaller than the following runs from
another automaton (we represent the transitions simply as arrows,
because they only have a single input and a single output):
\begin{center}
  \begin{tikzpicture}[scale=.6]
    \node (x0) at (0,0) {$x_0$};
    \node (x1) at (2,0) {$x_1$};
    \node (x2) at (4,0) {$x_2$};
    \node (x3) at (6,0) {$x_3$};
    \node (x4) at (8,0) {$x_4$};
    \node (x5) at (10,0) {$x_5$};
    \node (x6) at (12,0) {$x_6$};
    \edge[above] (x0)(x1)[a];
    \edge[above] (x1)(x2)[b];
    \edge[above] (x2)(x3)[c];
    \edge[above] (x3)(x4)[a];
    \edge[above] (x4)(x5)[b];
    \edge[above] (x5)(x6)[c];
    \node (y0) at (0,-1) {$y_0$};
    \node (y1) at (2,-1) {$y_1$};
    \node (y2) at (4,-1) {$y_2$};
    \node (y3) at (6,-1) {$y_3$};
    \node (y4) at (8,-1) {$y_4$};
    \node (y5) at (10,-1) {$y_5$};
    \node (y6) at (12,-1) {$y_6$};
    \node (y7) at (14,-1) {$y_7$};
    \edge[above] (y0)(y1)[a];
    \edge[above] (y1)(y2)[b];
    \edge[above] (y2)(y3)[c];
    \edge[above] (y3)(y4)[a];
    \edge[above] (y4)(y5)[b];
    \edge[above] (y5)(y6)[b];
    \edge[above] (y6)(y7)[c];
  \end{tikzpicture}
\end{center}
It stands to reason that we would reach a point where:
\begin{itemize}
\item for the first run: $\set{D,E}\duck \set{[x_3\mapsto D]}$;
\item for the second run: $\set{D,E}\duck \set{[y_3\mapsto D]}$.
\end{itemize}
So if it were possible to relate the end of the runs just with this
information, they should both be bigger than $1\cap abc$ or both
smaller or incomparable. But in fact the first run (recognising
$abcabc$) is bigger than $1\cap abc$ but the second (recognising
$abcabbc$) is not. This highlights the need for having some memory
of previously fired transition when trying to compare runs of
general Petri automata, thus preventing our local approach to bear
fruits. The same kind of example can be found with the converse
operation or $\top$ instead of $1$.

\section{Complexity}
\label{sec:complexity}

The notion of simulation from the previous section actually allows us
to decide language inclusion of Petri automata in \expspace. We will
eventually show that this problem is \expspace-complete.

\begin{prop}
  \label{prop:complexity-1}
  Comparing Petri automata is in \expspace.
\end{prop}
\begin{proof}
  Our measure for the size of an automaton here is its number of places
  (the number of transitions is at most exponential in this
  number). Here is a non-deterministic semi-algorithm that tries to
  refute the existence of a simulation relation between $\A_1$ and
  $\A_2$.
  \begin{enumerate}[label=\arabic*:]
  \item start with $S\eqdef\set{\iota_1}$ and
    $E\eqdef\set{[\iota_2\mapsto\iota_1]}$;
  \item\label{step2} if $S=\emptyset$, check if there is some
    $\eta\in E$ such that $\dom\eta=\emptyset$, if not return
    $\mathtt{FALSE}$;
  \item choose non-deterministically a transition $\tr\in\T_1$ such
    that $\trin\subseteq S$;
  \item fire $\tr$, which means that
    $S\eqdef S\setminus \trin \cup\settrout$;
  \item have $E$ progress along $\tr$ as well, according to the
    conditions from Definition~\ref{def:sim}.
  \item go to step \ref{step2}.
  \end{enumerate}
  All these computations can be done in exponential space. In
  particular as $S$ is a set of places in $P_1$, it can be stored in
  space $|P_1|\times\log(|P_1|)$. Similarly, $E$, being a set of
  partial functions from $P_2$ to $P_1$, each of which of size
  $|P_2|\times \log(|P_1|+1)$, can be stored in space
  $|P_1+1|^{|P_2|}\times|P_2|\times\log(|P_1|+1)$.  This
  non-deterministic \expspace\ semi-algorithm can then be turned into
  an \expspace\ algorithm by Savitch' theorem~\cite{savitch70}.
\end{proof}

\begin{prop}
  \label{prop:exp:hard}
  The (in)equational theory of representable identity-free Kleene
  lattices is \expspace-hard.
\end{prop}
\begin{proof}
  We perform a reduction from the problem of universality of regular
  expressions with intersection. 
  Fürer showed that given a regular expression with intersection $e$
  over the alphabet $\Sigma$, checking whether the language denoted by
  $e$ is equal to $\Sigma^\star$ requires exponential
  space~\cite{Furer1980}. In fact, by slightly modifying Fürer's
  proof, one can strengthen the result and use expressions
  without~$1$, \ie, over the syntax of $\SExp[\Sigma]$. We write
  $L(e)$ for the language (set of strings) denoted by such an
  expression.

  The reduction relies on the fact that for every expression $e$,
  $\tms{\Sigma^+}\subseteq\cltm{\tms e}$ if and only if
  $L\paren{\Sigma^+}\subseteq L\paren e$. First notice that that
  $\Sigma^+=L\paren{\Sigma^+}=\tms{\Sigma^+}$. Then, we can show by
  induction on $e$ that for every word $w\in\Sigma^+$ we have
  $w\in L\paren e$ if and only if $w\in\cltm{\tms e}$.

  This means that to answer the question of universality of the
  regular expression with intersection $e$, we can ask instead whether
  $\tms{\Sigma^+}\subseteq\cltm{\tms e}$ which is equivalent by
  Theorem~\ref{thm:interlang} to $\Rel\models \Sigma^+ \leq e$.
\end{proof}

\noindent
Consider the automaton $\A(e)$ we associate with an expression $e$.  The
number of places in $\A(e)$ is linear in the size of the simple
expression $e$. (And the exponential upper-bound on the number of
transitions is asymptotically reached, consider for instance the
automaton for
${(a_1\sumtm b_1)\cap(a_2\sumtm b_2)\cap\dots\cap(a_n\sumtm b_n)}$.)

Therefore, Propositions~\ref{prop:complexity-1}
and~\ref{prop:exp:hard} finally lead to
\begin{thm}
  \label{thm:cplx}
  The problem of comparing Petri automata, as well as the
  (in)equational theory of representable identity-free Kleene lattices
  are \expspace-complete.
\end{thm}

\noindent Table~\ref{tab:cplx} summarises the various complexity
results. The first two lines correspond to Theorem~\ref{thm:cplx}.
The upper-bound in the third line is Theorem~\ref{thm:jatmey} from the
following section, it is a corollary from a result by Jategaonkar and
Meyer~\cite{Jategaonkar96}; the lower-bound in this line follows from
the Kleene algebra fragment, which is known to be
\pspace-complete~\cite{MS73}. The fourth and fifth lines are
equivalent by Theorem~\ref{thm:interlang}; the lower bound follows
from Proposition~\ref{prop:exp:hard}, for the upper-bound it suffices
to enumerate all graphs: one can easily decide whether a given graph
is accepted by a Petri automaton. For the last line, the problem also
contains Kleene algebra, whence the lower bound; it is unclear however
whether this last problem is co-recursively enumerable: given a graph
$G$ and a Petri automaton, we need to decide whether $G=\retype{H}$
for some series-parallel graph $H$ produced by the automaton. A bound
on the size of such a graph $H$ w.r.t. $G$ would suffice; this seems
plausible but we do not have a proof.
\begin{table}[t]
  \centering
  \begin{tabular}{c|c|c|c}
    Inputs&Problem&Lower bound&Upper bound\\\hline
    $e,f\in\SExp[X]$ & $\Rel\models e=f$ & \expspace & \expspace\\
    \multirow{2}{*}{$\A,\B\in\PA[X]$} & $\Ln\A=\Ln\B$ & \expspace & \expspace\\
          & $\Gr\A=\Gr\B$ & \pspace & \expspace\\\hline
    $e,f\in\Exp[X]$ & $\Rel\models e=f$ & \expspace & co-r.e.\\
    \multirow{2}{*}{$\A,\B\in\PA[\Xb]$} & $\Ln\A=\Ln\B$ & \expspace & co-r.e.\\
          & $\retype{\Gr\A}=\retype{\Gr\B}$ & \pspace & -\\
  \end{tabular}
  \caption{Summary of the complexity results}
  \label{tab:cplx}
\end{table}

The table is presented with equations and equalities; the same table
holds for inequations and inclusions. Also remember that
$\Ln\A\subseteq\Ln\B$ is equivalent to $\Gr\A\subseteq\Ln\B$, but not
to $\Gr\A\subseteq\Gr\B$.

\section{Relationship with standard Petri net notions}
\label{sec:relat-with-stand}

Our notion of Petri automaton is really close to the standard notion
of labelled (safe) Petri net, where the transitions themselves are
labelled, rather than their outputs. We motivate this design choice,
and we relate some of the notions we introduced to the standard
ones~\cite{Murata}. 

A Petri automaton can be translated into a safe Petri net whose
transitions are labelled by $X\uplus \set\tau$, the additional label
$\tau$ standing for silent actions.  For each automaton transition
$\tuple{\set{p_1,\dots,p_n},\set{\tuple{x_1,q_1},\dots,\tuple{x_m,q_m}}}$
with $m>1$, we introduce $m$ fresh places $r_1,\dots,r_m$ and $m+1$
transitions:
\begin{itemize}
\item a silent transition $t_0$ with preset $\set{p_1,\dots,p_n}$ and
  postset $\set{r_1,\dots,r_m}$;
\item and for each $1\leqslant k\leqslant m$ a transition $t_k$
  labelled by $x_k$, with preset $\set{r_k}$ and postset $\set{q_k}$.
\end{itemize}

The inductive construction from Section~\ref{sec:exp:pa} is actually
simpler to write using such labelled Petri nets, as one can freely use
$\tau$-labelled transitions to assemble automata into larger ones, one
does not need to perform the $\tau$-elimination steps on the fly.

On the other hand, we could not define an appropriate notion of
simulation for Petri nets: we need to fire several transitions at once
in the small net, to provide enough information for the larger net to
answer; delimiting which transitions to group and which to separate is
non-trivial; similarly, defining a notion of parallel step is delicate
in presence of $\tau$-transitions.  By switching to our notion of
Petri automata, we impose strong constraints about how those
$\tau$-transitions should be used, resulting in a more fitted model.

To describe a run in a Petri net $N$, one may use a \emph{process}
$p\colon K\rightarrow N$, where $K$ is an \emph{occurrence net} (a
partially ordered Petri net)~\cite{goltz1983}. The graphical
representation~(Figures~\ref{fig:run} and~\ref{fig:embed}) we used to
describe runs in an automaton are in fact a mere adaptation of this
notion to our setting (with labels on arcs rather than on
transitions).

Our notion $\Gr\run$ of trace of a run actually corresponds to the
standard notion of \emph{pomset-trace}\cite{Jategaonkar96}, via
dualisation. Let $\run$ be a run in a Petri automaton, and let $\run'$
be the corresponding run in the corresponding labelled Petri net. Let
$\Gr\run=\tuple{V,E,\iota,o}$. It is not difficult to check that
the pomset-trace of $\run'$ is isomorphic to $\tuple{E,<_E}$, where
$<_E$ is the transitive closure of the relation $<$ defined on $E$ by
$\forall x,y,z\in V,~a,b\in X,\tuple{x,a,y}<\tuple{y,b,z}$.

The correspondence is even stronger: two graphs produced by accepting
runs in (possibly different) automata
$\Gr{\run_1}=\tuple{V_1,E_1,\iota_1,o_1}$ and
$\Gr{\run_2}=\tuple{V_2,E_2,\iota_2,o_2}$ are isomorphic if and only
if their pomset-traces $(E_1,<_{E_1})$ and $(E_2,<_{E_2})$ are
isomorphic. The proof of this relies on the fact that accepting runs
produce graphs satisfying the following properties:
\begin{align*}
 \forall \tuple{x,a,y}\in E,&~ x\neq\iota\Leftrightarrow 
 \paren{\exists \tuple{z,b}\in V\times X:\tuple{z,b,x}\in E},\\
 \forall \tuple{x,a,y}\in E,&~y\neq o\Leftrightarrow 
 \paren{\exists \tuple{z,b}\in V\times X:\tuple{y,b,z}\in E}.
\end{align*}
Hence, pomset-trace language equivalence corresponds exactly to
equivalence of the sets of produced graphs in our setting (up to
isomorphism).

Jategaonkar and Meyer showed that the pomset-trace equivalence problem
for safe Petri nets is \expspace-complete~\cite{Jategaonkar96}. As a
consequence, comparing the sets of series-parallel graphs produced by
Petri automata is in \expspace:
\begin{thm}
  \label{thm:jatmey}
  Deciding whether $\Gr\A=\Gr\B$ for two Petri automata
  $\A,\B\in\PA[X]$ is in \expspace.  
  % \todo{complete? -- no se}
\end{thm}
As explained before the corresponding equivalence does not coincide
with the one discussed in the present paper, where we compare the sets
of graphs modulo homomorphism (see, e.g.,
Proposition~\ref{prop:reduce}).  Also note that the above result is
for series-parallel graphs only: it does not explain how to compare
the sets of graphs $\retype{\Gr\A}$ and $\retype{\Gr\B}$ when given
two Petri automata $\A,\B\in\PA[\Xb]$ labelled in $\Xb$. This latter
question remains open.

\section{Branching Automata}
\label{sec:relat-with-branch}

There are several notions of regular/recognisable/definable sets of
graphs in the literature, and several Kleene-like theorems. See for
instance the whole line of research that followed from Courcelle's
work~\cite{courcelle:hal-00646514}, where monadic second order logic
plays a central role. Courcelle's conjecture that recognisabiliy
coincides with definability on classes of graphs of bounded treewidth
has been proved recently~\cite{BojanczykP16}---whether it also
coincides with the notions of recognisability and regularity we use
here is unclear. In particular, since we mainly restrict to
series-parallel graphs, our syntax for expressions greatly departs
from Courcelle's one, and we can work with a simpler automata
model. The work of Bossut et al.~\cite{bossut1995kleene} on planar
directed acyclic graphs is also worth mentioning, but again the
syntax and automata they use are really different from ours: their
graphs are hypergraphs and they are not necessarily series-parallel.

Much closer to the present work is that of Lodaya and
Weil~\cite{Lodaya1998,LODAYA2000347,LODAYA2001269}, who introduced
another kind of automata to recognise series-parallel graphs, called
``branching automata''. They obtained a Kleene Theorem for this model,
using the same notion of regularity as the one we use here
(Definition~\ref{def:reg:sp}). In this section we recall the
definition of branching automata and describe precisely the
relationship between their result and our own.

\subsection{Definitions and Kleene Theorem}
\label{sec:definitions-2}

\begin{defi}[Branching automaton]
  A \emph{branching automaton} over the alphabet $X$ is a tuple
  $\tuple{Q,T_{seq},T_{fork},T_{join},I,F}$, where $Q$ is a set of
  states, $I$ and $F$ are subsets of $Q$, respectively the input and
  output states, and the transitions are split in three sets:
  \begin{itemize}
  \item $T_{seq}\subseteq Q\times X\times Q$ is the set of sequential transitions;
  \item $T_{fork}\subseteq Q\times\Mns Q$ is the set of opening transitions;
  \item $T_{seq}\subseteq \Mns Q\times Q$ is the set of closing
    transitions.
  \end{itemize}
  (Here $\Mns Q$ represents the set of multisets over $Q$ with
  cardinality at least 2.)
\end{defi}

\begin{figure}[t]
  \centering
  \begin{tikzpicture}[xscale=1.5]
    \etat (0) (0,0.5);
    \etat (1) (1,1);
    \etat (2) (1,0);
    \etat (3) (6,0);
    \etat (4) (2,1.5);
    \etat (5) (2,.5);
    \etat (6) (3,1.5);
    \etat (7) (4,.5);
    \etat (8) (4,1.5);
    \etat (9) (5,1);
    \etat (10) (6,1);
    \etat (11) (7,0.5);
    \initst (0);\fnst(11);
    \draw[arc,bend left] (0) to node[inner sep=0,pos=.4](x1) {} (1);
    \draw[arc,bend right] (0) to node[inner sep=0,pos=.4](x2) {} (2);
    \draw[bend left] (x1) to (x2);
    \draw[arc,bend left] (1) to node[inner sep=0,pos=.4](y1) {} (4);
    \draw[arc,bend right] (1) to node[inner sep=0,pos=.4](y2) {} (5);
    \draw[bend left] (y1) to (y2);
    \edge[below](2)(3)[a];
    \edge(4)(6)[c];
    \edge[below](5)(7)[b];
    \edge(6)(8)[a];
    \draw[arc,bend left] (8) to node[inner sep=0,pos=.6](z1) {} (9);
    \draw[arc,bend right] (7) to node[inner sep=0,pos=.6](z2) {} (9);
    \draw[bend right] (z1) to (z2);
    \edge(9)(10)[d];
    \draw[arc,bend left] (10) to node[inner sep=0,pos=.6](t1) {} (11);
    \draw[arc,bend right] (3) to node[inner sep=0,pos=.6](t2) {} (11);
    \draw[bend right] (t1) to (t2);
  \end{tikzpicture}
  \caption{Example of a branching automaton}
  \label{fig:branch-aut}
\end{figure}
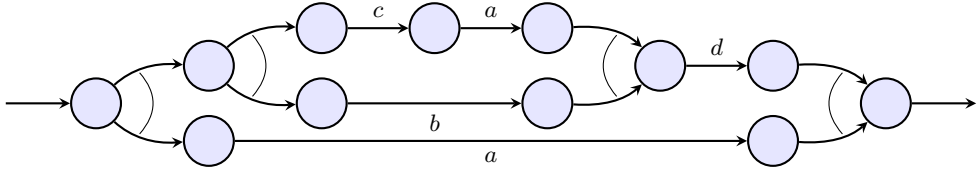
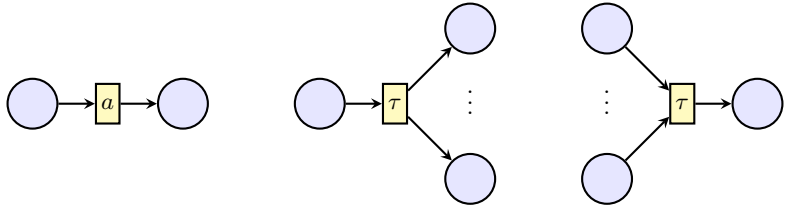
\begin{figure}[t]
  \centering
  \begin{tikzpicture}[baseline=(0.south),xscale=2]
    \etat (0) (0,0);\etat (1)(1,0);\trans[a](a)(.5,0);
    \edge (0)(a);\edge(a)(1);
  \end{tikzpicture}\hspace{1cm}
  \begin{tikzpicture}[baseline=(0.south),xscale=2]
    \etat (0) (0,0);\etat (1)(1,1);\etat (2)(1,-1);
    \node () at ($(1.south)!.5!(2.north)$) {$\rvdots$};
    \trans[\tau](a)(.5,0);
    \edge (0)(a);\edge(a)(1);\edge(a)(2);
  \end{tikzpicture}\hspace{1cm}
  \begin{tikzpicture}[baseline=(0.south),xscale=2]
    \etat (0) (1,0);\etat (1)(0,1);\etat (2)(0,-1);
    \node () at ($(1)!.5!(2)$) {$\rvdots$};\trans[\tau](a)(.5,0);
    \edge (a)(0);\edge(1)(a);\edge(2)(a);
  \end{tikzpicture}
  \caption{Prescribed transitions of branching automata}
  \label{fig:branch-trans}
\end{figure}
An example of such an automaton is given in Figure~\ref{fig:branch-aut}.
These automata can be seen as labelled Petri nets of a particular
shape: transitions are restricted to the three types described
on Figure~\ref{fig:branch-trans}.

Lodaya and Weil formally use these automata on simple terms quotiented
by associativity of $\cdot$ and $\cap$, and commutativity
of~$\cap$. This is equivalent to working with series-parallel graphs
modulo isomorphism.

\begin{defi}[Runs and language of a branching automaton]
  Let $t$ be a simple term, there is a \emph{run} on $t$ from state
  $p$ to state $q$ if:
  \begin{itemize}
  \item $t=a\in X$ and $\tuple{p,a,q}\in T_{seq}$;
  \item $t=t_1\cap\dots\cap t_n$ with $n\geqslant 2$, there are two transitions
    $\tuple{p,\left[p_1,\dots, p_n\right]}\in T_{fork}$ and
    $\tuple{\left[q_1,\dots, q_n\right],q}\in T_{join}$, and for every
    $1\leqslant i\leqslant n$ there is a run on $t_i$ from $p_i$ to $q_i$;
  \item $t=t_1\cdot \dots \cdot t_n$ with $n\geqslant 2$, there are states
    $p=p_0,p_1,\dots,p_n=q$, and for every $0\leqslant i< n$
    there is a run on $t_i$ from $p_i$ to $p_{i+1}$.
  \end{itemize}

  The \emph{language} of a branching automaton $\B$, written $\Gr\B$
  is defined as the set of series-parallel graphs $\Gr t$ such that
  there exists a pair of states $\tuple{q_i,q_f}\in I\times F$, and a
  run on $t$ from $q_i$ to $q_f$, .
\end{defi}

Lodaya and Weil impose restrictions on the runs of the automaton, that
correspond to a safety constraint over the underlying Petri net, much
like Constraint~\ref{cstr:safety}. Here we only consider branching
automata implicitly satisfying those constraints.

\begin{thm}[Kleene Theorem for branching
  automata~\cite{LODAYA2000347}]\label{thm:ktba}
  For every set $S$ of series-parallel graphs the following are equivalent:
  \begin{enumerate}[label=(\roman*)]
  \item there is a simple expression $e\in\SExp[X]$ such that
    $S=\Gr e$;
  \item there is a branching automaton $\B$ such that
    $S=\Gr \B$.
  \end{enumerate}
\end{thm}

\subsection{Comparison with Petri automata}
\label{sec:comp-with-petri}

At first glance the two Kleene theorems and the fact that both
branching automata and Petri automata are Petri net-based seem to mean
they are completely equivalent. Indeed the same set of regular-like
expressions may be used to describe their semantics.  However they
still exhibit some deep differences.

\begin{figure}[t]
  \centering
  \begin{tikzpicture}[scale=1.5]
    \etat[0](0)(0,0);
    \etat[1](1)(1,.5);
    \etat[2](2)(1,-.5);
    \etat[3](3)(2,.75);
    \etat[4](4)(2,.25);
    \etat[5](5)(2,-.25);
    \etat[6](6)(2,-.75);
    \etat[13](13)(5,0);
    \etat[12](12)(4,.5);
    \etat[11](11)(4,-.5);
    \etat[10](10)(3,.75);
    \etat[9](9)(3,.25);
    \etat[8](8)(3,-.25);
    \etat[7](7)(3,-.75);
    \initst(0);\fnst(13);
    \draw[arc,bend left] (0) to node[inner sep=0,pos=.4](x1) {} (1);
    \draw[arc,bend right] (0) to node[inner sep=0,pos=.4](x2) {} (2);
    \draw[bend left] (x1) to (x2);
    \draw[arc,in=180,out=40] (1) to node[inner sep=0,pos=.4](x1) {} (3);
    \draw[arc,in=180,out=-40] (1) to node[inner sep=0,pos=.4](x2) {} (4);
    \draw[bend left] (x1) to (x2);
    \draw[arc,in=180,out=40] (2) to node[inner sep=0,pos=.4](x1) {} (5);
    \draw[arc,in=180,out=-40] (2) to node[inner sep=0,pos=.4](x2) {} (6);
    \draw[bend left] (x1) to (x2);
    \edge (3)(10)[a];\edge(4)(9)[b];\edge(5)(8)[c];\edge(6)(7)[d];
    \draw[arc,out=0,in=140] (8) to node[inner sep=0,pos=.6](t1) {} (11);
    \draw[arc,out=0,in=-140] (7) to node[inner sep=0,pos=.6](t2) {} (11);
    \draw[bend right] (t1) to (t2);
    \draw[arc,out=0,in=140] (10) to node[inner sep=0,pos=.6](t1) {} (12);
    \draw[arc,out=0,in=-140] (9) to node[inner sep=0,pos=.6](t2) {} (12);
    \draw[bend right] (t1) to (t2);
    \draw[arc,bend left] (12) to node[inner sep=0,pos=.6](t1) {} (13);
    \draw[arc,bend right] (11) to node[inner sep=0,pos=.6](t2) {} (13);
    \draw[bend right] (t1) to (t2);
  \end{tikzpicture}
  \caption{Example of branching automaton}
  \label{fig:exbaglob}
\end{figure}
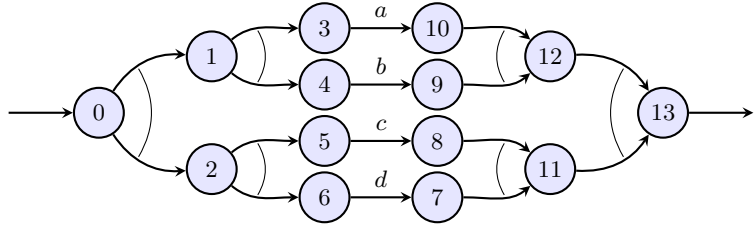

The first difference comes from the runs in the two models. In some
sense, the runs in a branching automaton require a ``global'' view of
the term being read. Consider the branching automaton in
Figure~\ref{fig:exbaglob}, and the term $t=b\cap c\cap a\cap d$. $t$
is accepted by this automaton, but in order to reach that conclusion,
one must: 1) refactor $t$ as $\paren{a\cap b}\cap\paren{c\cap d}$; 2)
``match'' the opening transition $\tuple{0,[1,2]}$ with the closing
transition $\tuple{[11,12],13}$; 3) read the subterms
$\paren{a\cap b}$ and $\paren{c\cap d}$ respectively from state $1$ to
state $12$ and from $2$ to $11$.  This means that the run is built as
a nesting of runs (rather than a sequential process), and that it
needs to manipulate the term as a whole (rather than using a partial,
local view of it).

By contrast, to fire a transition in a Petri automaton, one simply
needs to see one vertex of the graph and the edges coming out of
it. Furthermore, the matching of transitions is somewhat automatic in
our model, and knowledge of it is not needed to compute a run. For
these reasons, our notion of language of a run could be defined using
standard Petri net notions (namely \emph{pomset-traces}, see
Section~\ref{sec:relat-with-stand}), whereas runs in a branching
automaton rely crucially on the term representation.

Another difference stems from the way automata are labelled. Branching
automata do not label their opening and closing transitions, making
these ``silent'' transitions. This highly complicates the task of
comparing automata using simulation-based methods. In our case
instead, every transition is labelled, and this helps us to obtain an
algorithm.

These differences make the task of converting from one model to the
other rather subtle. To translate a branching automaton into a Petri
automaton, one would need some kind of epsilon elimination
procedure. We believe such a procedure could be devised using boxes
from Section~\ref{sec:pa:exp} to keep track of the order in which
opening and closing transitions are combined. However we do not have a
precise formulation of this translation yet.

The other direction is slightly easier. Starting from a Petri
automaton, first modify the transitions as sketched below.
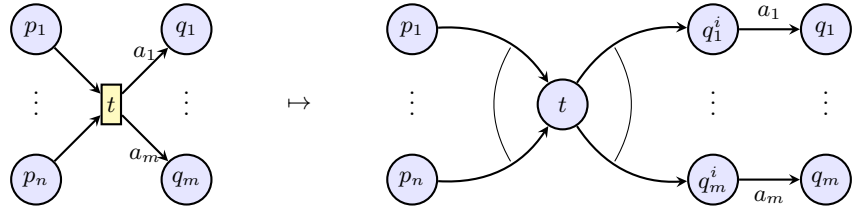
\begin{align*}
  \begin{tikzpicture}
    \etat[p_1](p1)(0,1);
    \etat[p_n](p2)(0,-1);
    \node () at ($(p1.south)!.5!(p2.north)$) {$\rvdots$};
    \etat[q_1](q1)(2,1);
    \etat[q_m](q2)(2,-1);
    \node (out1) at ($(q1.south)!.5!(q2.north)$) {$\rvdots$};
    \trans[t](t) (1,0);
    \edge (p1) (t);\edge(p2)(t);
    \edge(t)(q1)[a_1];\edge[below](t)(q2)[a_m];
    \etat[p_1](xp1)(5,1);
    \etat[p_n](xp2)(5,-1);
    \node (in2) at ($(xp1.south)!.5!(xp2.north)$) {$\rvdots$};
    \etat[q_1^i](q1i)(9,1);
    \etat[q_m^i](q2i)(9,-1);
    \node () at ($(q1i.south)!.5!(q2i.north)$) {$\rvdots$};
    \etat[q_1](q1o)(10.5,1);
    \etat[q_m](q2o)(10.5,-1);
    \node () at ($(q1o.south)!.5!(q2o.north)$) {$\rvdots$};
    \etat[t](tp) (7,0);
    \draw[arc,bend left] (xp1) to node[inner sep=0,pos=.6](x1) {} (tp);
    \draw[arc,bend right] (xp2) to node[inner sep=0,pos=.6](x2) {} (tp);
    \draw[bend right] (x1) to (x2);
    \draw[arc,bend left] (tp) to node[inner sep=0,pos=.4](y1) {} (q1i);
    \draw[arc,bend right] (tp) to node[inner sep=0,pos=.4](y2) {} (q2i);
    \draw[bend left] (y1) to (y2);
    \edge(q1i)(q1o)[a_1];\edge[below](q2i)(q2o)[a_m];
    \node () at ($(out1)!.5!(in2)$) {$\mapsto$};
  \end{tikzpicture}
\end{align*}
\noindent
This yields an automaton whose pomset-trace language (when seen as a
Petri net) is the language of the original Petri automaton. However,
its branching automaton language is not the same, as illustrated by
the following example:
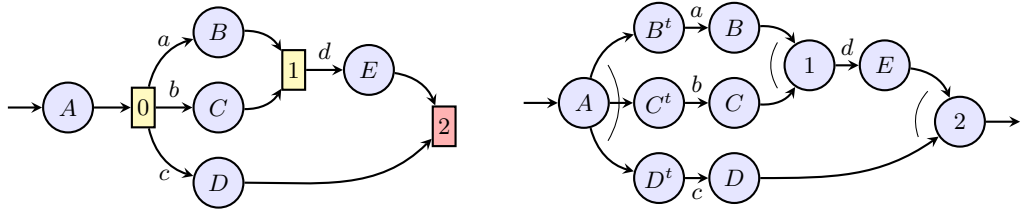
\begin{align*}
  \begin{tikzpicture}
  \etat[A](0)(0,0);\trans[0](t0)(1,0);\initst(0);
  \etat[B](1)(2,1);\etat[C](2)(2,0);\etat[D](3)(2,-1);
  \trans[1](t1)(3,.5);\etat[E](4)(4,.5);\transf[2](t2)(5,-.25);
  \edge(0)(t0);\edge(t0)(2)[b];\edge[bend left,above](t0)(1)[a];
  \edge[bend right,below](t0)(3)[c];
  \edge[bend left](1)(t1);\edge[bend right](2)(t1);\edge(t1)(4)[d];
  \edge[bend left](4)(t2);\edge[out=0,in=-135](3)(t2);
  \end{tikzpicture}
  &&
  \begin{tikzpicture}
  \etat[A](0)(0,0);\initst(0);
  \etat[B^t](1i)(1,1);\etat[C^t](2i)(1,0);\etat[D^t](3i)(1,-1);
  \etat[B](1)(2,1);\etat[C](2)(2,0);\etat[D](3)(2,-1);
  \etat[1](t1)(3,.5);\etat[E](4)(4,.5);\etat[2](t2)(5,-.25);\fnst(t2);
  \draw[arc] (0) to node[pos=.4] (x2) {} (2i);
  \draw[arc] (0) to[bend left] node[pos=.4] (x1) {} (1i);
  \draw[arc] (0) to[bend right] node[pos=.4] (x3) {} (3i);
  \draw[bend left] (x1) to (x3);
  \edge(2i)(2)[b];\edge(1i)(1)[a];\edge[below](3i)(3)[c];
  \draw[arc] (1) to[bend left] node[pos=.6] (y1) {} (t1);
  \draw[arc] (2) to[bend right] node[pos=.6] (y2) {} (t1);
  \draw[bend right] (y1) to (y2);
  \edge(t1)(4)[d];
  \draw[arc] (4) to[bend left] node[pos=.6] (z1) {} (t2);
  \draw[arc] (3) to[out=0,in=-145] node[pos=.9] (z2) {} (t2);
  \draw[bend right] (z1) to (z2);
  \end{tikzpicture}
\end{align*}
The pomset-trace language of the branching automaton corresponds to
the singleton set~$\set{\paren{\paren{a\cap b}\cdot d}\cap c}$, but
its language is empty, as no factorisation of that term has three
parallel subterms. To solve this problem, one needs to saturate the
automaton by splitting opening and closing transitions to allow for
every factorisation. Formally, it means that for an opening transition
$\tuple{p,[q_1,\dots,q_n]}$, for every tree\footnote{We consider here
  trees where no vertex has out-degree one.} with $n$ leaves labelled
with $[q_1,\dots,q_n]$ there should be a sequence of opening
transitions of that shape.

In the case of the above example, one would add three states
$BC^t,BD^t$ and $CD^t$, and six transitions:
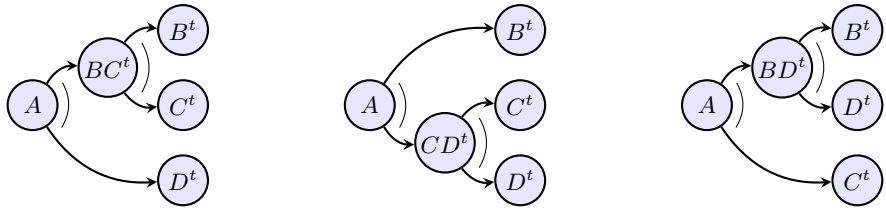
\begin{align*}
  \begin{tikzpicture}
  \etat[A](0)(-1,0);
  \etat[B^t](1i)(1,1);\etat[C^t](2i)(1,0);\etat[D^t](3i)(1,-1);
  \etat[BC^t](g)(0,.5);
  \draw[arc] (0) to[bend left] node[pos=.4] (x1) {} (g);
  \draw[arc] (0) to[bend right] node[pos=.1] (x2) {} (3i);
  \draw[bend left] (x1) to (x2);
  \draw[arc] (g) to[bend left] node[pos=.4] (y1) {} (1i);
  \draw[arc] (g) to[bend right] node[pos=.4] (y2) {} (2i);
  \draw[bend left] (y1) to (y2);
  \end{tikzpicture}&&
  \begin{tikzpicture}
  \etat[A](0)(-1,0);
  \etat[B^t](1i)(1,1);\etat[C^t](2i)(1,0);\etat[D^t](3i)(1,-1);
  \etat[CD^t](g)(0,-.5);
  \draw[arc] (0) to[bend right] node[pos=.4] (x1) {} (g);
  \draw[arc] (0) to[bend left] node[pos=.1] (x2) {} (1i);
  \draw[bend left] (x2) to (x1);
  \draw[arc] (g) to[bend left] node[pos=.4] (y1) {} (2i);
  \draw[arc] (g) to[bend right] node[pos=.4] (y2) {} (3i);
  \draw[bend left] (y1) to (y2);
  \end{tikzpicture}&&
  \begin{tikzpicture}
  \etat[A](0)(-1,0);
  \etat[B^t](1i)(1,1);\etat[D^t](2i)(1,0);\etat[C^t](3i)(1,-1);
  \etat[BD^t](g)(0,.5);
  \draw[arc] (0) to[bend left] node[pos=.4] (x1) {} (g);
  \draw[arc] (0) to[bend right] node[pos=.1] (x2) {} (3i);
  \draw[bend left] (x1) to (x2);
  \draw[arc] (g) to[bend left] node[pos=.4] (y1) {} (1i);
  \draw[arc] (g) to[bend right] node[pos=.4] (y2) {} (2i);
  \draw[bend left] (y1) to (y2);
  \end{tikzpicture}
\end{align*}
After this transformation, the pomset-trace and the branching automata
languages coincide: we obtained a branching automaton recognising the
same language as the original Petri automaton.

This translation gives rise to a completely different proof of
Theorem~\ref{thm:kl2:sp}: first translate the Petri automaton into an
equivalent branching automaton, then use Theorem~\ref{thm:ktba} to
obtain an expression describing its language. Conversely, one could
obtain Theorem~\ref{thm:ktba} through a translation from branching
automata to Petri automata and Theorem~\ref{thm:kl2:sp}.

% \section*{Acknowledgments}
% \noindent The authors wish to acknowledge fruitful discussions with A and B.

\bibliographystyle{abbrvurl}
\bibliography{long,main,pous}

\begin{thebibliography}{10}

\bibitem{AMN11}
H.~Andr{\'e}ka, S.~Mikul{\'a}s, and I.~N{\'e}meti.
\newblock \href {http://dx.doi.org/10.1016/j.tcs.2011.09.024} {The equational
  theory of {K}leene lattices}.
\newblock {\em Theoretical Computer Science}, 412(52):7099--7108, 2011.

\bibitem{AB95}
H.~Andréka and D.~Bredikhin.
\newblock \href {http://dx.doi.org/10.1007/BF01225472} {The equational theory
  of union-free algebras of relations}.
\newblock {\em Algebra Universalis}, 33(4):516--532, 1995.

\bibitem{BES95}
S.~L. Bloom, Z.~{\'E}sik, and G.~Stefanescu.
\newblock \href {http://dx.doi.org/10.1007/BF01190768} {Notes on equational
  theories of relations}.
\newblock {\em Algebra Universalis}, 33(1):98--126, 1995.

\bibitem{Boffa95}
M.~Boffa.
\newblock \href {http://archive.numdam.org/.../ITA_1995__29_6_515_0.pdf} {Une
  condition impliquant toutes les identités rationnelles}.
\newblock {\em Informatique Théorique et Applications}, 29(6):515--518, 1995.

\bibitem{BojanczykP16}
M.~Boja{\'{n}}czyk and M.~Pilipczuk.
\newblock \href {http://dx.doi.org/10.1145/2933575.2934508} {Definability
  equals recognizability for graphs of bounded treewidth}.
\newblock In {\em Proc.}\ {\em LICS}, pages 407--416. ACM, 2016.

\bibitem{bossut1995kleene}
F.~Bossut, M.~Dauchet, and B.~Warin.
\newblock \href {http://dx.doi.org/10.1006/inco.1995.1043} {A {K}leene theorem
  for a class of planar acyclic graphs}.
\newblock {\em Information and Computation}, 117(2):251--265, 1995.

\bibitem{rklm:web}
P.~Brunet.
\newblock {RKLM} software, 2014.
\newblock \url{http://paul.brunet-zamansky.fr/rklm.php}.

\bibitem{bp:ramics14:kac}
P.~Brunet and D.~Pous.
\newblock \href {http://dx.doi.org/10.1007/978-3-319-06251-8_7} {{K}leene
  algebra with converse}.
\newblock In {\em Proc.}\ {\em RAMiCS}, volume 8428 of {\em LNCS}, pages
  101--118. Springer, 2014.

\bibitem{bp:jlamp15:kac}
P.~Brunet and D.~Pous.
\newblock \href {http://dx.doi.org/10.1016/j.jlamp.2015.07.005} {Algorithms for
  {K}leene algebra with converse}.
\newblock {\em Journal of Logical and Algebraic Methods in Programming},
  85(4):574--594, 2015.

\bibitem{bp:lics15:paka}
P.~Brunet and D.~Pous.
\newblock \href {http://dx.doi.org/10.1109/LICS.2015.17} {Petri automata for
  {K}leene allegories}.
\newblock In {\em Proc.}\ {\em LICS}, pages 68--79. ACM, 2015.

\bibitem{Conway71}
J.~H. Conway.
\newblock {\em Regular algebra and finite machines}.
\newblock Chapman and Hall, 1971.

\bibitem{courcelle:hal-00646514}
B.~Courcelle and J.~Engelfriet.
\newblock \href {https://hal.archives-ouvertes.fr/hal-00646514} {{\em {Graph
  structure and monadic second-order logic. A language-theoretic approach.}}}
\newblock Encyclopedia of Mathematics and its applications, Vol. 138. Cambridge
  University Press, June 2012.
\newblock Collection Encyclopedia of Mathematics and Applications, Vol. 138.

\bibitem{EB95}
Z.~{\'E}sik and L.~Bern{\'a}tsky.
\newblock \href {http://dx.doi.org/10.1016/0304-3975(94)00041-G} {Equational
  properties of {K}leene algebras of relations with conversion}.
\newblock {\em Theoretical Computer Science}, 137(2):237--251, 1995.

\bibitem{FS90}
P.~Freyd and A.~Scedrov.
\newblock {\em Categories, Allegories}.
\newblock North Holland, 1990.

\bibitem{Furer1980}
M.~F{\"u}rer.
\newblock \href {http://dx.doi.org/10.1007/3-540-10003-2_74} {The complexity of
  the inequivalence problem for regular expressions with intersection}.
\newblock In {\em Proc.}\ {\em ICALP}, pages 234--245. Springer Verlag, 1980.

\bibitem{goltz1983}
U.~Goltz and W.~Reisig.
\newblock \href {http://dx.doi.org/10.1016/S0019-9958(83)80040-0} {The
  non-sequential behaviour of {P}etri nets}.
\newblock {\em Information and Control}, 57(2):125--147, 1983.

\bibitem{Jategaonkar96}
L.~Jategaonkar and A.~R. Meyer.
\newblock \href {http://dx.doi.org/10.1016/0304-3975(95)00132-8} {Deciding true
  concurrency equivalences on safe, finite nets}.
\newblock {\em Theoretical Computer Science}, 154(1):107--143, 1996.

\bibitem{Kleene}
S.~C. Kleene.
\newblock \href
  {http://www.rand.org/content/dam/rand/pubs/research_memoranda/2008/RM704.pdf}
  {{\em Representation of Events in Nerve Nets and Finite Automata}}.
\newblock Memorandum. Rand Corporation, 1951.

\bibitem{Kozen91}
D.~Kozen.
\newblock \href {http://dx.doi.org/10.1109/LICS.1991.151646} {A completeness
  theorem for {K}leene algebras and the algebra of regular events}.
\newblock In {\em Proc.}\ {\em LICS}, pages 214--225. IEEE, 1991.

\bibitem{kozen1998typed}
D.~Kozen.
\newblock \href {http://www.cs.cornell.edu/~kozen/papers/typed.pdf} {Typed
  {K}leene algebra}.
\newblock Technical Report TR98-1669, CS Dpt., Cornell University, 1998.

\bibitem{Krob90}
D.~Krob.
\newblock \href {http://dx.doi.org/10.1007/BFb0032022} {{A Complete System of
  {B}-Rational Identities}}.
\newblock In {\em Proc.}\ {\em ICALP}, volume 443 of {\em Lecture Notes in
  Computer Science}, pages 60--73. Springer Verlag, 1990.

\bibitem{Lodaya1998}
K.~Lodaya and P.~Weil.
\newblock \href {http://dx.doi.org/10.1007/BFb0028590} {{\em Series-parallel
  posets: Algebra, automata and languages}}, pages 555--565.
\newblock Springer Verlag, 1998.

\bibitem{LODAYA2000347}
K.~Lodaya and P.~Weil.
\newblock \href {http://dx.doi.org/10.1016/S0304-3975(00)00031-1}
  {Series–parallel languages and the bounded-width property}.
\newblock {\em Theoretical Computer Science}, 237(1):347--380, 2000.

\bibitem{LODAYA2001269}
K.~Lodaya and P.~Weil.
\newblock \href {http://dx.doi.org/10.1006/inco.2001.3077} {Rationality in
  algebras with a series operation}.
\newblock {\em Information and Computation}, 171(2):269--293, 2001.

\bibitem{MS72}
A.~Meyer and L.~J. Stockmeyer.
\newblock \href {http://dx.doi.org/10.1109/SWAT.1972.29} {The equivalence
  problem for regular expressions with squaring requires exponential space}.
\newblock In {\em Proc.}\ {\em SWAT}, pages 125--129. IEEE, 1972.

\bibitem{MS73}
A.~Meyer and L.~J. Stockmeyer.
\newblock \href {http://dx.doi.org/10.1145/800125.804029} {Word problems
  requiring exponential time}.
\newblock In {\em Proc.}\ {\em STOC}, pages 1--9. ACM, 1973.

\bibitem{moller1999typed}
B.~M{\"o}ller.
\newblock Typed {K}leene algebras.
\newblock In {\em Proc.}\ {\em MPC}, Lecture Notes in Computer Science.
  Citeseer, 1999.

\bibitem{Murata}
T.~Murata.
\newblock \href {http://dx.doi.org/10.1109/5.24143} {{P}etri nets: Properties,
  analysis and applications}.
\newblock {\em Proceedings of the IEEE}, 77(4):541--580, Apr 1989.

\bibitem{Petri62}
C.~A. Petri.
\newblock Fundamentals of a theory of asynchronous information flow.
\newblock In {\em Proc.}\ {\em {IFIP} Congress}, pages 386--390, 1962.

\bibitem{Petri}
C.~A. Petri.
\newblock {\em Kommunikation mit Automaten}.
\newblock PhD thesis, Darmstadt Univ. of Tech., 1962.

\bibitem{redko64}
V.~N. Redko.
\newblock On defining relations for the algebra of regular events.
\newblock {\em Ukrainskii Matematicheskii Zhurnal}, 16:120--126, 1964.

\bibitem{savitch70}
W.~J. Savitch.
\newblock Relationships between nondeterministic and deterministic tape
  complexities.
\newblock {\em Journal of computer and system sciences}, 4(2):177--192, 1970.

\bibitem{thompson68}
K.~Thompson.
\newblock \href
  {http://www.fing.edu.uy/inco/cursos/intropln/material/p419-thompson.pdf}
  {Regular expression search algorithm}.
\newblock {\em Communications of the ACM}, 11:419--422, 1968.

\bibitem{Valdes79}
J.~Valdes, R.~E. Tarjan, and E.~L. Lawler.
\newblock \href {http://dx.doi.org/10.1145/800135.804393} {The recognition of
  series parallel digraphs}.
\newblock In {\em Proc.}\ {\em STOC}, pages 1--12. ACM, 1979.

\end{thebibliography}

% \appendix

\end{document}